\documentclass[12pt]{article}
\usepackage[vmargin=1.18in,hmargin=1.1in]{geometry}				
\usepackage{amssymb,amsfonts,amsthm,mathtools,bm,bbm,mathrsfs} 	
\usepackage[dvipsnames]{xcolor} 								
\usepackage[colorlinks, linkcolor = blue, citecolor = blue, urlcolor = blue]{hyperref} 								
\usepackage{graphicx} 											
\usepackage{caption} 											
\usepackage{verbatim}
\usepackage[misc]{ifsym} 										
\usepackage[ruled, linesnumbered]{algorithm2e} 					
\usepackage{layouts} 
\usepackage[uniquename=false, 
			uniquelist = false, 
			url = false, 
			doi = false, 
			isbn = false, 
			natbib = true, 
			backend = bibtex, 
			style = authoryear, 
			maxbibnames = 5]{biblatex} 							
\usepackage{booktabs} 
\usepackage{multirow} 
\usepackage[noend]{algpseudocode}
\usepackage{subcaption}
\usepackage{multirow}
\usepackage{setspace}

\makeatother

\renewcommand{\P}{\mathbb{P}}
\newcommand{\E}{\mathbb{E}}

\newcommand{\R}{\mathbb{R}}
\newcommand{\N}{\mathbb{N}}
\newcommand{\indicator}{\mathbbm 1}						
\newcommand{\convd}{\overset{d}{\rightarrow}}
\newcommand{\convp}{\overset{p}{\rightarrow}}
\newcommand{\convdp}{\overset{d,p}{\longrightarrow}}
\newcommand{\convpp}{\overset{p,p}{\longrightarrow}}
\newcommand{\ep}{\epsilon}
\newcommand{\indep}{\perp\!\!\!\perp}
\newcommand{\dep}{\not\!\perp\!\!\!\perp}

\newtheorem{assumption}{Assumption}
\newtheorem{theorem}{Theorem}

\newtheorem{lemma}{Lemma}
\newtheorem{proposition}{Proposition}
\newtheorem{definition}{Definition}
\theoremstyle{definition}

\title{The permuted score test for robust \\ differential expression analysis}
\author{Timothy Barry$^1$, Ziang Niu$^2$, Eugene Katsevich$^2$, Xihong Lin$^1$ \\
$^1$Department of Biostatistics, Harvard University \\
$^2$Department of Statistics, University of Pennsylvania
}
\begin{document}
\maketitle

\begin{abstract}
Negative binomial (NB) regression is a popular method for identifying differentially expressed genes in genomics data, such as bulk and single-cell RNA sequencing data. However, NB regression makes stringent parametric and asymptotic assumptions, which can fail to hold in practice, leading to excess false positive and false negative results. We propose \textit{the permuted score test}, a new strategy for robust regression based on permuting score test statistics. The permuted score test provably controls type-I error across a much broader range of settings than standard NB regression while nevertheless approximately matching standard NB regression with respect to power (when the assumptions of standard NB regression obtain) and computational efficiency. We accelerate the permuted score test by leveraging emerging techniques for sequential Monte-Carlo testing and novel algorithms for efficiently computing GLM score tests. We apply the permuted score test to real and simulated RNA sequencing data, finding that it substantially improves upon the error control of existing NB regression implementations, including \texttt{DESeq2}. The permuted score test could enhance the reliability of differential expression analysis across diverse biological contexts.
\end{abstract}

\textbf{Keywords}: adaptive inference, differential expression, negative binomial regression, permutation test, robust inference
\\ \\ 
Differential expression (DE) analysis entails assessing whether a gene (or other feature of interest, such as a protein) exhibits variable activity across conditions. Identifying differentially expressed features is a fundamental statistical task in the analysis of many kinds of biological data, including single-cell RNA-seq \citep{Berge2020}, bulk RNA-seq \citep{Robinson2009}, ChIP-seq \citep{Eder2022}, spatial transcriptomics \citep{Cable2022}, microbiome \citep{Lutz2022}, and CRISPR screen \citep{Schmidt2022} data. The most widely used method for differential expression analysis is negative binomial (NB) regression; popular implementations of NB regression include \texttt{DESeq2} (which has accrued over 70,000 citations as of 2024; \cite{Love2014}) and \texttt{MASS} \citep{Ripley2020}. NB regression provides a simple and principled model for the count-based response, supports complex experimental designs involving multiple covariates, and offers competitive computational performance.

However, NB regression makes stringent parametric and asymptotic assumptions \citep{Nelder1972}, which can fail to hold in practice, leading to excess false positive and false negative errors. Parametric assumptions may be violated, for instance, if the NB dispersion parameter is incorrectly specified, if the response is zero-inflated (i.e., contains more zeroes than expected), or if key covariates or interaction terms are omitted from the model \citep{Li2022}. On the other hand, asymptotic assumptions may not hold if there are few samples or if the count-based response contains many small values \citep{Niu2024}. If either parametric or asymptotic assumptions break down, NB regression can produce highly misleading results. Several recent studies have demonstrated that popular NB regression implementations --- including \texttt{DESeq2} and \texttt{MASS} --- failed to control type-I error on negative control single-cell and bulk RNA-seq data, i.e.\ real data devoid of biological signal \citep{Barry2021,Li2022,Barry2024}. This loss of error control resulted from a violation of the NB regression modeling assumptions.

Our central claim is that NB regression can be made substantially more robust at minimal cost to power or computational efficiency by combining NB regression with emerging techniques for robust and adaptive inference. We propose the \textit{permuted score test}, a robust approach to NB regression based on the simple (although, to the best of our knowledge, novel) strategy of permuting NB GLM score test statistics. Extending theoretical tools developed by \citet{Diciccio2017} and \citet{Janssen1997}, we prove that the permuted score test controls type-I error across a much broader range of settings than standard NB regression. Specifically, if the ``treatment'' (i.e., the covariate of interest) is unconfounded by the nuisance covariates, then the permuted score test controls type-I error in finite samples and under arbitrary model misspecification. On the other hand, if the treatment is confounded, then the permuted score test controls type-I error in sufficiently large samples, even in situations where the dispersion parameter is incorrectly specified or inaccurately estimated. We refer to this key robustness phenomenon as ``confounder adjustment via marginal permutations,'' or CAMP.

To ensure computational efficiency, we implement several accelerations to reduce the running time of the permuted score test so that it is within about 1--3 times that of standard NB regression. First, we leverage strategies for adaptive permutation testing of multiple hypotheses via anytime-valid inference (\cite{Ramdas2023,Fischer2024a,Fischer2024b}), substantially reducing the number of permutations that must be computed across hypotheses. Second, building on the recent work of \citet{Barry2024}, we introduce a novel algorithm for efficiently computing GLM score tests, which is tens to hundreds of times faster than its classical counterpart. Finally, we implement the permuted score test in C++ for maximum speed. Extensive simulation studies demonstrate that the permuted score test maintains error control across a much broader range of settings than standard NB regression while at the same time approximately matching standard NB regression with respect to power (when the assumptions of standard NB regression are satisfied) and computational efficiency.

This paper is organized as follows. Section \ref{sec:background} reviews NB regression and introduces the genomics application that will serve as our working example throughout the paper, namely arrayed CRISPR screens with bulk RNA-seq readout \citep{Schmidt2022,Chardon2024}. Section \ref{sec:method_and_theory} studies the theoretical properties of the permuted score test and establishes the key robustness phenomenon of CAMP. Section \ref{sec:computational_accelerations} introduces the optimizations that we leverage to accelerate the method. Section \ref{sec:sim_studies} presents the results of an extensive array of simulation studies comparing the permuted score test to standard NB regression and other competitors. Finally, Section \ref{sec:real_data_analysis} discusses the analysis of real arrayed CRISPR screen dataset.

\section{Background}\label{sec:background}

Although we expect the permuted score test to apply broadly, we develop it in the context of a specific genomics application --- namely, arrayed CRISPR screens with bulk RNA-seq readout --- for concreteness. This section introduces this application and provides a brief overview of NB regression and its limitations.

\subsection{Arrayed CRISPR screens with bulk RNA-seq readout} An arrayed CRISPR screen with bulk RNA-seq readout (``arrayed CRISPR screen'' for short) is an experiment designed to identify the subset of genes whose expression changes in response to a CRISPR perturbation. To carry out an arrayed CRISPR screen experiment, we recruit $s$ donors and collect $t$ replicate tissue samples from each donor. (We collect a total of $n = st$ tissue samples across donors). Within each donor, we randomly assign half the tissue samples to the treatment condition and half to the control condition. We apply a targeting CRISPR perturbation to the treatment samples and a non-targeting CRISPR perturbation to the control samples. The non-targeting CRISPR perturbation acts as a placebo, exerting no impact on gene expression levels. We use bulk RNA-seq to profile the expression level of each of $m \approx 30,000$ genes contained within each tissue sample. Finally, we test for association between the expression level of each gene and treatment status to identify differentially expressed genes.
\subsection{Problem setup and NB regression} We formalize the differential expression testing problem as follows. Consider a given gene. We observe i.i.d.\ tuples $\{(X_i, Y_i, Z_i)\}_{i=1}^n,$ where $X_i \in \{0,1\}$ is the ``treatment'' covariate (i.e., the covariate of interest), $Y_i \in \{0, 1, 2, \dots\}$ is the expression level of the gene (quantified as the number of sequenced RNA transcripts mapping to the gene), and $Z_i \in \R^p$ is a low-dimensional vector of nuisance covariates expected to explain variability in gene expression but not of direct interest. In the context of the arrayed CRISPR screen application introduced above, each experimental unit $i$ represents a tissue sample; $X_i$ indicates whether the given sample was subjected to a targeting ($X_i = 1$) or non-targeting ($X_i = 0$) CRISPR perturbation; and $Z_i$ records sample-wise covariates, such as donor source and library size (i.e., the total number of transcripts sequenced across all genes in the sample). This setup extends beyond arrayed CRISPR screens to a wide range of other genomics applications. For example, in the context of single-cell RNA-seq, each experimental unit $i$ is an individual cell; $X_i$ may indicate whether the given cell belongs to a specified cell type ($X_i = 1$) or a different cell type ($X_i = 0$); and $Z_i$ captures cell-specific covariates, such as sequencing batch. Although we refer to $X_i$ as the ``treatment'' throughout, we allow the data to come from either a randomized controlled trial (RCT) or an observational study.

Our analysis objective is to test whether the gene is differentially expressed across treatment and control conditions. Formally, we seek to test the null hypothesis that the treatment $X_i$ is independent of gene expression $Y_i$ conditional on the nuisance covariates $Z_i$, i.e. 
\begin{equation}\label{ci_null_hyp}
H_0: X_i\indep Y_i \mid Z_i.
\end{equation}
The conditional independence null hypothesis (\ref{ci_null_hyp}) asserts that the treatment provides no information about gene expression above and beyond the information provided by the nuisance covariates. This model-free framing of the null hypothesis enables us to meaningfully define ``differential expression,'' even when the NB model is misspecified. If the treatment and nuisance covariates are independent (as they are in an RCT), the conditional independence null hypothesis (\ref{ci_null_hyp}) is equivalent to the marginal independence null hypothesis, $X_i \indep Y_i$.

The standard approach to testing the null hypothesis (\ref{ci_null_hyp}) is to perform an NB regression analysis. The NB regression model is as follows:
\begin{equation}\label{eqn:nb_glm}
Y_i \sim \textrm{NB}_{\phi}(\mu_i); \quad \log(\mu_i) = \gamma X_i + \beta^T Z_i; \quad (X_i, Z_i) \sim \mathcal{P}.
\end{equation}
Here, $\textrm{NB}_\phi(\mu_i)$ is a negative binomial distribution with mean $\mu_i$ and dispersion parameter $\phi$; $\gamma \in \R$ is the treatment coefficient; $\beta \in \R^p$ is the vector of nuisance coefficients; and $\mathcal{P}$ is an (unknown) probability distribution on $\R^{p+1}$. The variance of $Y_i$ (conditional on $X_i$ and $Z_i$) is $\mu_i + \phi\mu_i^2.$ The dispersion parameter captures overdispersion commonly present within genomics data; it either can be specified by the user or estimated from the data. The NB model (\ref{eqn:nb_glm}) reduces to a Poisson model when the dispersion parameter is zero.

To conduct an NB regression analysis, we regress $Y$ onto $X$ and $Z$ using maximum likelihood estimation, yielding estimates $\hat{\gamma}$ and $\hat{\beta}$ of $\gamma$ and $\beta$, respectively. Letting $\textrm{se}(\hat{\gamma})$ denote the standard error of $\hat{\gamma}$, we can compute a (e.g., left-tailed) Wald $p$-value for a test of the null hypothesis $\gamma = 0$ via $p = \Phi(\hat{\gamma}/\textrm{se}(\hat{\gamma})),$ where $\Phi$ is the standard Gaussian CDF. If the NB model (\ref{eqn:nb_glm}) is correctly specified, the conditional independence null hypothesis (\ref{ci_null_hyp}) holds if and only if $\gamma = 0$ \citep{Candes2018}. In this sense the Wald test can be interpreted as a test of conditional independence under correct model specification.

Typically, we are interested in testing not just a single hypothesis but rather a large number (e.g., hundreds to tens of thousands) of hypotheses, which generally correspond to different genes. Let $\mathcal{T}_1, \dots, \mathcal{T}_m$ denote $m$ sets of tuples, where the $j$th set of tuples $\mathcal{T}_j$ contains i.i.d.\ data $\{(X_i^j, Y_i^j, Z_i^j)_{i=1}^n\}$. We wish to test the conditional independence null hypothesis $\mathcal{H}^j_0 : (Y_i^j \indep X_i^j)\mid Z_i^j$ for each set of tuples $\mathcal{T}_j$, yielding a set $\mathcal{D} \subset \{\mathcal{H}^1_0, \dots, \mathcal{H}^m_0\}$ of putative discoveries (i.e., the ``discovery set''). We require that the discovery set $\mathcal{D}$ control the false discover rate (FDR) at a user-specified level $\alpha$. The standard approach to producing an FDR-controlling discovery set is to subject the NB GLM $p$-values (computed as described above) to a Benjamini-Hochberg (BH) correction \citep{Benjamini1995}.

\subsection{Shortcomings of NB regression}\label{sec:result_preview}
To assess the error control of NB regression on real data, we applied NB regression (implemented as \texttt{MASS} and \texttt{DESeq2}) to an arrayed CRISPR screen dataset \citep{Schmidt2022}. Following \citet{Li2022}, we randomly permuted the expression vector of each gene (independently of the other genes) to generate a negative control dataset for which the null hypothesis (\ref{ci_null_hyp}) held for each gene. Both \texttt{MASS} and \texttt{DESeq2} failed to control type-I error on the negative control data, as exhibited by the inflated $p$-values outputted by these methods (Figure \ref{fig:arrayed_screen_preview}a). We repeated this experiment across 1,000 Monte Carlo replicates and found that \texttt{MASS} and \texttt{DESeq2} made an average of 1,138 and 83 false discoveries (after a BH correction at level $0.1$) across replicates (Figure \ref{fig:arrayed_screen_preview}b). These results underscore the fragility of NB regression to modeling assumption violations on real data.

\begin{figure}
    \centering
    \includegraphics[width=0.75\linewidth]{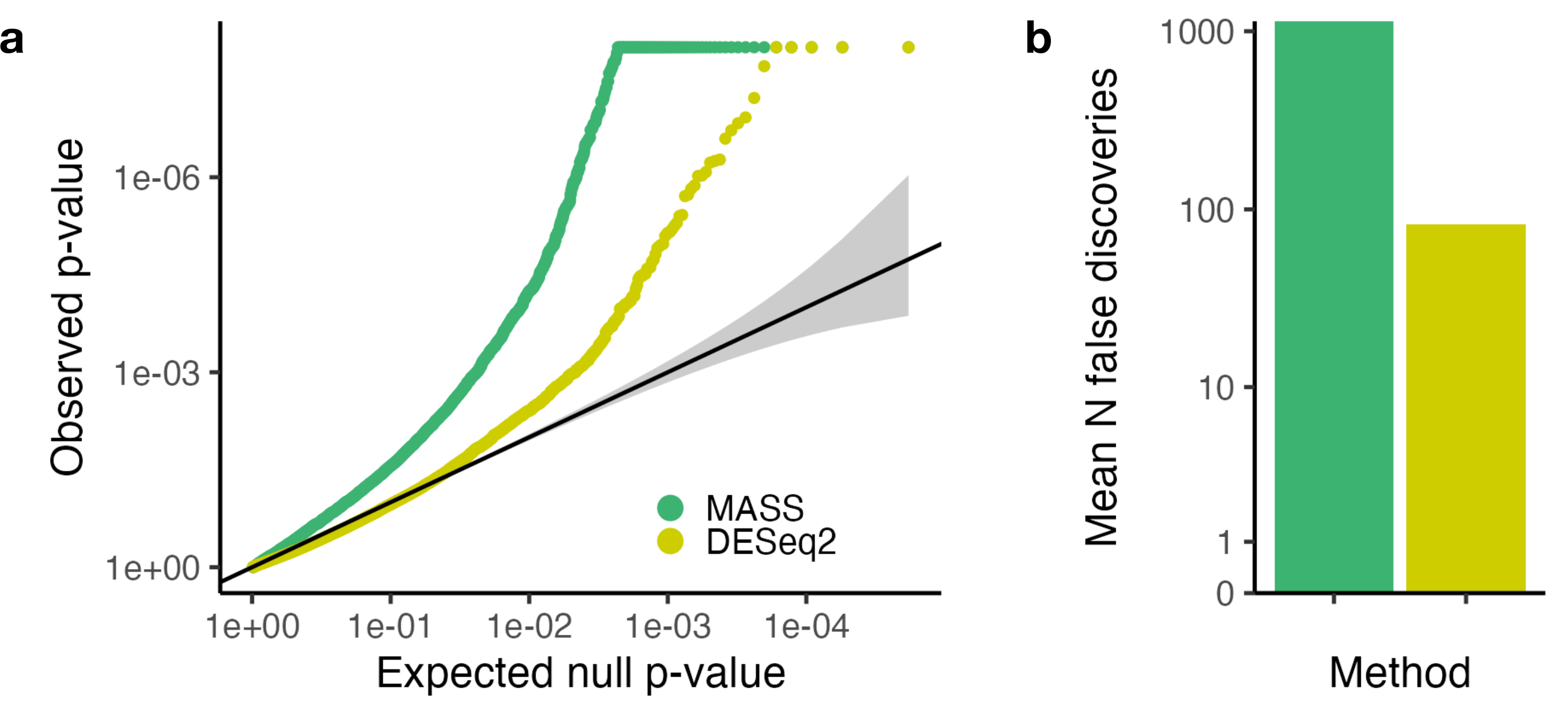}
    \caption{Application of NB regression (implemented as \texttt{MASS} and \texttt{DESeq2}) to negative control arrayed CRISPR screen data. \textbf{a}, QQ plot of $p$-values on one replicate of negative control data; \textbf{b}, average number of false discoveries across $B=1,000$ Monte Carlo replicates.}\label{fig:arrayed_screen_preview}
\end{figure}

\section{Permuted score test: method and theory}\label{sec:method_and_theory}

We propose the permuted score test, a permutation test based on an NB GLM score test statistic (Algorithm \ref{algo:permuting_score_stats}). Consider a set of i.i.d.\ tuples $\{(X_i, Y_i, Z_i)\}_{i=1}^n$. First, we regress the response vector $Y$ onto the nuisance covariate matrix $Z$ by fitting an NB GLM, producing a fit of the NB model (\ref{eqn:nb_glm}) under the null hypothesis. Next, we test the treatment vector $X$ for inclusion in the fitted model via a score test, yielding an ``original'' $z$-score $z_\textrm{orig}$. We randomly permute the treatment vector $B$ times, producing permuted treatment vectors $\tilde{X}_1, \dots, \tilde{X}_B$.  We then test each permuted treatment vector for inclusion in the fitted model, yielding ``null'' $z$-scores $z_1, \dots, z_B$. Finally, we compute a $p$-value by comparing the original $z$-score to the null $z$-scores. Importantly, the NB GLM is fitted only once and shared across the observed and permuted treatment vectors.

\begin{algorithm}
\caption{The permuted score test. Note that \texttt{fit} is shared across steps 2 and 4.}\label{algo:permuting_score_stats}
Regress $Y$ onto $Z$, yielding a fitted NB GLM, $\texttt{fit}.$

Compute $z_\textrm{orig} = T\left(\texttt{fit}, X\right)$, where $T$ is the score test statistic for adding $X$ as a covariate to the fitted model \texttt{fit}.

Randomly permute $X$, yielding permuted vectors $\tilde{X}_1, \dots, \tilde{X}_B.$

For each $b \in \{1, \dots, B\}$, compute $z_b = T(\texttt{fit}, \tilde{X}_b)$.

Compute a (e.g., left-tailed) $p$-value:
\begin{equation}\label{eqn:perm_p_value}
p = \frac{1}{B+1}\left(1 + \sum_{b=1}^B \mathbb{I}( z_b \leq z_\textrm{orig})\right).
\end{equation}
\end{algorithm}

The remainder of this paper aims to show that the permuted score test (Algorithm \ref{algo:permuting_score_stats}), although simple, enjoys highly favorable statistical and computational properties. This section proves that the permuted score test provides a valid test of conditional independence under a much broader range of settings standard NB regression. We begin by stating and proving Propositions \ref{thm:asy_distribution_misspecification} and \ref{thm:asy_permutation_distribution}, which establish the asymptotic \textit{sampling} and \textit{permutation} distributions, respectively, of the NB GLM score test statistic. We treat the dispersion parameter as specified by the user (as opposed to estimated from the data), as this permits a more precise theoretical analysis.

\subsection{Asymptotic results}\label{sec:asymptotic_results}

Denote the user-specified dispersion parameter by $\bar \phi$. Recall that the true dispersion parameter (i.e., the dispersion parameter used to generate the data (\ref{eqn:nb_glm})) is $\phi$. We do not assume that the dispersion parameter is correctly specified, i.e.\ we allow for $\bar \phi \neq \phi$. The score vector $U_{\bar \phi}(\beta)$ for the nuisance coefficients $\beta$ under the null hypothesis of $\gamma = 0$ is
\begin{align}\label{eqn:misspec_score_equation}
	U_{\bar \phi}(\beta)=\frac{1}{n}\sum_{i=1}^n \frac{Z_i(Y_i-\mu_i)}{1+\bar\phi \mu_i}; \quad \mu_i\equiv \exp(\beta^\top Z_i).
\end{align}
The score vector (\ref{eqn:misspec_score_equation}) is misspecified in the sense that $\bar \phi$ may not equal $\phi$. The estimator $\hat{\beta}_n$ is defined as a vector that satisfies $U_{\bar \phi}(\hat \beta_n) = 0.$ We can obtain $\hat{\beta}_n$ via an NB regression of $Y$ onto $Z$ using dispersion parameter $\bar \phi$. The score test statistic $T_n(X,Y,Z)$ for adding $X$ as a covariate to the fitted model (i.e., for testing $\gamma = 0$) is given by
\begin{align}\label{eq:score_statistic}
	T_n(X,Y,Z)\equiv \frac{X^\top  \hat r}{\sqrt{X^\top \hat W X- X^\top \hat WZ(Z^\top \hat WZ)^{-1}Z^\top \hat WX}},
\end{align}
where $\hat{r} \equiv [\hat{r}_1, \dots, \hat{r}_n]^T$, $\hat r_{i}\equiv (Y_{i}-\hat \mu_{i})/(1+ \bar \phi \hat \mu_{i})$, $\hat \mu_{i}\equiv \exp(\hat \beta_n^\top Z_{i}),$ $\hat W\equiv \mathrm{diag}\{\hat{W}_i \}$, and $\hat W_{i} \equiv \hat \mu_{i}/(1+ \bar\phi \hat \mu_{i}).$ We additionally define the random variables $W_i$ and $S_i$ by $W_i \equiv \mu_i /(1 + \bar \phi \mu_i)$ and $S_i \equiv \mu_i(1+\phi\mu_i)/(1+ \bar\phi\mu_i)^2$. Finally, let $R_i$ denote the expected residual from a least squares regression of $X$ onto $Z$ using weights $W_1,\dots, W_n$, i.e. $$R_i \equiv \sqrt{W_i} X_i-\E[X_iZ_i^\top W_i](\E[Z_iZ_i^\top W_i])^{-1}\sqrt{W_i}Z_i.$$

Propositions \ref{thm:asy_distribution_misspecification} and \ref{thm:asy_permutation_distribution} assume the following regularity conditions.
\begin{assumption}[Strong consistency of $\hat \beta_n$]\label{assu:consistency_beta}
	We assume $\|\hat\beta_n-\beta\| \rightarrow0$, almost surely.
\end{assumption}
\begin{assumption}[Bounded moments of treatment variable $X_{i}$]\label{assu:moment_treatment} We assume $\E[X_{i}]=0$ and $\E[X_{i}^2]=1$, without loss of generality. Furthermore, we assume $\E[X_i^{4}]<\infty$.
\end{assumption}
\begin{assumption}[Boundedness of covariate $Z_{i}$]\label{assu:boundedness_covariate}
	We assume there exists a constant $C > 0$ such that, for all $i \in \mathbb{N}$, $\|Z_{i}\| \leq C$ almost surely.
\end{assumption}
\begin{assumption}[Presence of an intercept term]\label{assu:intercept}
	We assume that $Z_i$ contains an entry of ``1'' corresponding to the intercept.
\end{assumption}
\begin{assumption}[Invertibility of covariance matrix]\label{assu:lower_bound_eigenvalue}
	We assume the population covariance matrix $\E[Z_iZ_i^\top]$ is invertible.
\end{assumption}

Assumption \ref{assu:consistency_beta} is sensible because the dispersion parameter $\phi$ and regression coefficients $\beta$ are orthogonal in the NB model (\ref{eqn:nb_glm}). Under mild regularity conditions, and assuming stochastic regressors (as in our setup), $\hat{\beta}_n$ converges strongly to $\beta$, even if the dispersion parameter is misspecified \citep{Fahrmeir1985,Cox1987}. Assumptions \ref{assu:moment_treatment}-\ref{assu:boundedness_covariate} are mild conditions on the covariates $X_i$ and $Z_i$, and Assumption \ref{assu:lower_bound_eigenvalue} always can be satisfied by removing redundant covariates from $Z_i$.

We are now prepared to state Proposition \ref{thm:asy_distribution_misspecification}, which derives the asymptotic \textit{sampling} distribution of the NB GLM score test statistic under possibly misspecified dispersion.

\begin{proposition}[Asymptotic sampling distribution of score test statistic.]\label{thm:asy_distribution_misspecification}
Suppose that the data $(X,Y,Z)$ are generated from the NB GLM (\ref{eqn:nb_glm}) and that Assumptions \ref{assu:consistency_beta}-\ref{assu:lower_bound_eigenvalue} hold. Then under the null hypothesis (of $\gamma=0$ in 
(\ref{eqn:nb_glm})), 
	\begin{align*}
		T_n(X,Y,Z)\overset{d}{\rightarrow} N(0,\sigma_s^2), \quad \textup{where} \quad \sigma_s^2\equiv\frac{\E\left[\frac{R_i^2 S_i}{W_i} \right]}{\E\left[R_i^2\right]}.
	\end{align*}
\end{proposition}
If the dispersion parameter is correctly specified (i.e., $\bar \phi = \phi$), then $S_i = W_i$, implying that $\sigma_s^2 = 1$. In this case the score test statistic $T_n(X,Y,Z)$ converges weakly to a standard Gaussian, consistent with classical theory.

Next, Proposition \ref{thm:asy_permutation_distribution} derives the asymptotic \textit{permutation} distribution of the NB GLM score test statistic. To state Proposition \ref{thm:asy_permutation_distribution} carefully, we require the following notion of conditional convergence in distribution, which we borrow from \citet{Niu2022}.

\begin{definition}[Conditional convergence in distribution] \label{def:conditional-convergence-distribution}
	Let $W_n$ be a sequence of random variables and $\mathcal F_n$ a sequence of $\sigma$-algebras. We say that $W_n$ converges in distribution to a random variable $W$ conditionally on $\mathcal F_n$ if, for all $t \in \R$ at which the CDF of $W$ is continuous,
	\begin{equation*}
		\P[W_n \leq t \mid \mathcal F_n] \convp \P[W \leq t].
	\end{equation*}
	We denote this relation via $W_n \mid \mathcal F_n \convdp W$.
\end{definition}

\begin{proposition}[Asymptotic permutation distribution of score test statistic]\label{thm:asy_permutation_distribution}
Suppose that the data $(X,Y,Z)$ are generated from the NB GLM (\ref{eqn:nb_glm}) and that Assumptions \ref{assu:consistency_beta}-\ref{assu:lower_bound_eigenvalue} hold. Let $\mathcal{F}_n$ be the sequence of $\sigma$-algebras generated by the data. Let $\pi$ be a uniform draw from the permutation group of $\{1, \dots, n\}$, independent of $\mathcal{F}_n$. Then under the null hypothesis (of $\gamma=0$ in \ref{eqn:nb_glm}),
\begin{align*}
	T_n(X_{\pi},Y,Z) \mid \mathcal{F}_n \convdp N(0,\sigma_p^2); \quad \sigma_p^2 \equiv \frac{\E[S_i]}{\E[W_i]}.
\end{align*}
\end{proposition}
Informally, Proposition \ref{thm:asy_permutation_distribution} states $T_n(X_\pi, Y, Z)\mid (X,Y,Z) \convd N(0, \sigma_p^2),$ i.e.\ the permuted score test statistic $T_n(X_\pi,Y,Z)$ converges in distribution to a $N(0,\sigma_p^2)$ variate conditional on the original data $(X,Y,Z)$. If the dispersion parameter is correctly specified (i.e., $\bar \phi = \phi$), then $\sigma_p^2 = 1$, implying that the permuted score test statistic converges conditionally in distribution to a standard Gaussian. Importantly, Proposition \ref{thm:asy_permutation_distribution} does not make any assumptions about the dependence structure between $X_i$ and $Z_i$. The proof of Proposition \ref{thm:asy_permutation_distribution} involves extending techniques developed by \citet{Diciccio2017} for studying the asymptotic permutation distribution of the Wald statistic in a linear model. The primary technical tool is a central limit theorem for ranks due to H\'ajek \citep{Janssen1997}.

An important fact is that the standard deviation of the asymptotic sampling distribution $\sigma_s$ and that of the permutation distribution $\sigma_p$ are often quite similar in practice, even if the dispersion parameter is misspecified. To illustrate this point, we studied an NB GLM with a binary treatment $X_i$, a nuisance covariate vector $Z_i$ distributed according to a multivariate Gaussian, and a dispersion parameter $\phi$ of one. We varied the specified dispersion parameter $\bar \phi$ over the interval $[0.1, 10]$ and calculated $\sigma_s$ and $\sigma_p$ via Monte Carlo integration. The ratio $\sigma_p/\sigma_s$ remained approximately equal to one over all values of $\bar \phi$ (Figure \ref{fig:misspec_disp}a). On the other hand, $\sigma_s$ equaled one only when the specified dispersion coincided with the true dispersion (i.e., $\bar \phi = \phi = 1$) and deviated from one otherwise (Figure \ref{fig:misspec_disp}b). Given that the reference distribution of the standard score test is a standard Gaussian, setting the dispersion parameter to a value that is too small (resp., too large) causes the score test to produce an anti-conservative (resp., conservative) $p$-value.

\begin{figure}
    \centering
    \includegraphics[width=0.85\linewidth]{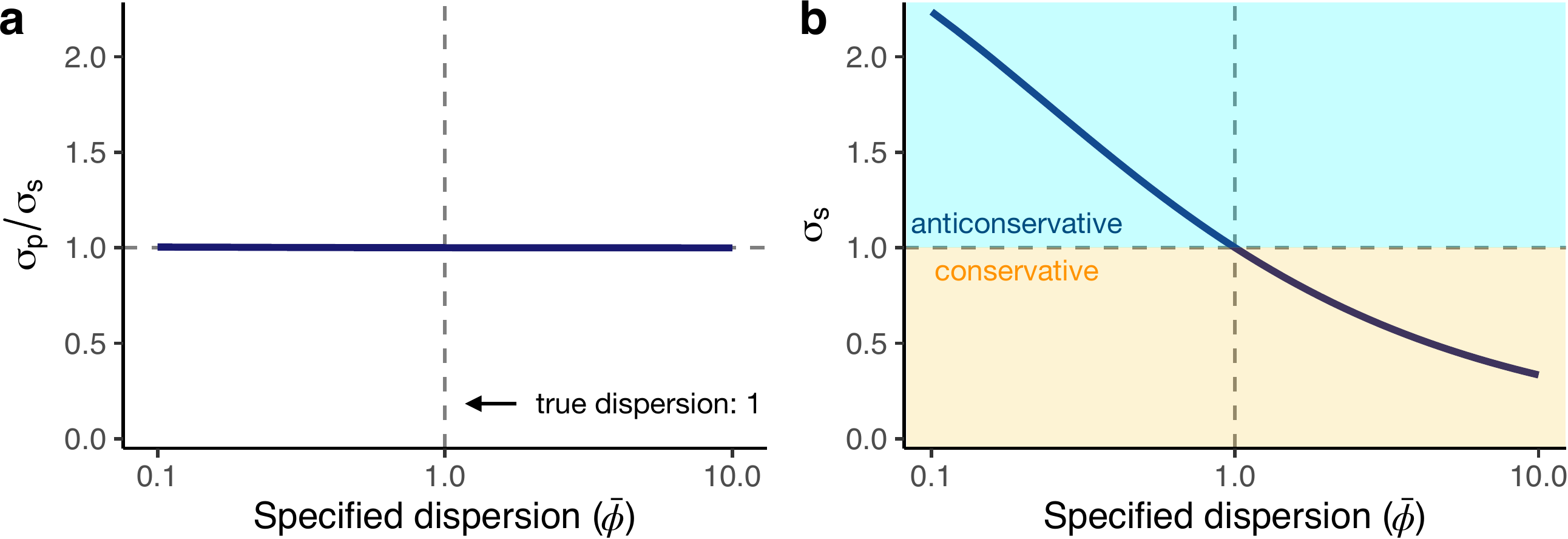}
    \caption{Effect of dispersion parameter misspecification on  $\sigma_p/\sigma_s$ and $\sigma_s$. \textbf{a} (resp., \textbf{b}), $\sigma_p/\sigma_s$ (resp., $\sigma_s$) as a function of the specified dispersion parameter $\bar \phi$ in an example NB GLM. The true dispersion $\phi$ is one.}\label{fig:misspec_disp}
\end{figure}

\subsection{Confounder adjustment via marginal permutations (CAMP)}\label{sec:camp} The permuted score test satisfies a key robustness property that we term ``confounder adjustment via marginal permutations,'' or CAMP. Informally, CAMP states that we have two separate chances to control type-I error. First, if $X_i$ is unconfounded (i.e., if $X_i$ is independent of $Z_i$), then type-I error control obtains in finite samples and under arbitrary misspecification of the NB model. On the other hand, if the data are generated by the NB GLM (\ref{eqn:nb_glm}) and Assumptions \ref{assu:consistency_beta}-\ref{assu:lower_bound_eigenvalue} hold, then type-I error control obtains in large samples. The role of the dispersion parameter in the second (i.e., large-sample) prong of CAMP is somewhat subtle. We formalize CAMP below.

\begin{theorem}[Confounder adjustment via marginal permutations]\label{thm:camp} Let $\phi_n(X,Y,Z)$ be the level-$\alpha$ test function of the permutation test based on the NB GLM score test statistic. Let $\E_\mathcal{L}\left[\phi_n(X,Y,Z)\right]$ denote the expectation of $\phi_n$ under distribution $\mathcal{L}$. Let $\mathcal{N}$ be the set of distributions for which $X_i \indep Y_i \mid Z_i$. Let $\mathcal{K}$ be the set of distributions for which $X_i \indep Z_i$. The following finite-sample type-I error bound holds under no further assumptions:
\begin{align}\label{eqn:camp_1}
\sup_{\mathcal{L} \in \mathcal{N} \cap \mathcal{K}} \E_{\mathcal{L}} \left[\phi_n(X,Y,Z) \right] \leq \alpha.    
\end{align} Next, let $\mathcal{R}$ be the set of distributions for which the data are generated by the NB GLM (\ref{eqn:nb_glm}) and Assumptions \ref{assu:consistency_beta}-\ref{assu:lower_bound_eigenvalue} hold.  The following asymptotic excess type-I error bound holds:
\begin{align}\label{eqn:camp_2}
\sup_{\mathcal{L} \in \mathcal{N} \cap \mathcal{R}} \limsup_{n\to\infty} \text{  } \E_{\mathcal{L}} [\phi_n(X, Y, Z)] \leq \alpha + \frac{1}{\sqrt{2\pi}}\Phi^{-1}(1-\alpha)(1-\sigma_p/\sigma_s),
\end{align}
where $\Phi^{-1}$ denotes the inverse of the standard Gaussian CDF.
\end{theorem}

Statement (\ref{eqn:camp_1}) follows from the contraction property of conditional independence\footnote{Contraction property of conditional independence: $(X_i \indep Z_i) \land (X_i \indep Y_i | Z_i) \implies X_i \indep (Y_i, Z_i).$} \citep{Koller2009}, while Statement (\ref{eqn:camp_2}) follows from Propositions \ref{thm:asy_distribution_misspecification} and \ref{thm:asy_permutation_distribution} and a result regarding the asymptotic equivalence of tests based on the same statistic but different critical values \citep{Niu2022}. The tightness of the asymptotic excess type-I error bound (\ref{eqn:camp_2}) depends on the closeness of the specified dispersion parameter to the true dispersion parameter. If the dispersion is correctly specified (i.e., $\bar \phi = \phi$), then the asymptotic variance of the sampling and permutation distributions coincide (i.e., $\sigma_p/\sigma_s = 1$), implying that exact type-I error control obtains in large samples\footnote{In fact, the following slightly stronger statement holds: $\sup_{\mathcal{L} \in \mathcal{N} \cap \mathcal{R}} \limsup_{n\to\infty} \E_{\mathcal{L}} [\phi_n(X, Y, Z)] = \alpha$.}:
    $$\sup_{\mathcal{L} \in \mathcal{N} \cap \mathcal{R}} \limsup_{n\to\infty} \E_{\mathcal{L}} [\phi_n(X, Y, Z)] \leq \alpha.$$ On the other hand, if the dispersion parameter is misspecified (i.e., $\bar \phi \neq \phi$), then the terms $\sigma_p$ and $\sigma_s$ (in general) disagree. However, $\sigma_p$ and $\sigma_s$ tend to be quite similar in practice (i.e., $\sigma_p/\sigma_s \approx 1$), even under severe dispersion parameter misspecification (Figure \ref{fig:misspec_disp}). If $\sigma_s$ and $\sigma_p$ are approximately equal, then excess type-I error inflation is minimal in large samples.

Although we formulated CAMP in terms of the NB GLM score test statistic, we conjecture that CAMP holds more generally for test statistics whose asymptotic distribution is invariant to the dependence structure between $X_i$ and $Z_i$. In support of this conjecture, Proposition \ref{thm:lm_asy_perm_distribution} (in Appendix \ref{sec:lin_model}) proves that the permutation distribution of the score test statistic in a \textit{linear} model converges weakly to a standard Gaussian, from which CAMP follows. CAMP is a nonparametric extension of a robustness phenomenon identified by \citet{Diciccio2017} in the context a linear model Wald statistic. See Appendices \ref{sec:notation}-\ref{sec:camp_proof} for a proof of Propositions \ref{thm:asy_distribution_misspecification} and \ref{thm:asy_permutation_distribution} and Theorem \ref{thm:camp}.

\subsection{Conceptual comparison to existing methods}

The permuted score test overcomes several limitations of standard NB regression by virtue of CAMP. First, when the treatment is unconfounded, the permuted score test protects against arbitrary model misspecification (resulting from, e.g., zero inflation, missing covariates or interaction terms, misspecification of the dispersion parameter, etc.) and asymptotic breakdown (caused by, e.g., small sample sizes or low counts). Standard NB regression, by contrast, lacks robustness to either parametric or asymptotic violations. Thus, we should expect the permuted score test to exhibit superior error control in RCTs or observational studies in which dependence between the treatment and nuisance covariates is minimal. Second, the permuted score test is less sensitive to dispersion parameter misspecification when the treatment is confounded. The ratio $\sigma_p/\sigma_s$ can approximately equal one even if the dispersion is misspecified, enabling the permuted score test to control type-I error under dispersion parameter misspecification. Standard NB regression, on the other hand, yields inflated $p$-values when the specified dispersion is too small and conservative $p$-values when the specified dispersion is too large. In summary the permuted score test, in contrast to standard NB regression, yields (essentially) assumption-free inference when the treatment is unconfounded and protects against dispersion parameter misspecification when the treatment is confounded.

Another popular method for DE testing is the the Mann-Whitney (MW) test \citep{Mann1947}, a nonparametric two-sample test. The MW test is the default method for DE testing in the popular single-cell analysis package \texttt{Seurat} \citep{Stuart2019} and was recommended by \citet{Li2022} for bulk RNA-seq analysis in lieu of NB regression (given the stringent parametric assumptions of the latter). The MW test is a finite-sample valid, nonparametric test of marginal independence between the treatment $X_i$ and response $Y_i$. Like all nonparametric two-sample tests, the MW test can be calibrated via permutations \citep{Lehmann2022}. The MW test does not adjust for nuisance covariates $Z_i$; it yields a valid test of conditional independence (\ref{ci_null_hyp}) only when the treatment is unconfounded by the nuisance covariates. The permuted score test (Algorithm \ref{algo:permuting_score_stats}) also controls type-I error when the treatment is unconfounded, but unlike the MW test, the permuted score test can control type-I error when the treatment is confounded as well (under the conditions of Theorem \ref{thm:camp}). In this sense the permuted score test can be understood as a blend of standard NB regression and the MW test, inheriting strengths of both approaches.



\section{Computational accelerations}\label{sec:computational_accelerations}

Permutation tests typically are computationally expensive; this issue is especially pronounced in multiple testing scenarios, where we must conduct a separate permutation test for each hypothesis and compute sufficiently many permutations per test to produce a $p$-value small enough to pass the multiple testing threshold. Despite this computational challenge, the permuted score test (Algorithm \ref{algo:permuting_score_stats}) is fast: as we will show empirically, its runtime is within about 5--10\% that of \texttt{MASS} and 3-4 times that of \texttt{DESeq2}. One reason for this efficiency is that the permuted score test is based on a score test statistic as opposed to a more common Wald or likelihood ratio statistic. Use of the score statistic eliminates the need to refit the NB GLM on each permuted dataset, substantially reducing compute. We implemented two additional optimizations to accelerate the method: an algorithm for efficiently computing GLM score tests and a strategy for adaptive permutation testing based on anytime-valid inference. Taken together, these optimizations accelerated the permuted score test by a factor of roughly $45,000\times$.

\paragraph{Efficiently computing GLM score tests.} Building on the work of \citet{Barry2024}, we developed an algorithm for efficiently computing GLM score tests. Both our algorithm and the classical algorithm \citep{Dunn2018} begin with a QR decomposition of the matrix $\hat{W}^{1/2}Z$, which is computed as part of the GLM fitting procedure. The classical algorithm leverages the $Q$ matrix from this QR decomposition to compute the score test statistic; our algorithm, by contrast, exploits the $R$ matrix for this purpose. We refer to the classical algorithm as ``Algorithm $Q$'' and our algorithm as ``Algorithm $R$''. Algorithm $R$ generally requires fewer operations than Algorithm $Q$ and is especially well-suited to settings in which the treatment vector (i.e., the vector to be tested for inclusion in the fitted model) contains many zeros. (Note that the treatment vector contains at least 50\% zeros in a two-sample testing problem, possibly after label swapping.)

We compared Algorithm $R$ to Algorithm $Q$ from both theoretical and empirical perspectives. First, we analytically tallied the total number of floating point operations (i.e., additions, subtractions, multiplications, and divisions) required of both algorithms as a function of the sample size $n$, the number of covariates $p$, the fraction $\pi$ of nonzero entries within the treatment vector, and the number $B$ of permuted treatment vectors to test for inclusion in the fitted model. (We excluded the GLM fitting step --- which is shared between both algorithms --- from this calculation.) We noted that the total number of floating point operations grew more rapidly for Algorithm $Q$ than Algorithm $R$ as a function of $n$, $p$, $\pi$, and $B$. As an illustration, Figure \ref{fig:lin_alg_main}a displays the operation count of both algorithms as a function of sample size $n$ (with fixed $p$, $B$, and $\pi$).

Next, we benchmarked the empirical runtime of Algorithm $R$ and Algorithm $Q$, varying the sample size $n$ while holding fixed $p$, $B$, and $\pi$. We used the implementation of Algorithm $Q$ provided by the \texttt{statmod} R package \citep{Dunn2018}. We replicated this experiment $40$ times for each value of $n$ and calculated the mean runtime across replicates. Algorithm $R$ was considerably faster than Algorithm $Q$, especially for large values of $n$ (Figure \ref{fig:lin_alg_main}b). The empirical runtime plot (Figure \ref{fig:lin_alg_main}b) closely mirrored the operation count plot (Figure \ref{fig:lin_alg_main}a), indicating good agreement between theoretical and empirical efficiency. When we included the time required to fit the NB GLM (using \texttt{MASS}) in the empirical runtime calculation, the results did not change significantly (Figure \ref{fig:lin_alg_main}c). Finally, the $z$-scores outputted by Algorithm $R$ and Algorithm $Q$ were highly concordant (Figure \ref{fig:lin_alg_main}d), indicating no degradation in quality on the part of Algorithm $R$. See Appendix \ref{sec:computing_glm_score_tests} for a more detailed comparison of Algorithm $R$ and Algorithm $Q$.

\begin{figure}[htbp]
	\begin{center}
		\includegraphics[width=0.7\textwidth]{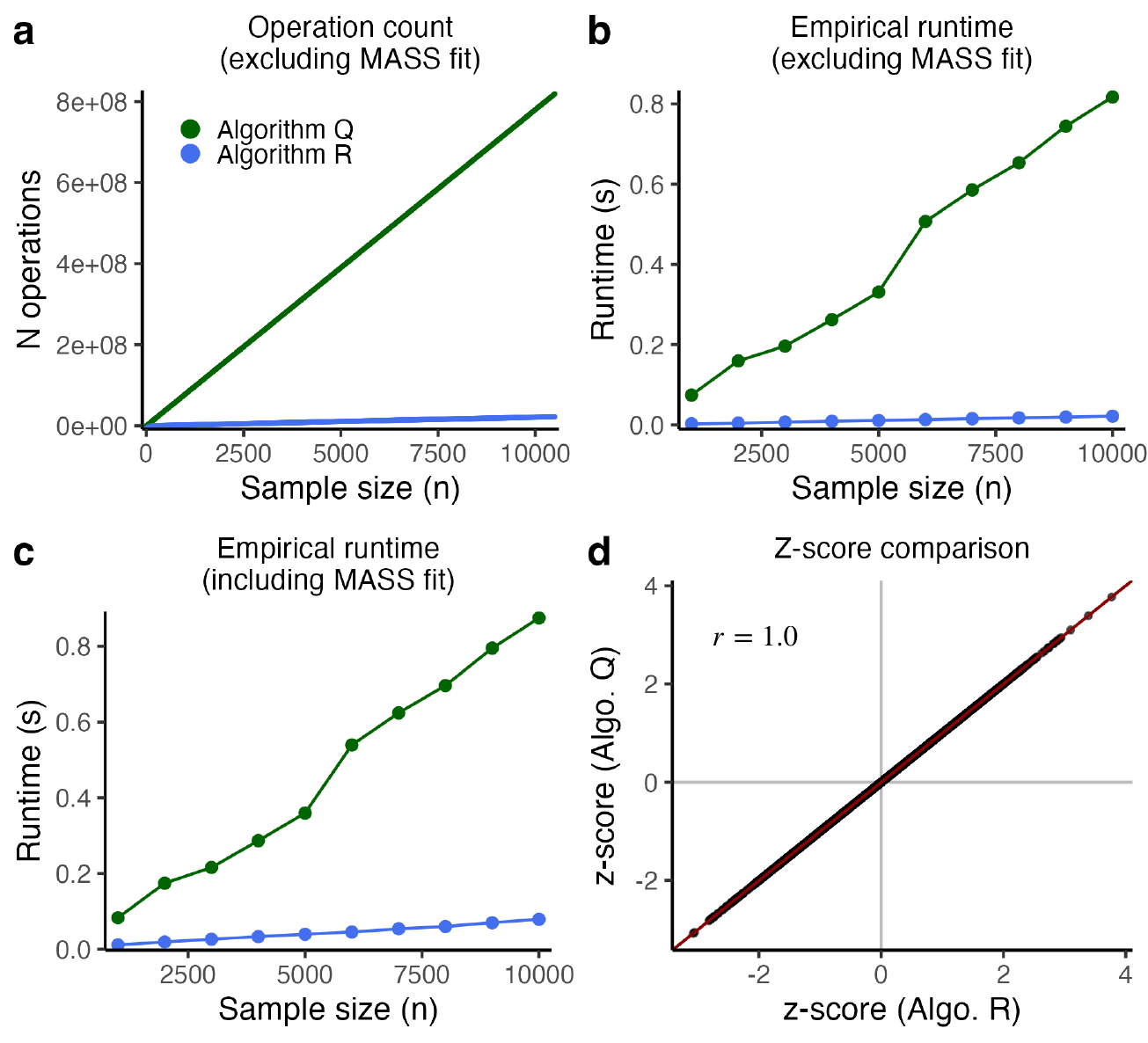}
		\caption{Algorithm $R$ and Algorithm $Q$ for computing GLM score tests. \textbf{a}, floating point operation count; \textbf{b} empirical runtime (excluding GLM fit); \textbf{c} empirical runtime (including GLM fit); \textbf{d} comparison of $z$-scores outputted by Algorithms $R$ and $Q$.}\label{fig:lin_alg_main}
	\end{center}
\end{figure}

\paragraph{Adaptive permutation testing via anytime-valid inference.} One typically sets the number of permutations $B$ to a large, predetermined number. However, this approach can be inefficient, as one may be able to terminate the permutation test early if the evidence against the null hypothesis is either very strong or very weak. Recently, \citet{Fischer2024a} developed a framework for selecting the number of permutations in a data-adaptive manner, leveraging techniques from anytime-valid inference \citep{Howard2021,Ramdas2023}. Their approach allows for a significant reduction in the number of permutations while maintaining type-I error control and sacrificing minimal power. Subsequently, \citet{Fischer2024b} extended these ideas to the multiple testing context, proposing a method for generating an FDR-controlling discovery set via adaptive permutation testing of multiple hypotheses. We applied the method of \citet{Fischer2024b} to accelerate the permuted score test. See Appendix \ref{sec:adaptive_permutation_testing} for more information about this method.

\section{Simulation studies}\label{sec:sim_studies}

We conducted a series of simulation studies to compare the permuted score test to three competing methods: standard NB regression, the MW test, and a permutation test based on NB GLM residuals. The latter method --- which we term the ``residual permutation test'' --- entails (i) regressing $Y$ onto $Z$ by fitting an NB GLM; (ii) extracting the vector $e$ of residuals from the fitted model; and (iii) testing for association between $e$ and $X$ via permutations (Algorithm \ref{algo:permuting_residuals}). We included the residual permutation test in our experiments to investigate whether our score-based permutation test might improve over this simpler alternative. We calibrated the MW test via permutations to ensure finite-sample validity. Additionally, we applied the adaptive permutation testing scheme of \citet{Fischer2024b} (Section \ref{sec:computational_accelerations}) to accelerate the MW test, the residual permutation test, and (as discussed earlier) the permuted score test. Finally, we implemented the NB regression-based methods using \texttt{MASS}; we refer to standard NB regression as ``\texttt{MASS}'' our ``robustified'' version of \texttt{MASS} as ``robust \texttt{MASS}'' throughout our exposition.

\begin{algorithm}
	\caption{Residual permutation test. The test statistic $S : \{0,1\}^n \times \R^n \to \R$ is defined as $S(X, e) = s^{-1/2} \sum_{i : X_i = 1} e_i$, where $s = \sum_{i=1}^n X_i$ is the number of ones contained within the binary vector $X$.}\label{algo:permuting_residuals}
 Regress $Y$ onto $Z$ via an NB GLM; extract the residual vector $e = Y - \hat{\mu}$.
	
Compute $t_\textrm{orig} = S(X, e).$

Randomly permute $X$, yielding permuted vectors $\tilde{X}_1, \dots, \tilde{X}_B$

For each $b \in \{1, \dots, B\}$, compute $t_b = S(\tilde{X}_b, e).$

Compute a $p$-value by comparing $t_\textrm{orig}$ to $t_1, \dots, t_B$ via (\ref{eqn:perm_p_value}).
\end{algorithm}

Our simulation study design was motivated by bulk RNA-seq experiments, such as arrayed CRISPR screens with bulk RNA-seq readout. We considered two main settings for the data-generating process: ``Setting 1'' (Figure  \ref{fig:sim_1}), where $X_i$ and $Z_i$ were independent, and ``Setting 2,'' (Figure \ref{fig:sim_2}), where $X_i$ and $Z_i$ were dependent. Within each setting we conducted three numerical experiments, which involved: (\textit{a}) generating data from the NB GLM (\ref{eqn:nb_glm}) with a \textit{large} sample size; (\textit{b}) generating data from the NB GLM with a \textit{small} sample size; and (\textit{c}) generating data from a zero-inflated NB GLM with a large sample size. (The latter experiment aimed to explore the impact of model misspecification on method performance.) In each experiment we varied an experiment-specific parameter of the data-generating process (to be described later) over a grid of values. Next, for each parameter configuration, we simulated a treatment vector $X \in \{0, 1\}^n$, a design matrix $Z \in \mathbb{R}^{n \times p}$, and $m = 500$ gene expression vectors $Y^1, \dots, Y^m \in \mathbb{Z}^{n}.$ Of these gene expression vectors, 90\% (i.e., 450) were generated under the null hypothesis and 10\% (i.e., 50) under the alternative. We conducted 1,000 (or more) Monte Carlo replicates for each parameter configuration to estimate the FDR, the mean number of true discoveries, and the mean running time of each method. Importantly, we did not specify the dispersion parameter; rather, we estimated it using \texttt{MASS}, as this is more common in practice.

\textbf{Simulation 1\textit{a}}: We generated data from the NB GLM (\ref{eqn:nb_glm}) with an unconfounded treatment and a large ($n = 1,000$) sample size. We varied the effect size $\gamma$ under the alternative hypothesis from $0.05$ to $1.5$. The left (resp., middle) panel plots the FDR (resp., mean number of true discoveries) of each method versus effect size, while the right panel displays the mean runtime of each method averaged across parameter settings. As expected, all methods controlled the FDR at the nominal level of 10\%. Moreover, the methods made monotonically more discoveries as the effect size increased. \texttt{MASS}, robust \texttt{MASS}, and the residual permutation test exhibited similar power and computational efficiency. The MW test was less powerful, as it did not adjust for covariates; however, it was fastest. 

\textbf{Simulation 1\textit{b}}: We generated data from the NB GLM (\ref{eqn:nb_glm}) with an unconfounded treatment, varying the sample size from a small value ($n = 100$) to a large value ($n = 1,000$). \texttt{MASS} exhibited severe type-I error inflation when the sample size was small, yielding an FDR exceeding $35\%$ for $n = 100$. This loss of error control likely resulted from inaccurate dispersion parameter estimation in small samples.\footnote{However, NB regression can fail to control type-I error in small samples \textit{even if} the dispersion is known exactly \citep{Niu2024}.} The other three methods, by contrast, maintained error control across all sample sizes, as these methods yield valid inference when the treatment is unconfounded. The three NB regression-based methods exhibited approximately equal power for large sample sizes. 

\textbf{Simulation 1\textit{c}}: We sought to study zero inflation, a canonical example of model misspecification \citep{Jiang2022}. We generated data from the NB GLM (\ref{eqn:nb_glm}), corrupting the response $Y_i$ by setting it to zero with probability $\psi \in [0,1]$. We generated the treatment such that it was unconfounded, and we set the sample size to a large value ($n = 1,000$). We varied the zero inflation probability over the grid $\psi \in \{0.0, 0.02, \dots, 0.16\}$. All methods controlled the FDR; however, \texttt{MASS} lost power more rapidly than the other methods as the extent of zero inflation increased. These results illustrate that violating the assumptions of NB regression can lead not only to false positive but also false negative error inflation. 

\begin{figure}
	\begin{center}
    \includegraphics[width=1\textwidth]{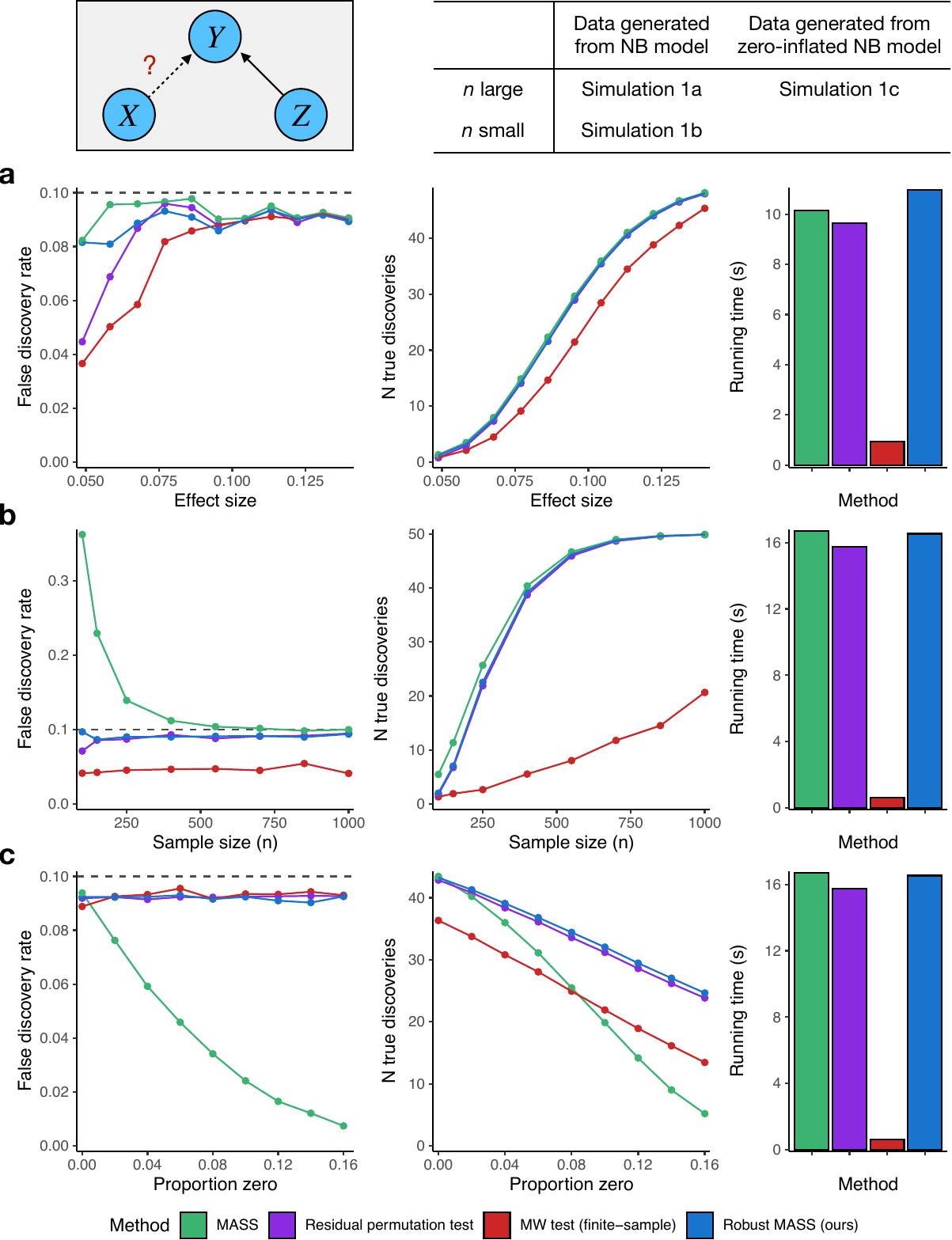}
		\caption{Simulation setting 1: $X_i \indep Z_i$.}\label{fig:sim_1}
	\end{center}
\end{figure}

\textbf{Simulation 2\textit{a}}: We generated data from the NB GLM (\ref{eqn:nb_glm}) with a confounded treatment and a large sample size. We modeled the treatment $X_i$ as a logistic function of the nuisance covariates $Z_i$, i.e. $X_i \sim \textrm{Bern}(\pi_i)$, $\textrm{logit}(\pi_i) = \delta^\top Z_i$, and $\delta \in \R^p.$ We varied the strength of dependence between the treatment and nuisance covariates, ranging from independence ($\left\|\delta\right\|_\infty = 0$) to moderate correlation ($\left\|\delta \right\|_\infty = 0.8$). As confounding strength increased, the MW test lost type-I error control, as it did not adjust for the nuisance covariates. The other three methods, by contrast, maintained error control; robust \texttt{MASS} controlled type-I error by virtue of CAMP (Theorem \ref{thm:camp}), and we hypothesize that the residual permutation test controlled type-I error as a result CAMP adapted to the residual-based test statistic. \texttt{MASS} and robust \texttt{MASS} exhibited approximately equal power; the residual permutation test, on the other hand, was less powerful, especially when dependence between the treatment and nuisance covariates was moderate. This power loss occurred because the residual-based test statistic lacks the key normalization term present in the denominator of the score test statistic (\ref{eq:score_statistic}). 

\begin{figure}
	\begin{center}
\includegraphics[width=1\textwidth]{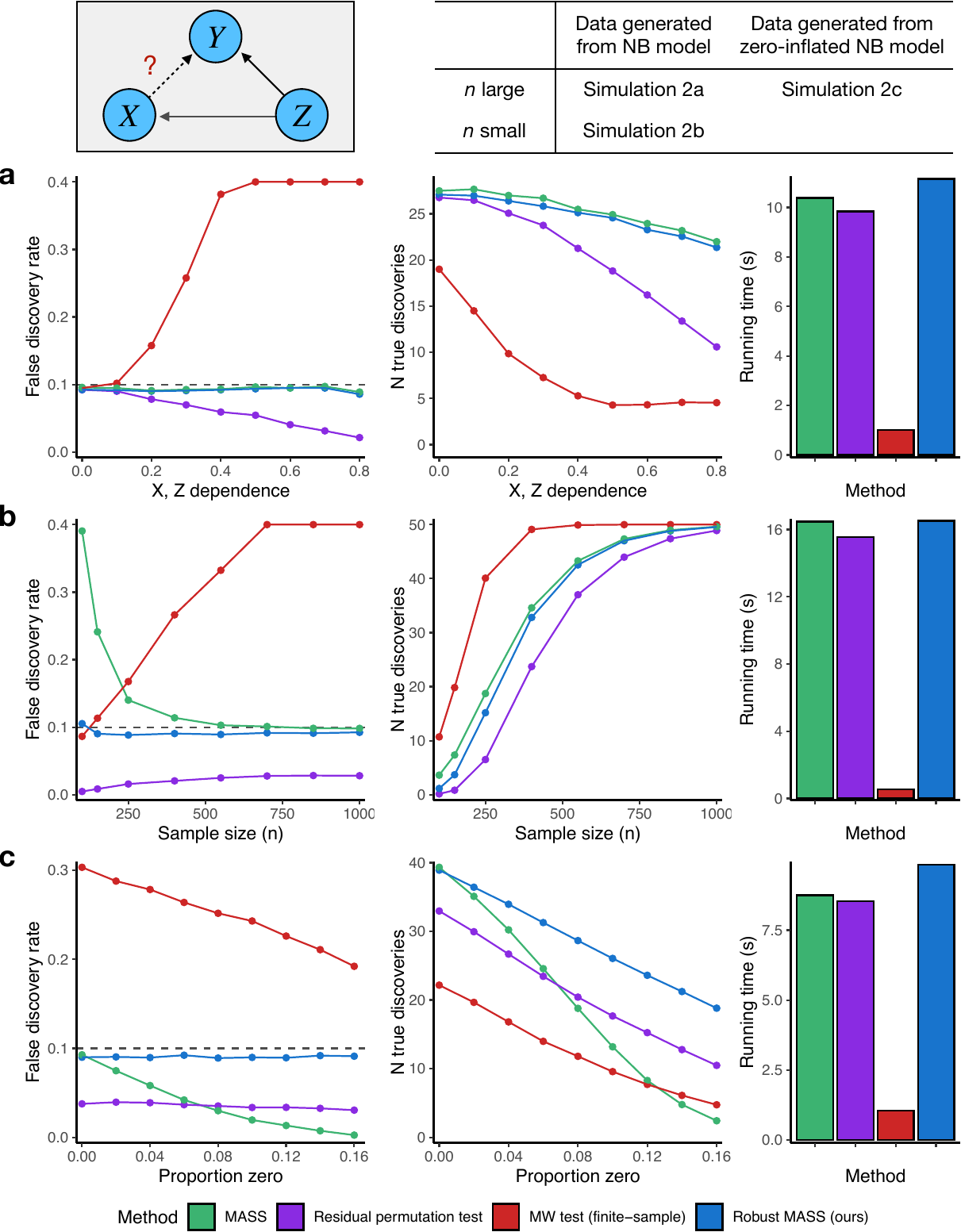}
		\caption{Simulation setting 2: $X_i \dep Z_i$.}\label{fig:sim_2}
	\end{center}
\end{figure}

\textbf{Simulation 2\textit{b}}: We generated data from the NB GLM (\ref{eqn:nb_glm}) with a confounded treatment, varying the sample size from a small value ($n = 100$) to a large value ($n = 1,000$). Again, the MW test failed to control type-I error due to confounding. Additionally, \texttt{MASS} exhibited type-I error inflation when the sample size was small, likely due to inaccurate dispersion parameter estimation. By contrast, robust \texttt{MASS} controlled type-I error over all sample sizes, as robust \texttt{MASS} is considerably less sensitive to inaccurate dispersion parameter specification than \texttt{MASS} (Section \ref{sec:camp}). Moreover, robust \texttt{MASS} matched the power of \texttt{MASS} in large sample sizes and exhibited greater power than the residual permutation test across all sample sizes. 

\textbf{Simulation 2c}: We generated data from a zero-inflated NB GLM (as in \textbf{Simulation 1c}) but this time with a confounded treatment. As before, the MW test exhibited type-I error inflation due to confounding. The three NB regression-based methods, by contrast, controlled type-I error over all values of $\delta$, possibly due to an analogue of CAMP for misspecified NB GLMs. \texttt{MASS} lost power more rapidly than any other method, while robust \texttt{MASS} was the most powerful method by a large margin. This simulation study suggests that the permuted score test has favorable performance in scenarios beyond those covered by our theoretical results.

\textbf{Additional simulation studies}: We conducted several additional simulation studies (Appendix \ref{sec:additional_simulation_studies}). Simulation S1 studied the effect of a low signal-to-noise ratio; Simulation S2 more carefully explored the impact of setting the dispersion parameter to an incorrect value; Simulation S3 investigated use of Pearson or deviance residuals in place of response residuals in the residual permutation test; and Simulation S4 evaluated \texttt{DESeq2}. These additional studies confirmed the favorable performance of the permuted score test.

\section{Real data analysis}\label{sec:real_data_analysis}

We developed a robust version of \texttt{DESeq2} based on the permuted score test (``robust \texttt{DESeq2}'') and applied this method to analyze a real arrayed CRISPR screen dataset with bulk RNA-seq readout \citep{Schmidt2022}. We compared robust \texttt{DESeq2} to \texttt{MASS}, \texttt{DESeq2}, the finite-sample MW test (calibrated via permutations; Section \ref{sec:sim_studies}), and the asymptotic MW test (implemented via the \texttt{wilcox.test()} function in R). We included the latter method to evaluate the relative performance of permutations and asymptotic approximations on the real data. To assess the type-I error control of the methods, we generated a negative control dataset by randomly permuting the expression vector of each gene (independently of the other genes), as in Section \ref{sec:result_preview}. \texttt{MASS} and \texttt{DESeq2} produced inflated $p$-values, the asymptotic MW test deflated $p$-values, and robust \texttt{DESeq2} and the finite-sample MW test well-calibrated $p$-values (Figure \ref{fig:arrayed_screen}a). We repeated this experiment across $B = 1,000$ Monte Carlo replicates, finding that \texttt{MASS} (resp., \texttt{DESeq2}) made an average of 1,138 (resp., 83) false discoveries across replicates; robust \texttt{DESeq2}, the finite-sample MW test, and the asymptotic MW test, on the other hand, yielded zero false discoveries in each replicate (Figure \ref{fig:arrayed_screen}b).

Next, we analyzed the original (i.e., unpermuted) CRISPR screen data (Figure \ref{fig:arrayed_screen}c). Robust \texttt{DESeq2} made about $3,500$ discoveries (out of $\approx 27,000$ genes). \texttt{MASS} and \texttt{DESeq2}, meanwhile, made approximately $8,000$ and $4,000$ discoveries, respectively. Given the inflation of \texttt{MASS} and \texttt{DESeq2} on the negative control data, we reasoned that many of these discoveries likely were false positives. Next, the finite-sample MW test and the asymptotic MW test made 292 and zero discoveries, respectively. The MW test was underpowered because it did not adjust for key covariates, such as donor source. As an orthogonal point, the asymptotic MW test likely was less powerful than its finite-sample counterpart because of inaccurate asymptotic approximations, highlighting the utility of finite-sample methods for bulk RNA-seq data.

Encouragingly, all methods --- excluding the asymptotic MW test --- discovered gene \textit{FOXQ1}, a positive control gene targeted for increased expression by the CRISPR perturbation. We benchmarked the computational efficiency of the methods (across 500 runs), finding that the runtime of robust \texttt{DESeq2} (mean runtime 49.7\textit{s}) was within 3.5 times that of \texttt{DESeq2} (mean runtime 14.4\textit{s}). \texttt{MASS} (mean runtime 345\textit{s}), on the other hand, was considerably slower than either of these methods. Finally, the discovery set of \texttt{DESeq2} and robust \texttt{DESeq2} diverged substantially (Figure \ref{fig:arrayed_screen}d): of the 4,347 discoveries made by either method, \texttt{DESeq2} and robust \texttt{DESeq2} agreed on only 3,132 genes, or 72\%.

\begin{figure}
	\begin{center}
\includegraphics[width=0.85\textwidth]{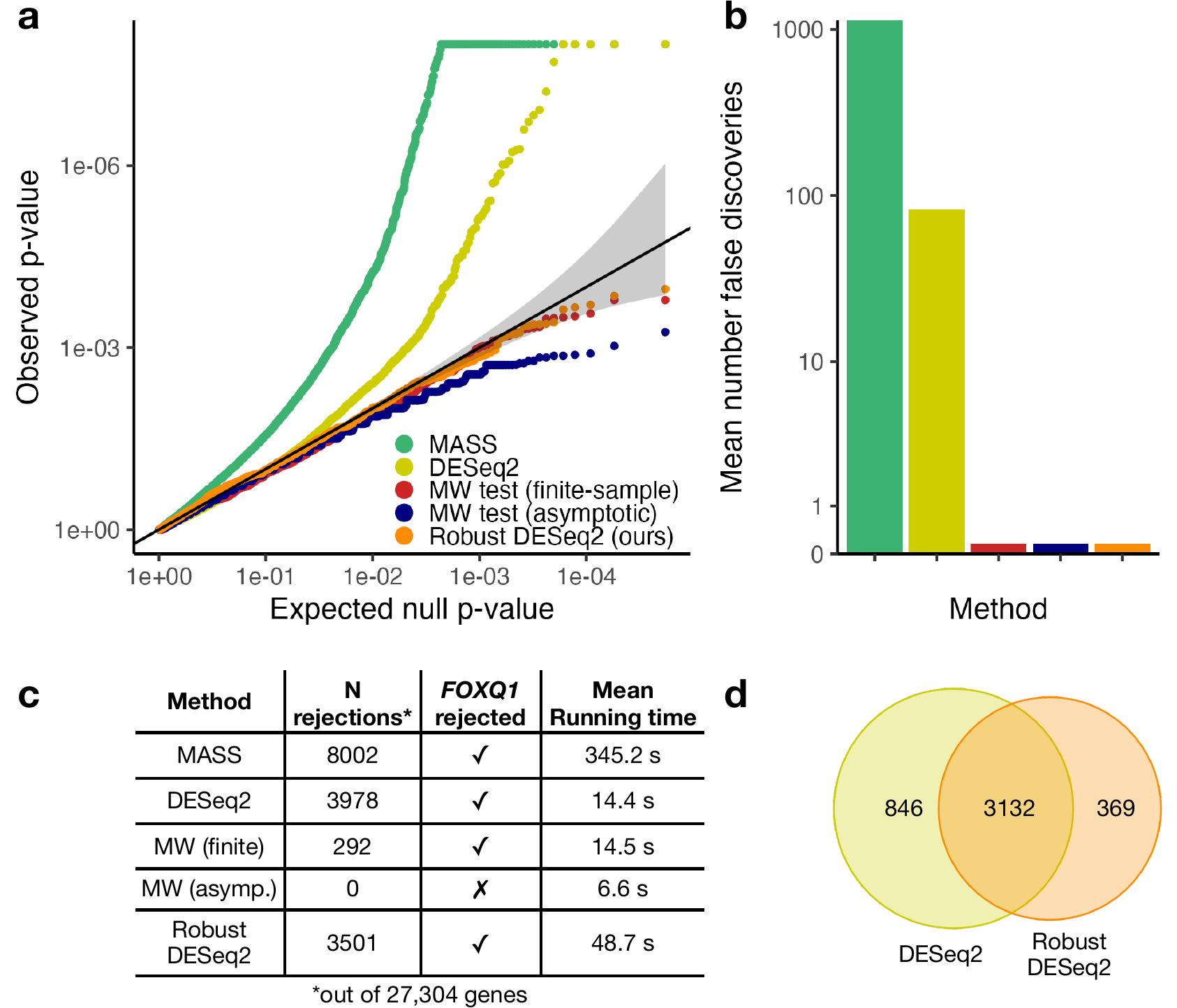}
		\caption{Analysis of a real arrayed CRISPR screen dataset. \textbf{a}, QQ plot of negative control $p$-values on a single negative control dataset; \textbf{b}, average number of false discoveries (after a BH correction at level $0.1$) across Monte Carlo replicates; \textbf{c}, table summarizing the original (i.e., unpermuted) data analysis; \textbf{d}, Venn diagram comparing the discovery set of \texttt{DESeq2} and robust \texttt{DESeq2}.}\label{fig:arrayed_screen}
	\end{center}
\end{figure}

\section{Additional related literature}

Our work intersects three main literatures, each vast: differential expression analysis, resampling-based testing, and robust inference for GLMs. We highlight the works from these literatures most relevant to our own.

\paragraph{Differential expression analysis.}

The most popular methods for bulk RNA-seq DE analysis are built on parametric models: \texttt{DESeq2} \citep{Love2014} and \texttt{EdgeR} \citep{Robinson2009} are based an NB model, and \texttt{limma} is based on a Gaussian model for transformed counts. Nonparametric methods, while less popular, also are available: \citet{Li2013} developed a DE test based on resampling MW statistics, and \citet{Gauthier2020} introduced a variance component score test robust to arbitrary model misspecification. DE methods for single-cell data are less established, and consensus regarding which is ``best'' remains elusive \citep{Heumos2023}. The most common approach is probably the MW test as implemented by \texttt{Seurat} \citep{Stuart2019}. Other options include \texttt{MAST} \citep{Finak2015}, which is based on a hurdle regression model, and \texttt{Memento}, which is based on bootstrapping a method-of-moments estimator derived from a hypergeometric model \citep{Kim2024}. Finally, the method proposed in this work is a more sophisticated version our own \texttt{sceptre} method for single-cell CRISPR screen (i.e., perturb-seq) analysis \citep{Barry2024}.

\paragraph{Resampling-based testing.} Several methods are available for testing conditional independence (without reference to any particular model) via resampling, including the conditional randomization test (CRT; \cite{Candes2018}), the conditional permutation test (CPT; \cite{Berrett2020}), and local permutation test (LPT; \cite{Kim2022}). The CRT and CPT involve resampling or permuting the treatment vector conditional on the nuisance covariates; these methods yield finite-sample valid inference when the conditional distribution of the treatment given the nuisance covariates is known exactly, which may be a reasonable assumption in some applications \citep{Sesia2019,Ham2024}. 
Next, the LPT entails permuting the treatment within strata or ``bins'' of the nuisance covariate; it yields (approximately) finite-sample valid inference when (i) $Y_i \indep X_i | \tilde{Z}_i$ and (ii) $(X_i, Y_i, \tilde{Z}_i)$ is ``close to'' $(X_i, Y_i,Z_i)$, where $\tilde{Z}_i$ is the binned version of $Z_i$. An advantage of the CRT, CPT, and LPT is that these methods are model-agnostic and are compatible with arbitrary test statistics, such as statistics derived from machine learning models \citep{Bates2020,Liu2022}.

Permutation-based methods also have been developed for testing the more granular hypothesis of whether a coefficient in a linear model is equal to zero. The classical Freedman-Lane (FL) procedure and related methods are based on permuting a test statistic constructed from the residuals of a linear model \citep{Freedman1983,Winkler2014}. Although it generally performs well in practice, the FL procedure does not provide finite-sample guarantees under any scenario \citep{Diciccio2017}. Recently, \citet{Lei2021} developed the cyclic permutation test, a test for linear hypotheses that assumes exchangeable (i.e., not necessarily Gaussian) errors, achieving finite-sample validity while sidestepping parametric assumptions. \citet{Wen2024} and \citet{Guan2024} later improved the cyclic permutation test along several dimensions, although it is unclear whether these techniques can be applied to GLMs.

Romano, Janssen, and others have demonstrated that certain permutation test statistics are asymptotically robust in a spirit similar to CAMP \citep{Romano1990,Janssen1997,Chung2013}. For example, the two-sample permutation test based on an appropriately studentized sample mean (or median) test statistic yields an asymptotically valid test of mean (or median) equality, even if the data are not exchangeable. \citet{Diciccio2017} extended these ideas to study the linear model Wald statistic, and we extended \citet{Diciccio2017} to study the NB GLM score test statistic. Finally, our framework leverages techniques for adaptive permutation testing. Classical adaptive permutation testing methods (e.g., \citet{Besag1991}) allow for early stopping under the null, while newer methods, such as the multi-armed bandit method of \citet{Zhang2019} and the anytime-valid method of \citet{Fischer2024a}, allow for early stopping under both null \textit{and} alternative, saving compute.

\paragraph{Robust inference and estimation in GLMs.}
Several approaches have been proposed for robust inference and estimation in GLMs, including sandwich estimators, resampling procedures, and M-estimators. The classical sandwich estimator \citep{Huber1967,White1982} yields robust standard errors for the MLE but does not appear to allow finite-sample inference under any circumstances. Several authors have proposed resampling-based techniques, including bootstrapping \citep{Moulton1991} and random sign-flipping \citep{Hemerik2020}, which can improve upon the accuracy of the sandwich standard errors in small samples. However, these methods also do not appear to provide valid finite-sample inference in any scenario. In parallel, M-estimation approaches have been developed to obtain robust estimates of GLM coefficients while mitigating the effects of outliers \citep{Cantoni2001,Bianco2013}. \citet{Aeberhard2014} applied this idea to negative binomial (NB) regression, carefully handling the dispersion parameter. This line of work is complementary to our effort, as our framework is compatible with arbitrary NB GLM fitting procedures, including those based on robust M-estimation methods.

\section{Discussion}

NB regression is the most popular tool for differential expression testing, a fundamental statistical task in the analysis of genomics data. We proposed the permuted score test, a robust approach to NB regression based on adaptively permuting score test statistics. By virtue of CAMP, the permuted score test provably controls type-I error across a broader range of settings than standard NB regression, protecting against model misspecification (e.g., dispersion parameter misspecification, zero inflation, missing covariates, etc.) and asymptotic breakdown (caused by small sample sizes or low counts). We conducted extensive simulation studies and a real data analysis, demonstrating that the permuted score test substantially enhances the robustness of NB regression while nevertheless approximately matching NB regression with respect to power (when the assumptions of NB regression obtain) and computational efficiency. We expect the permuted score test to bolster the reliability of differential expression analysis across diverse biological contexts.

CAMP is related to --- but distinct from --- double robustness, a robustness phenomenon common among methods for estimating causal effects and testing conditional independence \citep{Chernozhukov2018,Kennedy2024}. A quintessential doubly robust method for conditional independence testing is the generalized covariance measure test (GCM; \cite{Shah2020}). The GCM constructs a test statistic using residuals from an ``outcome'' regression ($\E[Y_i | Z_i]$) and a ``treatment'' regression ($\E[X_i | Z_i]$). Suppose for simplicity that the outcome and treatment regression functions are modeled parametrically. If either model is correctly specified, then the GCM yields an asymptotically valid test of conditional independence \citep{Smucler2019}. In this sense the GCM and CAMP each provide ``two chances'' for obtaining valid inference. However, CAMP differs from the GCM in several respects. First, CAMP in effect replaces the treatment regression with a permutation mechanism, requiring estimation of only the outcome regression function. Furthermore, CAMP offers a mix of asymptotic and non-asymptotic guarantees, while the GCM provides strictly asymptotic guarantees. Carefully studying the interplay between CAMP and double robustness --- as well as characterizing CAMP for more general test statistics (e.g., statistics derived from nonparametric regression estimators) --- are potentially exciting directions for future work.


\printbibliography

\newpage
\appendix

\textbf{Organization of Appendix}: The Appendix is organized as follows. Section \ref{sec:notation} introduces notation that we will use throughout. Section \ref{sec:general_lemmas} states several general lemmas. Section \ref{sec:nb_lemmas} states and proves lemmas specific to the NB GLM. Section \ref{sec:prop_asy_proof} proves Proposition \ref{thm:asy_distribution_misspecification}, Section \ref{sec:asy_permutation_distribution_proof} proves Proposition \ref{thm:asy_permutation_distribution}, and Section \ref{sec:camp_proof} proves Theorem \ref{thm:camp}. Section \ref{sec:lin_model} derives the asymptotic permutation distribution of the score test statistic in the linear model. Section \ref{sec:additional_simulation_studies} reports the results of several additional simulation studies. Section \ref{sec:computing_glm_score_tests} explicates Algorithm $R$ and Algorithm $Q$ in greater detail. Finally, Section \ref{sec:adaptive_permutation_testing} provides a brief primer on adaptive permutation testing via anytime-valid inference.

\section{Notation}\label{sec:notation}

We will find it convenient to parameterize the NB model in terms of its ``size parameter'' as opposed to its dispersion parameter. The size parameter $\theta$ is the inverse of the dispersion parameter, i.e.\ $\theta = \phi^{-1}$. An NB random variable $Y_i$ with mean $\mu_i$ and size parameter $\theta$ has probability mass function
$$ p(y_i; \mu_i, \theta) = \binom{y_i - \theta - 1}{y_i}\left( \frac{\mu_i}{\mu_i+k} \right)^{y_i} \left( \frac{\theta}{\theta + \mu_i}\right)^\theta.$$ We do not assume that the user-specified size parameter $\bar \theta$ coincides with the true size parameter; in other words, we allow $\bar \theta \neq \theta$. Let $U_{\bar \theta}(\beta)$ be the (misspecified) score equation for $\beta$ under $\gamma = 0$, i.e.
$$U_{\bar \theta}(\beta) = \frac{1}{n} \sum_{i=1}^n \frac{Z_i(Y_i - \mu_i)}{1 + \mu_i/\bar \theta}.$$ Let the estimator $\hat{\beta}_n$ be a solution to the score equation, i.e.\ $U_{\bar \theta}(\hat \beta_n) = 0.$ We define the random quantities $T_n(X,Y,Z)$, $\hat \beta_n$ $\hat r$, $\hat r_i$, $\hat \mu_i$, $\hat W$, $\hat W_i$, $W_i$, $S_i$, $\mu_i$, and $R_i$ in the same way as in Section \ref{sec:asymptotic_results}, but with $\bar \phi$ replaced by $1/\bar \theta$ and $\phi$ replaced by $1/\theta$. Finally, we define the random variable $\xi_i$ as the theoretical analogue of $\hat{r}_i$, i.e.\
$$
\xi_{i}\equiv \frac{Y_i-\mu_{i}}{1+\mu_{i}/\bar{\theta}}.
$$

For a natural number $n \in \N$, let $[n]$ denote the set $\{1, 2, \dots, n\}$. Next, let $A$ be a matrix and $x$ a vector. Throughout the Appendix, $\left\|x\right\|$ denotes the 2-norm of $x$ and $\left\|A\right\|$ the spectral norm of $A$. Let $x_1, x_2, \dots$ be a sequence of random vectors (of fixed dimension $p$) and $A_1, A_2, \dots$ a sequence of random matrices (of fixed dimension $p \times p$). We define $x_n \convp x$ by
\begin{equation}\label{eqn:def_vector_conv}
		\left\|x_n - x\right\| \convp 0
\end{equation}
and $A_n \convp A$ by
\begin{equation}\label{eqn:def_matrix_conv}
	\left\| A_n - A \right\| \convp 0.
\end{equation}
Note that if $(x_j)_n \convp x_j$ for all $j \in [p]$, then (\ref{eqn:def_vector_conv}) follows; likewise, if $(A_{jk})_n \convp A_{jk}$ for all $j,k \in [p]$, then (\ref{eqn:def_matrix_conv}) follows.

\section{General lemmas}\label{sec:general_lemmas}

We borrow Definitions \ref{def:conditional-convergence-distribution}--\ref{def:conditional-convergence-probability} and Lemmas \ref{lem:cond_slutsky}--\ref{lem:conditional-convergence-to-quantile} from \citet{Niu2022}. Additionally, we borrow Lemmas \ref{lem:hajek_clt}-\ref{lem:hoeffding_identity} from \citet{Diciccio2017} and \citet{Janssen1997}.

\begin{definition}[Conditional convergence in probability] \label{def:conditional-convergence-probability}
	Let $W_n$ be a sequence of random variables and $\mathcal F_n$ a sequence of $\sigma$-algebras. We say that $W_n$ converges in probability to a constant $c$ conditionally on $\mathcal F_n$ if $W_n$ converges in distribution to the delta mass at $c$ conditionally on $\mathcal F_n$ (recall Definition~\ref{def:conditional-convergence-distribution}). We denote this convergence by ${W_n \mid \mathcal F_n \convpp c}$. In symbols, 
	\begin{equation*}
		W_n \mid \mathcal F_n \convpp c \quad \text{if} \quad W_n \mid \mathcal F_n \convdp \delta_c.
	\end{equation*}
\end{definition}

\begin{lemma}[Conditional Slutsky's theorem]\label{lem:cond_slutsky}
	Let $W_n$ be a sequence of random variables and $\mathcal F_n$ a sequence of $\sigma$-algebras. Suppose $a_n$ and $b_n$ are sequences of random variables such that $a_n \convp 1$ and $b_n \convp 0$. If $W_n \mid \mathcal F_n \convdp W$ for some random variable $W$ with continuous CDF, then
	\begin{equation*}
		a_n W_n + b_n \mid \mathcal F_n \convdp W.
	\end{equation*}
\end{lemma}

\begin{lemma}[Conditional convergence in probability implies marginal convergence in probability]\label{lem:cond_conv_prob_marginal_conv_prob}
	Let $W_n$ be a sequence of random variables and $\mathcal F_n$ a sequence of $\sigma$-algebras. If $W_n | \mathcal{F}_n \convpp 0$, then  $W_n \convp 0$.
\end{lemma}

\begin{lemma}[Asymptotic equivalence of tests] \label{lem:equivalence-lemma}
	
	Consider two hypothesis tests based on the same test statistic $S_n(X, Y, Z)$ but different critical values:
	\begin{equation*}
		\phi_n^1(X, Y, Z) \equiv \indicator(S_n(X, Y, Z) > C_n(X, Y, Z)); \quad \phi_n^2(X, Y, Z) \equiv \indicator(S_n(X, Y, Z) > z_{1-\alpha}). 
	\end{equation*}
	If the critical value of the first converges in probability to that of the second:
	\begin{equation}
		C_n(X, Y, Z) \convp z_{1-\alpha}
		\label{eq:convergence-of-critical-value}
	\end{equation}
	and the test statistic does not accumulate near the limiting critical value:
	\begin{equation*}
		\lim_{\delta \rightarrow 0}\limsup_{n \rightarrow \infty}\ \P[|S_n(X, Y, Z)-z_{1-\alpha}| \leq \delta] = 0,
		\label{eq:non-accumulation-app}
	\end{equation*}
	then the two tests are asymptotically equivalent:
	\begin{equation*}
		\lim_{n \rightarrow \infty}\P[\phi^1_n(X, Y, Z) = \phi^2_n(X, Y, Z)] = 1.
	\end{equation*}
\end{lemma}

\begin{lemma}[Conditional convergence implies quantile convergence] \label{lem:conditional-convergence-to-quantile} 
	Let $W_n$ be a sequence of random variables, $\mathcal F_n$ a sequence of $\sigma$-algebras, and $\alpha \in (0,1)$ a real number. If $W_n \mid \mathcal F_n \convdp W$ for some random variable $W$ whose CDF is continuous and strictly increasing at $\mathbb Q_{\alpha}[W]$, then
	\begin{equation*}
		\mathbb Q_{\alpha}[W_n \mid \mathcal F_n] \convp \mathbb Q_{\alpha}[W].
	\end{equation*}
\end{lemma}

\begin{lemma}[H\'ajek central limit theorem]\label{lem:hajek_clt} 
	Let $X_1, X_2, \dots, $ be an i.i.d.\ sequence of random variables such that $\E[X_{i}] = 0$ and $\E[X_{i}^2] = 1.$ Let $c_1, c_2, \dots$ be a (possibly dependent) sequence of random variables. Suppose that the $c_{i}$s satisfy the following conditions:
	\begin{enumerate}
		\item \textbf{Sum zero:} $\sum_{i=1}^n c_{i} = 0$ almost surely;
		\item \textbf{Sum one:} $\sum_{i=1}^n c_{i}^2 = 1$ almost surely;
		\item \textbf{Truncation condition:} $$\lim_{d\to\infty} \limsup_{n\to\infty} \sum_{i=1}^n c_{i}^2 \indicator (|\sqrt{n} c_{i}| \geq d)=0$$ almost surely;
		\item \textbf{Tail assumption:} $\max_{1\leq i\leq n}|X_{i}|/\sqrt{n}\convp 0$.
	\end{enumerate}
	Let $\mathcal{F}_n$ denote the $\sigma$-algebra generated by $X_1,\ldots, X_n$ and $c_1,\ldots, c_n$. Let $\pi$ denote a uniform draw from the permutation group of $[n]$ independent of $\mathcal{F}_n$. Then
	\begin{equation*}\label{eq_permuted_inner_product_1}
	\sum_{i=1}^n X_{\pi(i)} c_{i}\mid \mathcal{F}_n \convdp N(0,1).
	\end{equation*}
\end{lemma}
	
\begin{lemma}[Hoeffding's identity]\label{lem:hoeffding_identity}
	Let $c_1,\ldots,c_n$ and $X_1,\ldots,X_n$ be (possibly dependent) random variables. 	Let $\mathcal{F}_n$ denote the $\sigma$-algebra generated by $X_1,\ldots, X_n$ and $c_1,\ldots, c_n$. Suppose that $\pi$ is a uniform draw from the permutation group of $[n]$ independent of $\mathcal{F}_n$. Defining the random variable $S_n\equiv \sum_{i=1}^n c_i X_{\pi(i)}/n$, we have 
	\begin{align*}
		\E[S_n|\mathcal{F}_n]=\frac{1}{n}\sum_{i=1}^n c_i\frac{1}{n}\sum_{i=1}^n X_i
	\end{align*} 
	and 
	\begin{align*}
		\mathrm{Var}[S_n|\mathcal{F}_n]&=\frac{1}{n-1}\left(\frac{1}{n}\sum_{i=1}^n c_i^2-\left(\frac{1}{n}\sum_{i=1}^n c_i\right)^2\right)\cdot \left(\frac{1}{n}\sum_{i=1}^n X_i^2-\left(\frac{1}{n}\sum_{i=1}^n X_i\right)^2\right) \\ &\leq \frac{1}{n-1} \left( \frac{1}{n} \sum_{i=1}^n c_i^2 \right) \cdot \left( \frac{1}{n} \sum_{i=1}^n X_i^2 \right).
	\end{align*}
\end{lemma}

\begin{lemma}[Law of large numbers]\label{lem:SLLN}
	Let $X_1,\ldots, X_n$ be i.i.d.\ random variables such that $\E[|X_i|]<\infty$. Then
	\begin{align*}
		\lim_{n\to\infty} \frac{1}{n}\sum_{i=1}^n X_i = \E[X_i],\ \text{almost surely}.
	\end{align*}
	As a corollary, $(1/n)\sum_{i=1}^n X_i \xrightarrow{p} \E[X_i].$
\end{lemma}

\begin{lemma}[Conditional convergence in quadratic mean implies marginal convergence in probability]\label{lem:conditional_convergence}
	Let $X_n$ be a sequence of random variables and $\mathcal F_n$ a sequence of $\sigma$-algebras. Suppose the following conditions hold:
	\begin{enumerate}
		\item $\mathrm{Var}[X_n|\mathcal{F}_n]\convp 0$;
		\item $\E[X_n|\mathcal{F}_n]\convp 0$.
	\end{enumerate}
	Then $X_n\convp 0$.
\end{lemma}

\begin{proof}[Proof of Lemma \ref{lem:conditional_convergence}] An application of Markov's inequality shows that $X_n | \mathcal{F}_n \convpp 0$. Next, the result follows from Lemma \ref{lem:cond_conv_prob_marginal_conv_prob}.
	
\end{proof}

\section{Lemmas for the NB model}\label{sec:nb_lemmas}

Section \ref{sec:nb_lemma_statements} states several lemmas related to the NB regression model, and Section \ref{sec:nb_lemma_proofs} proves these Lemmas. Throughout, we assume that the data were generated by the NB GLM (\ref{eqn:nb_glm}) and that Assumptions \ref{assu:consistency_beta}--\ref{assu:lower_bound_eigenvalue} hold. We do not assume that the size parameter is correctly specified; in other words, we allow $\bar \theta \neq \theta$.

\subsection{Lemmas}\label{sec:nb_lemma_statements}

\begin{lemma}[Bounded moments of $Y_i$, $\xi_i$, and $W_i$]\label{lem:NB_model_moments_bound}
	Let $m \in \mathbb{N}$ be given. We have $\E[Y_{i}^m] < \infty, \E[\xi_{i}^m] < \infty,$ and $\E[W_i^m] < \infty$.
\end{lemma}
	
	\begin{lemma}[Invertibility of information matrix]\label{lem:invert_info_matrix}
		The matrix $\E\left[ W_i Z_i Z_i^\top  \right]$ is invertible.
	\end{lemma}
	
	\begin{lemma}[Uniform bounds on $f_{y,\bar \theta}$ and $f'_{y,\bar \theta}$ ]\label{lem:glm_f_bound} Define the function $f_{y, \bar{\theta}} : \R \to \R$ by
		$$f_{y, \bar{\theta}}(x)\equiv \frac{y-e^x}{1+ e^x/\bar{\theta}},$$ where $y \geq 0$ and $\bar \theta > 0$ are constants. We have
		$$\sup_{x \in \R} |f_{y, \bar \theta}(x)| \leq y + \bar \theta; \quad \sup_{x \in \R} |f'_{y, \bar{\theta}}(x)| \leq y +\bar{\theta}.$$
	\end{lemma}
	
	\begin{lemma}[Bound on $| \hat{r}_i^m - \xi_i^m|$ and $|\hat{W}_i^m - W_i^m|$]\label{lem:xi_r_diff}
		For all $i \in [n],$ $m \in \{1, 2, 4\},$ and $n \in \mathbb{N}$,
		$$ \Big| f_{Y_i, \bar \theta}^m( \hat{\beta}_n^\top Z_i) - f_{Y_i, \bar \theta}^m(\beta^\top Z_i) \Big| \leq m C (Y_i + \bar{\theta})^m ||\hat{\beta}_n - \beta||.$$ As corollaries,
		$$ | \hat{r}_i^m - \xi_i^m| \leq m C (Y_i + \bar{\theta})^m ||\hat{\beta}_n - \beta||; \quad |\hat{W}^m_i - W^m_i| \leq m C \bar{\theta}^m ||\hat{\beta}_n - \beta||.$$
	\end{lemma}
	
	\begin{lemma}[Consistency of observed information matrix]\label{lem:inverse_matrix_convergence}
		We have
		$$ \frac{1}{n} Z^\top  \hat{W} Z \xrightarrow{p} \E\left[ W_i Z_i Z_i^\top \right].$$ 
		As a corollary,
		$$ \left(\frac{1}{n} Z^\top  \hat{W} Z\right)^{-1} \xrightarrow{p} \left( \E\left[ W_i Z_i Z_i^\top  \right]  \right)^{-1},$$ where convergence in probability of a sequence of random matrices is defined as (\ref{eqn:def_matrix_conv}).
	\end{lemma}

	\begin{lemma}[Consistency of quadratic form]\label{lem:unconditional_convergence_term2_1}
		We have 
		\begin{align*}
			\frac{1}{n}X^\top\hat WX\convp \E[X_i^2W_i].
		\end{align*}
	\end{lemma}
	
	\begin{lemma}[Consistency of $(1/n) X^\top \hat{W} Z$]\label{lem:unconditional_convergence_term2_2}
		We have 
		\begin{align*}
			\frac{1}{n}X^\top \hat{W} Z \convp \E[X_i W_iZ_i],
		\end{align*}
		where convergence in probability of a sequence of random vectors is defined as (\ref{eqn:def_vector_conv}).
	\end{lemma}

	\begin{lemma}[Consistency of observed information matrix evaluated at $\tilde{\beta}_n$]\label{lem:consistency_tilde_info_matrix} Let $\tilde{\beta}_n$ be a vector between $\hat\beta_n$ and $\beta$, i.e. $\tilde\beta_n = t\beta + (1-t)\hat{\beta}_n$ for some $t \in (0,1)$, almost surely. Define $$\tilde W_i = \frac{\tilde \mu_i}{1 + \tilde \mu_i / \bar \theta}; \quad \tilde \mu_i = \exp\left(\tilde \beta_n^\top Z_i \right); \quad \tilde W = \text{diag}\{\tilde W_1 ,\dots, \tilde W_n\}.$$ We have
		$$ \frac{1}{n} Z^\top  \tilde{W} Z \convp \E\left[ W_i Z_i Z_i^\top  \right].$$ As a corollary,
		$$ \left( \frac{1}{n} Z^\top  \tilde W Z \right)^{-1} \convp \left( \E\left[ W_i Z_i Z_i^\top  \right]  \right)^{-1}.$$
	\end{lemma}
	
	\begin{lemma}[Convergence of $(1/n)Z^\top  \tilde{W} \tilde{R} Z$ to $0$]\label{lem:B_matrix_convergence} Let $\tilde W$ and $\tilde{\mu}_i$ be defined as in Lemma \ref{lem:consistency_tilde_info_matrix}. Also, define $$\tilde r_i = \frac{Y_i - \tilde \mu_i}{ 1 + \tilde \mu_i/\bar \theta}; \quad \tilde{R} = \text{diag}\{\tilde r_1, \dots, \tilde r_n\}.$$
		We have
		$$ \frac{1}{n} Z^\top  \tilde{W}\tilde{R} Z \convp 0.$$
	\end{lemma}
	
	\begin{lemma}[Asymptotic linearity of $\hat{\beta}_n$]\label{lem:consistency_beta} We have
		\begin{align*}
			\sqrt{n}(\hat\beta_n-\beta)=\left(\E[W_iZ_iZ_i^\top ]\right)^{-1}\frac{1}{\sqrt{n}}\sum_{i=1}^n Z_i\frac{Y_i-\mu_i}{1+\mu_i/\bar \theta}+o_p(1).
		\end{align*}
	\end{lemma}
	
	\begin{lemma}[Consistency of $\hat{r}_i^m$]\label{lem:second_moment_r_lower_bound} 
		For $m \in \{1, 2, 4\}$, we have
		$$
		\lim_{n\rightarrow\infty}\frac{1}{n}\sum_{i=1}^n \hat r_{i}^m = \E[\xi_{i}^m], \text{almost surely.}
		$$
		As corollaries,
		$$
			\liminf_{n\rightarrow\infty}\frac{1}{n}\sum_{i=1}^n \hat r_{i}^2>0,  \text{almost surely}; \quad \limsup_{n\to\infty} \frac{1}{n} \sum_{i=1}^n \hat{r}_i^4 < \infty,  \text{almost surely}.
			$$
	\end{lemma}
	
	\begin{lemma}[Truncation condition]\label{lem:truncation_condition_Hajek_CLT}
		We have 
		\begin{align*}
			\limsup_{d\rightarrow\infty}\limsup_{n\rightarrow\infty}\frac{1}{n}\sum_{i=1}^n \indicator\left(\left|\frac{\hat r_{i}}{\sqrt{\sum_{i=1}^n \hat r_{i}^2/n}}\right|\geq d\right)=0,\ \text{almost surely}.
		\end{align*}
	\end{lemma}

	\begin{lemma}[Consistency of $\hat{W}_i$ and $\hat{W}_i^2$]\label{lem:consistency_W_i}
		We have 
		\begin{align}\label{eq:convergence_W}
			\frac{1}{n}\sum_{i=1}^n \hat W_{i} \convp \E[W_i]; \quad \frac{1}{n}\sum_{i=1}^n \hat W_{i}^2 \convp \E[W_i^2].
		\end{align}
		As a corollary, for given $j \in [p]$, we have
		\begin{align}\label{eq:stochastic_boundedness_WZ}
			 \frac{1}{n}\sum_{i=1}^n \hat W_i Z_{ij}=O_p(1);\quad\frac{1}{n}\sum_{i=1}^n \hat W_i^2Z_{ij}^2=O_p(1).
		\end{align}
	\end{lemma}

\subsection{Proofs of Lemmas}\label{sec:nb_lemma_proofs}

\textbf{Proof of Lemma \ref{lem:NB_model_moments_bound}.} Let $X$ be a negative binomial random variable with mean $\mu$ and size parameter $\theta$. We can express the $m$th moment of $X$ as a weighted sum of the powers of $\mu$ from $1$ to $m$. That is,
$$\E[X^m] = \sum_{k=1}^m l_k \mu^k,$$ where $l_1, \dots, l_m \in \R$ are combinatorial constants that depend on $\theta$. (See \cite{Simsek2024} for a proof of this statement and an explicit formula for the $l_k$s.) Turning our attention to $Y_i$, by the law of total expectation,
$$
\E[Y_i^m] = \E\left(\E[Y_i^m | Z_i \right]) =\sum_{k=1}^m \E\left[l_k \left(e^{\beta^\top Z_i}\right)^k \right] \leq \sum_{k=1}^m l_ke^{kC||\beta||} < \infty.
$$
Next, because $e^{\beta^\top Z_i}/ \bar \theta \geq 0,$ we have that
\begin{equation}\label{eq_lemma_nb_bounded_moments_1}
	\left(1 + e^{\beta^\top Z_i}/\bar \theta\right)^m \geq 1.    
\end{equation}
On the other hand, by the binomial theorem and the boundedness of $\E[Y_i^m]$, we observe
\begin{multline}\label{eq_lemma_nb_bounded_moments_2}
	\E\left[\left(Y_i - e^{\beta^\top  Z_i} \right)^m \right] \leq \sum_{k=0}^m \binom{m}{k}\E\left[Y_i^{m-k} e^{k\beta^\top Z_i}\right] \\ \leq \sum_{k=0}^m \binom{m}{k} \sqrt{\E\left[Y_i^{2m-2k}\right] e^{2kC||\beta||}} < \infty.
\end{multline}
Combining (\ref{eq_lemma_nb_bounded_moments_1}) and (\ref{eq_lemma_nb_bounded_moments_2}) yields
$$\E[\xi_i^m] = \E \left[\frac{(Y_i - e^{\beta^\top Z_i})^m}{(1 + e^{\beta^\top Z_i}/ \bar \theta)^m}\right] \leq \E\left[(Y_i - e^{\beta^\top Z_i})^m\right] < \infty.$$ Finally, we see
$$ \E\left[ W_i^m  \right] = \E \left[ \frac{e^{m\beta^\top Z_i }}{(1 + e^{\beta^\top  Z_i}/\bar \theta)^m} \right] \leq \E \left[ e^{m \beta^\top  Z_i} \right] \leq e^{mC|| \beta||} < \infty.$$
$\square$

\textbf{Proof of Lemma \ref{lem:invert_info_matrix}.}
By Assumption \ref{assu:lower_bound_eigenvalue}, the minimum eigenvalue of $\E\left[Z_i Z_i^\top  \right]$ is positive:
$$\lambda_\textrm{min}\left[\E\left[Z_i Z_i^\top \right]\right] =  \inf_{a \in \R^p, ||a|| = 1}  a^\top  \E[Z_i Z_i^\top ] a = \inf_{a \in \R^p, ||a|| = 1} \E\left[ (a^\top  Z_i)^2 \right] > 0.$$ Additionally, the random variable $W_i$ is almost surely greater than or equal to some positive constant $\delta$:
$$W_i = \frac{\exp( \beta^\top  Z_i)}{1 + \exp(\beta^\top  Z_i)/\bar \theta} \geq \frac{\exp(-C ||\beta||)}{ 1 + \exp(C ||\beta||)/\bar \theta }  \geq \delta > 0.$$ Thus, the minimum eigenvalue of $\E\left[ W_i Z_i Z_i^\top \right]$ also is positive:
$$\lambda_\textrm{min}\left[ \E\left[ W_i  Z_i Z_i^\top  \right] \right] =  \inf_{a \in \R^p, ||a|| = 1} \E \left[ W_i (a^\top  Z_i)^2 \right] \geq \inf_{a \in \R^p, ||a|| = 1} \delta \E\left[ (a^\top  Z_i)^2 \right] > 0.$$
Hence, $\E\left[ W_i Z_i Z_i^\top  \right]$ is invertible.
$\square$

\textbf{Proof of Lemma \ref{lem:glm_f_bound}.} We separate proofs of the two bounds.

\underline{Proof of the bound on $f_{y, \bar \theta}(x)$:}
The derivative $f'_{y,\bar \theta}$ of $f_{y,\bar \theta}$ is
\begin{align}\label{f_prime}
	f'_{y,\bar \theta}(x) = \frac{(-e^x)(1+e^x/\bar \theta) - (y - e^x)(e^x/\bar \theta)}{(1 + e^x/\bar \theta)^2} = \frac{-e^x(1 + y/\bar \theta)}{(1+e^x/\bar \theta)^2} = \frac{-\bar \theta e^x (\bar \theta + y)}{(\bar \theta + e^x)^2}. 
\end{align}
The function $f'_{y,\bar \theta}$ is strictly negative, as its numerator is strictly negative and its denominator is strictly positive. This implies that the function $f_{y,\bar \theta}$ is strictly decreasing. Next,
\begin{align*}
	\lim_{x\to-\infty} f_{y,\bar \theta}(x) = \frac{y-0}{1+0} = y,\ \lim_{x\to\infty} f_{y,\bar \theta} = \lim_{x\to\infty} \frac{ye^{-x} - 1}{e^{-x} + 1/\bar \theta} = \frac{0-1}{0+1/\bar \theta} = -\bar \theta.
\end{align*}
We conclude from the above statements that
$-\bar \theta \leq f_{y,\bar \theta}(x) \leq y$ for all $x \in \R.$ Finally,
$$ \sup_{x \in \R} | f_{y, \bar{\theta}}(x)| \leq \max\{ \bar \theta, y \} \leq \bar \theta + y.$$	
\underline{Proof of the bound on $f'_{y,\bar \theta}(x)$:} The derivative $f''_{y,\bar \theta}$ of $f'_{y,\bar \theta}$ is
\begin{align*}
	f''_{y,\bar \theta}(x) 
	&
	= \frac{-\bar \theta e^x (\bar \theta + y) (\bar \theta + e^x)^2 + \bar \theta e^x(\bar \theta + y)2(\bar \theta+e^x)e^x}{(\bar \theta + e^x)^4} \\ 
	&
	= \frac{-\bar \theta e^x (\bar \theta + y)(\bar \theta + e^x) + \bar \theta e^x (\bar \theta+y)2e^x}{(\bar \theta+e^x)^3} = \frac{\bar \theta e^x(\bar \theta + y)(e^x - \bar \theta)}{(\bar \theta + e^x)^3}.
\end{align*}
We find the critical value(s) of $f'_{y,\bar \theta}$ by setting $f''_{y,\bar \theta}$ to zero and solving for $x$:
$$
\frac{\bar \theta e^x (\bar \theta + y)(e^x - \bar \theta)}{(\bar \theta + e^x)^3} = 0 \iff e^x - \bar \theta = 0 \iff x = \log(\bar \theta).
$$
Thus, $f'_{y,\bar \theta}$ has a single critical value at $x = \log(\bar \theta)$, implying that the extreme values of $f'_{y,\bar \theta}$ occur at $x=\log(\bar \theta), x=-\infty,$ or $x=\infty$. We study these three cases one-by-one. First,
$$f'_{y,\bar \theta}(\log(\bar \theta)) = \frac{-\bar \theta (\bar \theta)(\bar \theta + y)}{ (\bar \theta + \bar \theta)^2}=-(1/4)(\bar \theta + y).$$
Next, observe that we can write
$$f'_{y,\bar \theta}(x) = \frac{-\bar \theta e^x(\bar \theta + y)}{\bar \theta^2 + 2\bar \theta e^x + e^{2x}} = \frac{-\bar \theta(\bar \theta + y)}{\bar \theta^2 e^{-x} + 2\bar \theta + e^x}.$$ We see that
$$ \lim_{x\to-\infty} f_{y,\bar \theta}'(x) = \lim_{x\to\infty}f_{y, \bar \theta}'(x) = 0.$$ It follows that
$-(1/4)(\bar  \theta + y) \leq f'_{y,\bar  \theta}(x) \leq 0$ for all $x \in \R$, or
$$|f'_{y,\bar  \theta}(x)| \leq (1/4)(\bar \theta+y) \leq \bar \theta + y$$ for all $x \in \R$.	
$\square$

\textbf{Proof of Lemma  \ref{lem:xi_r_diff}.} $f_{Y_{i}, \bar \theta}$ is everywhere differentiable. Thus, by the mean value theorem, there exists some $\gamma$ between $\hat{\beta}^\top_n Z_{i}$ and $\beta^\top Z_{i}$ such that
\begin{align*}
	|f_{Y_{i},\bar \theta}(\hat{\beta}^\top_n Z_{i}) - f_{Y_{i},\bar \theta}(\beta^\top _n Z_{i})| \leq |f_{Y_i,\bar \theta}'(\gamma)| \cdot |\hat{\beta}_n^\top Z_{i} - \beta^\top Z_{i}|.
\end{align*}
But by Lemma \ref{lem:glm_f_bound}, $|f_{Y_{i},\bar \theta}'(\gamma)| \leq Y_{i} + \bar \theta$ uniformly over $Z_{i},$ $\hat{\beta}_n$, and $\beta$.
Hence,
\begin{align*}
	|\hat{r}_{i} - \xi_{i}| = |f_{Y_{i},\bar \theta}(\hat{\beta}_n^\top Z_{i}) - f_{Y_{i}, \bar \theta}(\beta^\top Z_{i}) | \leq C (Y_{i}+\bar \theta)\|\hat{\beta}_n - \beta\|,
\end{align*}
where we have applied the Cauchy-Schwarz inequality. Next, by the above result and Lemma \ref{lem:glm_f_bound}, 
\begin{align*}
	|f_{Y_{i},\bar \theta}(\hat{\beta}_n^\top Z_{i})^2 &- f_{Y_{i},\bar \theta}(\beta^\top Z_{i})^2| \\ 
	&= |f_{Y_{i},\bar \theta}(\hat{\beta}_n^\top Z_{i}) - f_{Y_{i},\bar \theta}(\beta^\top Z_{i})| \cdot |f_{Y_i,\bar \theta}(\hat{\beta}_n^\top Z_i) + f_{Y_i,\bar \theta}(\beta^\top Z_i)| \\ 
	&\leq C(Y_{i} + \bar \theta)\|\hat{\beta}_n - \beta\| \left(|f_{Y_{i},\bar \theta}(\hat{\beta}_n^\top Z_{i})| + |f_{Y_{i},\bar \theta}(\beta^\top Z_i)|\right) \\ 
	&\leq C(Y_{i} + \bar \theta)\|\hat{\beta}_n - \beta\| \left(2 \sup_{x \in \R} |f_{Y_{i}, \bar \theta}(x) | \right) \leq 2C(Y_{i} + \bar \theta)^2\|\hat{\beta}_n - \beta\|. 
\end{align*}
Applying similar reasoning,
\begin{align*}
	|f_{Y_{i},\bar \theta}^4(\hat{\beta}_n^\top Z_{i}) &- f_{Y_{i},\bar \theta}^4(\beta^\top Z_{i})|\\
	&
	=|f_{Y_{i},\bar \theta}^2(\hat{\beta}_n^\top Z_{i}) - f_{Y_{i},\bar \theta}^2(\beta^\top Z_{i})| \cdot |f_{Y_{i},\bar \theta}^2(\hat{\beta}_n^\top Z_{i}) + f_{Y_{i},\bar \theta}^2(\beta^\top Z_{i})|\\
	&
	\leq 2C(Y_{i} + \bar \theta)^2\|\hat{\beta}_n - \beta\| \left(2 \sup_{x \in \R} |f_{Y_{i}, \bar \theta}(x) |^2 \right) \leq 4C(Y_{i} + \bar \theta)^4\|\hat{\beta}_n - \beta\|.
\end{align*}
As a corollary, for $m \in \{1, 2, 4\}$,
$$
|\hat{r}_{i}^m - \xi_{i}^m| = |f_{Y_{i},\bar \theta}^m(\hat{\beta}_n^\top Z_{i}) - f_{Y_{i},\bar \theta}^m(\beta^\top Z_{i})| \leq mC(Y_{i} + \bar \theta)^m\|\hat{\beta}_n - \beta\|
$$
and 
$$| \hat{W}_i^m - W_i^m| =  |f_{0,\bar \theta}^m(\beta^\top Z_{i}) - f_{0, \bar \theta}^m(\hat{\beta}_n^\top Z_{i})| \leq mC\bar \theta^m\|\hat{\beta}_n - \beta\|.$$ $\square$

\textbf{Proof of Lemma \ref{lem:inverse_matrix_convergence}.}  We show that the matrix $(1/n) Z^\top  \hat{W} Z$ converges in probability elementwise to $\E[W_i Z_i Z_i^\top ]$. Let $j, k \in \{ 1, \dots, p \}$ be given. By the triangle inequality, the boundedness of $Z_i$, and Lemma \ref{lem:xi_r_diff},
\begin{align*}
	\Bigg| \frac{1}{n} \sum_{i=1}^n Z_{ij} Z_{ik} \hat{W}_i - \frac{1}{n} \sum_{i=1}^n Z_{ij} Z_{ik} W_i \Bigg| &\leq 	
	\frac{1}{n} \sum_{i=1}^n |Z_{ij} Z_{ik}| |\hat{W}_i - W_i | \\
	&\leq \frac{C^2}{n} \sum_{i=1}^n | \hat{W}_i - W_i| \\
	&\leq C^3 \bar \theta ||\hat{\beta}_n - \beta||^2 \xrightarrow{p} 0.
\end{align*}
Next, by LLN (Lemma \ref{lem:SLLN}), the boundedness of $\E[W_i^2]$ (Lemma \ref{lem:NB_model_moments_bound}), and the boundedness of $Z_i$ (Assumption \ref{assu:boundedness_covariate}),  $$(1/n)\sum_{i=1}^n W_i Z_{ij} Z_{ik}  \xrightarrow{p} \E[W_i Z_{ij} Z_{ik}].$$ Hence, by the triangle inequality, and because $j,k \in [p]$ were chosen arbitrarily,
$$(1/n) Z^\top  \hat{W}Z \xrightarrow{p} \E[W_i  Z_i Z_i^\top ].$$ Finally, because matrix inverse is a continuous operation and the limit matrix is invertible (Lemma \ref{lem:invert_info_matrix}), the result holds by the continuous mapping theorem. $\square$

\textbf{Proof of Lemma  \ref{lem:unconditional_convergence_term2_1}}. Applying Lemma \ref{lem:xi_r_diff}, the boundedness of $\E[X_i^2]$ (Assumption \ref{assu:moment_treatment}), and LLN (Lemma \ref{lem:SLLN}), 
\begin{align*}
	\Bigg| \frac{1}{n} \sum_{i=1}^n X_i^2 \hat{W}_i - \frac{1}{n} \sum_{i=1}^n X_i^2 W_i \Bigg| \\ 
	&\leq  \frac{1}{n} \sum_{i=1}^n X_i^2 |\hat{W}_i - W_i| \\ 
	&\leq \frac{1}{n} \sum_{i=1}^n X_i^2 C \bar \theta || \hat{\beta}_n - \beta || = O_p(1) ||\hat{\beta}_n - \beta|| \convp 0. 
\end{align*}
Next, by LLN (Lemma \ref{lem:SLLN}), the boundedness of $\E[W_i^2]$ (Lemma \ref{lem:NB_model_moments_bound}), and the boundedness of $\E[X_i^4]$ (Assumption \ref{assu:moment_treatment}), $$(1/n) \sum_{i=1}^n X_i^2 W_i \xrightarrow{p} \E[X_i^2 W_i].$$ Finally, by the triangle inequality, $(1/n) \sum_{i=1}^n X_i^2 \hat{W}_i \to \E[X_i^2 W_i].$
$\square$

\textbf{Proof of Lemma \ref{lem:unconditional_convergence_term2_2}}. Let $j \in [p]$ be given. Applying Lemma \ref{lem:xi_r_diff}, the boundedness of $\E[|X_i|]$ (Assumption \ref{assu:moment_treatment}), and LLN (Lemma \ref{lem:SLLN}), we have
\begin{align*}
	\Bigg| \frac{1}{n} \sum_{i=1}^n X_i \hat{W}_i Z_{ij} - \frac{1}{n} \sum_{i=1}^n X_i &W_i Z_{ij} \Bigg| \\ 
	&\leq \left| \frac{1}{n} \sum_{i=1}^n (\hat{W}_i - W_i) X_i Z_{ij } \right|
	\\ & \leq \frac{1}{n} \sum_{i=1}^n C^2 \bar \theta||\hat{\beta}_n - \beta|| \cdot |X_i| = ||\hat{\beta}_n - \beta|| O_p(1) \convp 0.
\end{align*}
Next, by LLN (Lemma \ref{lem:SLLN}), the boundedness of $\E[X_i^2]$ (Assumption \ref{assu:moment_treatment}) and the boundedness of $\E[W_i^m]$ (Lemma \ref{lem:NB_model_moments_bound}), $(1/n) \sum_{i=1}^n X_i W_i Z_{ij} \convp \E[X_i W_i Z_{ij}].$ By triangle inequality, $(1/n) \sum_{i=1}^n X_i \hat{W}_i Z_{ij} \convp \E[X_i W_i Z_{ij}].$ Finally, the result holds because $j$ was chosen arbitrarily. $\square$

\textbf{Proof of Lemma \ref{lem:consistency_tilde_info_matrix}}. Because $\tilde{\beta}_n$ is between $\beta$ and $\hat{\beta}_n$, we have that $||\tilde{\beta}_n - \beta|| \convp 0$ by Assumption \ref{assu:consistency_beta}. The proof of Lemma \ref{lem:consistency_tilde_info_matrix} is exactly the same as the proof of Lemma \ref{lem:inverse_matrix_convergence}, but we replace $\hat{\beta}_n$ with $\tilde{\beta}_n$. $\square$

\textbf{Proof of Lemma \ref{lem:B_matrix_convergence}}. Let $j, k \in [p]$ be given. For real numbers $a, b, c, d, z \in \R$, the following inequality holds: $|zab - zcd| \leq |z|\cdot|a - c|\cdot |b| + |z|\cdot|b - d|\cdot|c|.$ Applying this inequality, the boundedness of $\E[Y_i]$, and LLN (Lemma \ref{lem:SLLN}),
\begin{align}\label{eqn:B_matrix_convergence_1}
	\Bigg| \frac{1}{n} \sum_{i=1}^n & \tilde{W}_i \tilde{r}_i Z_{ij} Z_{ik} - \frac{1}{n} \sum_{i=1}^n W_i \xi_i Z_{ij} Z_{ik} \Bigg|  \notag \\
	&\leq \frac{1}{n} \sum_{i=1}^n \left| Z_{ij} Z_{ik}  \right|  \cdot \left| \tilde{r}_i - \xi_i \right|  \cdot |\tilde{W}_i| + \frac{1}{n} \sum_{i=1}^n |Z_{ij} Z_{ik}| \cdot \left| \tilde{W}_i - W_i \right| \cdot |\xi_i| \notag \\ 
	&\leq \frac{1}{n} \sum_{i=1}^n C^2 \cdot C(\bar{\theta} + Y_i) || \tilde{\beta}_n - \beta || \cdot \bar \theta + \frac{1}{n} \sum_{i=1}^n C^2 \cdot  C \bar \theta ||\tilde{\beta}_n -\beta|| \cdot (Y_i + \bar \theta) \notag \\ 
	&= 2 C^3 \cdot ||\tilde{\beta}_n - \beta|| \cdot \bar \theta \cdot \frac{1}{n} \sum_{i=1}^n \left( Y_i + \bar{\theta} \right) = ||\tilde{\beta}_n - \beta|| O_p(1) \convp 0.
\end{align}
The second inequality above bounds $|\tilde{r}_i - \xi_i|$ and $| \tilde{W}_i - W_i|$ via Lemma \ref{lem:xi_r_diff} and $|\tilde{W}_i|$ and $|\xi_i|$ via Lemma \ref{lem:glm_f_bound}. Next, by LLN (Lemma \ref{lem:SLLN}) and the boundedness of $\E[W_i^m]$ and $\E[\xi^m_i]$ (Lemma \ref{lem:NB_model_moments_bound}),
\begin{equation}\label{eqn:B_matrix_convergence_2}
	\frac{1}{n}\sum_{i=1}^n W_i \xi_i Z_{ij} Z_{ik} \convp \E \left[ W_i \xi_i Z_{ij} Z_{ik} \right].
\end{equation}
Finally, 
\begin{equation}\label{eqn:B_matrix_convergence_3}
	\E\left[W_i \xi_i Z_{ij} Z_{ik} \right] = \E\left[\E\left[\frac{(Y_i - \mu_i)\mu_i}{(1 + \mu_i^2/\bar \theta)^2} Z_{ij} Z_{ik} \Big| Z_i \right] \right] = 0.
\end{equation}
Combining (\ref{eqn:B_matrix_convergence_1}), (\ref{eqn:B_matrix_convergence_2}), and (\ref{eqn:B_matrix_convergence_3}) via the triangle inequality, we obtain the result. $\square$

\textbf{Proof of Lemma \ref{lem:consistency_beta}}. 
Applying the multivariate mean value theorem to the score vector $U_{\bar \theta}$, we see
\begin{align*}
	U_{\bar\theta}(\hat\beta_n)-U_{\bar\theta}(\beta)=\nabla U_{\bar\theta}(\tilde \beta_n)(\hat\beta_n-\beta),
\end{align*}
where $\nabla U_{\bar \theta}$ denotes the Jacobian of the score vector, and $\tilde{\beta}_n$ is a vector between $\hat\beta_n$ and $\beta$ (i.e. $\tilde\beta_n = t\beta + (1-t)\hat{\beta}_n$ for some $t \in (0,1)$.) Because $U_{\bar \theta}(\hat{\beta}_n)=0,$ we can simplify the above as
\begin{align}\label{eq:u_taylor_expansion}
	U_{\bar\theta}(\beta)=-\nabla U_{\bar\theta}(\tilde \beta_n)(\hat\beta_n-\beta).
\end{align}
First, we establish the convergence
\begin{equation}\label{eqn:asy_invert_u_tilde_beta}
	\nabla U_{\bar \theta} (\tilde \beta_n) \convp -\E[W_i Z_i Z_i^\top].
\end{equation}
We can compute $\nabla U_{\bar \theta}(\beta)$ as follows:
$$ \nabla U_{\bar \theta}(\beta) = -\frac{1}{n}  \sum_{i=1}^n \frac{\mu_i}{1 + \mu_i/\bar \theta} Z_iZ_i^\top  - \frac{1}{n} \sum_{i=1}^n \frac{(Y_i - \mu_i) \mu_i/\bar \theta}{(1 + \mu_i/\bar \theta)^2} Z_i Z_i^\top .$$
Define 
\begin{align*}
	A_{n,\bar\theta}(\beta) &\equiv - \frac{1}{n}\sum_{i=1}^n \frac{\mu_i}{1+\mu_i/\bar\theta}Z_iZ_i^\top = -\frac{1}{n} \sum_{i=1}^n  W_i Z_iZ_i^\top  ,\\ 
	B_{n,\bar\theta}(\beta)&\equiv -\frac{1}{n}\sum_{i=1}^n \frac{(Y_i - \mu_i)\mu_i}{(1+ \mu_i^2/\bar\theta)^2}Z_iZ_i^\top = -\frac{1}{n} \sum_{i=1}^n \xi_i W_i Z_i Z_i^\top .
\end{align*}
We can express $\nabla U_{\bar \theta}(\beta)$ as $\nabla U_{\bar \theta} (\beta) = A_{n, \bar \theta}(\beta) + B_{n, \bar \theta}(\beta)/\bar \theta.$ Lemmas \ref{lem:consistency_tilde_info_matrix} and \ref{lem:B_matrix_convergence} imply that
$$A_{n, \bar \theta}(\tilde \beta_n) = - \frac{1}{n} Z^\top  \tilde{W} Z  \convp -\E\left[W_i Z_i Z_i^\top  \right]$$ and
$$B_{n, \bar \theta}(\tilde \beta_n) =  -\frac{1}{n} Z^\top  \tilde{R} \tilde{W} Z \convp 0,$$ respectively. Hence, (\ref{eqn:asy_invert_u_tilde_beta}) is proven. By the invertibility of $\E[W_i Z_i Z_i^\top ]$ (Lemma \ref{lem:invert_info_matrix}), the continuous mapping theorem, and the continuity of matrix inverse, $$\left(-\nabla U_{\bar \theta} (\tilde \beta_n)\right)^{-1} - \left(\E[W_i Z_i Z_i^\top ]\right)^{-1} = o_p(1).$$
Having established the asymptotic invertibility of $-\nabla U_{\bar \theta}(\tilde \beta_n)$, we can rewrite (\ref{eq:u_taylor_expansion}) as
$$ \left(- \nabla U_{\bar \theta} (\tilde{\beta}_n) \right)^{-1} \sqrt{n} U_{\bar \theta}(\beta) = \sqrt{n}\left( \hat{\beta}_n - \beta \right)$$ with high probability. The multivariate central limit theorem implies
$\sqrt{n}U_{\bar \theta} (\beta) = O_p(1).$ Thus,
\begin{multline*}
	\left(-\nabla U_{\bar \theta} (\tilde{\beta}_n) \right)^{-1} \sqrt{n} U_{\bar \theta}(\beta) = \left(\left( \E[W_i Z_i Z_i^\top ]\right)^{-1} + o_p(1) \right) \sqrt{n} U_{\bar \theta}(\beta) \\ = \left(\E[W_i Z_i Z_i^\top ]\right)^{-1} \sqrt{n} U_{\bar \theta}(\beta) + o_p(1) O_p(1),
\end{multline*}
from which we conclude 
$$\sqrt{n}(\hat\beta_n-\beta)=\left(\E[W_iZ_iZ_i^\top ]\right)^{-1}\frac{1}{\sqrt{n}}\sum_{i=1}^n Z_i\frac{Y_i-\mu_i}{1+\mu_i/\bar \theta}+o_p(1).$$
$\square$

\textbf{Proof of Lemma \ref{lem:second_moment_r_lower_bound}}. 
First, we prove that, for $m \in \{2, 4\}$,
\begin{align*}
	\lim_{n\rightarrow\infty}\left|\frac{1}{n}\sum_{i=1}^n \hat r_{i}^m-\frac{1}{n}\sum_{i=1}^n \xi_{i}^m\right|=0,\ \text{almost surely.}
\end{align*}
By Lemma \ref{lem:xi_r_diff} and the boundedness of $\E[Y_i^m]$ (Lemma \ref{lem:NB_model_moments_bound}), we have
\begin{multline*}
	\Bigg|\frac{1}{n} \sum_{i=1}^n \hat{r}_i^m - \frac{1}{n} \sum_{i=1}^n  \xi_i^m  \Bigg| 
	\leq \frac{1}{n} \sum_{i=1}^n | \hat{r}_i^m - \xi_i^m|
	\leq \frac{1}{n} \sum_{i=1}^n 2m(Y_i - \bar{\theta})^m|| \hat{\beta}_n - \beta || \\
	\\ = O_p(1) ||\hat{\beta}_n - \beta|| \to 0, \text{almost surely}.
\end{multline*}
Next, by the boundedness of $\E\left[\xi_i^m\right]$ (Lemma \ref{lem:NB_model_moments_bound}), it follows that $(1/n) \sum_{i=1}^n \xi_i^m  \to \E[\xi_i^m],$ almost surely. Hence, applying the triangle inequality once more, we conclude that $(1/n) \sum_{i=1}^n \hat{r}_i^m \to  \E[\xi_i^m],$ almost surely. Next, we establish that $\E[\xi_{i}^2]>0$. Observe 
\begin{align*}
	\E[\xi_{i}^2]&=\E\left[\frac{(Y_{i}-\mu_{i})^2}{(1+\mu_{i}/\bar\theta)^2}\right]
	\geq \E\left[(Y_{i}-\mu_{i})^2\right] = \E\left[ \E \left[ (Y_i - \mu_i)^2 | Z_i  \right] \right] = \E \left[\textrm{Var}(Y_i | Z_i)  \right] \\
	&=\E[\mu_i+\mu_i^2/\theta]
	\geq \E[\mu_i]=\E[\exp(\beta^\top Z_i)]\geq\exp(-\|\beta\|C) > 0.
\end{align*}
As a corollary,
$$ \liminf_{n\to\infty} \frac{1}{n} \sum_{i=1}^n \hat{r}_i^2 > 0, \text{almost surely}.$$ Finally, the boundedness  of $\E[\xi_i^4]$ implies the corollary
$$ \limsup_{n\to\infty} \frac{1}{n} \sum_{i=1}^n \hat{r}_i^4 < \infty, \text{almost surely}.$$
$\square$

\textbf{Proof of Lemma \ref{lem:truncation_condition_Hajek_CLT}.} Consider the sequence of events $\mathcal{E}_1, \mathcal{E}_2, \dots$ defined by
\begin{align*}
	\mathcal{E}_n\equiv \left\{\frac{1}{n}\sum_{i=1}^n \hat r_{i}^2\geq  0.9\E[\xi_{i}^2]\right\}.
\end{align*}
For arbitrary events $A$ and $B$, the following inequality holds:
\begin{equation}\label{eqn:set_inequality}
\indicator(A) \leq \indicator(A \cap B) + \indicator(B^c).
\end{equation}
Applying this inequality,
\begin{align}\label{eq:truncation_upper_bound_1}
	\frac{1}{n}\sum_{i=1}^n \indicator&\left[\left|\frac{\hat r_{i}}{\sqrt{\sum_{i=1}^n \hat r_{i}^2/n}}\right|\geq d\right] \nonumber \\	
	&\leq \frac{1}{n}\sum_{i=1}^n \indicator\left[\left|\frac{\hat r_{i}}{\sqrt{\sum_{i=1}^n \hat r_{i}^2/n}}\right|\geq d, \mathcal{E}_n\right]+\indicator[\mathcal{E}_n^c] \text{ (Inequality (\ref{eqn:set_inequality}))}  \nonumber \\
	&\leq \frac{1}{n}\sum_{i=1}^n \indicator\left[\left|\frac{\hat r_{i}}{\sqrt{0.9\E[\xi_{i}^2]}}\right|\geq d\right]+\indicator[\mathcal{E}_n^c] \text{ (Definition of $\mathcal{E}_n$)}  \nonumber \\
	&\leq \frac{1}{n}\sum_{i=1}^n \indicator\left[\frac{|\hat r_{i}-\xi_{i}|+\left|\xi_{i}\right|}{\sqrt{0.9\E[\xi_{i}^2]}}\geq d\right]+\indicator[\mathcal{E}_n^c]  \text{ (Triangle inequality)}  \nonumber \\
	&\leq \frac{1}{n}\sum_{i=1}^n \indicator\left[\frac{C(Y_i+\bar\theta)\|\hat\beta_n-\beta\|+\left|\xi_{i}\right|}{\sqrt{0.9\E[\xi_{i}^2]}}\geq d\right]+\indicator[\mathcal{E}_n^c] \text{ (Lemma \ref{lem:xi_r_diff})}.
\end{align}
Next, define the sequence of events $\mathcal{D}_1, \mathcal{D}_2, \dots$ by $\mathcal{D}_n\equiv \{\|\hat\beta_n-\beta\| \leq 1/C\}$. We can further decompose the left-hand summand of (\ref{eq:truncation_upper_bound_1}) as
\begin{align}\label{eq:truncation_upper_bound_2}
	\frac{1}{n}\sum_{i=1}^n \indicator&\left[\frac{C(Y_i+\bar\theta)\|\hat\beta_n-\beta\|+\left|\xi_{i}\right|}{\sqrt{0.9\E[\xi_{i}^2]}}\geq d\right] \nonumber \\
	&\leq 
	\frac{1}{n}\sum_{i=1}^n \indicator\left[\frac{C(Y_i+\bar\theta)\|\hat\beta_n-\beta\|+\left|\xi_{i}\right|}{\sqrt{0.9\E[\xi_{i}^2]}}\geq d,\mathcal{D}_n\right]+\indicator(\mathcal{D}_n^c) \text{ (Inequality (\ref{eqn:set_inequality}))} \nonumber \\
	&\leq \frac{1}{n}\sum_{i=1}^n \indicator\left[\frac{Y_i+\bar\theta+\left|\xi_{i}\right|}{\sqrt{0.9\E[\xi_{i}^2]}}\geq d\right]+\indicator(\mathcal{D}_n^c) \text{ (Definition of $\mathcal{D}_n$)}.
\end{align}
Taking the limit of (\ref{eq:truncation_upper_bound_2}) and applying SLLN and Markov's inequality, we obtain
\begin{align}\label{eqn:truncation_upper_bound_3}
	\limsup_{n\rightarrow\infty}\frac{1}{n}\sum_{i=1}^n \indicator&\left[\frac{Y_i+\bar\theta+\left|\xi_{i}\right|}{\sqrt{0.9\E[\xi_{i}^2]}}\geq d\right] \nonumber \\
	&=\P\left[\frac{Y_i+\bar\theta+\left|\xi_{i}\right|}{\sqrt{0.9\E[\xi_{i}^2]}}\geq d\right] \textrm{ (SLLN)} \nonumber \\ &=\P\left[\frac{(Y_i+\bar\theta+|\xi_{i}| )^2}{0.9\E[\xi_{i}^2]}\geq d^2\right] \nonumber \\
	&\leq\frac{\E\left[\left(Y_i+\bar\theta+\left|\xi_{i}\right|\right)^2\right]}{d^2\cdot 0.9\E[\xi_{i}^2]} \text{ (Markov's)} \nonumber \\
	&\leq\frac{4\E[Y_i^2]+4\bar\theta^2+2\E[\xi_{i}^2]}{d^2\cdot 0.9\E[\xi_{i}^2]},
\end{align}
almost surely. Additionally, by Lemma \ref{lem:second_moment_r_lower_bound} and Assumption \ref{assu:consistency_beta},
\begin{equation}\label{eqn:truncation_upper_bound_4}
\limsup_{n\to\infty} \indicator(\mathcal{E}^c_n) = \limsup_{n\to\infty} \indicator(\mathcal{D}_n^c) = 0,
\end{equation}
 almost surely. Combining \eqref{eq:truncation_upper_bound_1}, \eqref{eq:truncation_upper_bound_2}, \eqref{eqn:truncation_upper_bound_3}, and \eqref{eqn:truncation_upper_bound_4} we conclude that
\begin{align*}
	&
	\limsup_{n\rightarrow\infty}\frac{1}{n}\sum_{i=1}^n \indicator\left[\left|\frac{\hat r_{i}}{\sqrt{\sum_{i=1}^n \hat r_{i}^2/n}}\right|\geq d\right]\\
	&
	\leq \frac{4\E[Y_i^2]+4\bar\theta^2+2\E[\xi_{i}^2]}{d^2\cdot 0.9\E[\xi_{i}^2]}+\limsup_{n\rightarrow\infty}\indicator(\mathcal{E}_n^c)+\limsup_{n\rightarrow\infty}\indicator(\mathcal{D}_n^c)\\
	&
	=\frac{4\E[Y_i^2]+4\bar\theta^2+2\E[\xi_{i}^2]}{d^2\cdot 0.9\E[\xi_{i}^2]},
\end{align*}
almost surely. Finally, by Lemma \ref{lem:NB_model_moments_bound}, $\E[Y_i^2]$ and $\E[\xi_{i}^2]$ are finite; thus, sending $d$ to $\infty$, we obtain the result.
$\square$

\textbf{Proof of Lemma \ref{lem:consistency_W_i}}. By Lemma \ref{lem:xi_r_diff}, for $m \in \{1, 2\}$, we have
$$ \Bigg| \frac{1}{n} \sum_{i=1}^n \hat{W}_i^m - \frac{1}{n} \sum_{i=1}^n W^m_i  \Bigg| \leq \frac{1}{n} \sum_{i=1}^n | \hat{W}^m_i - W^m_i| \leq  m C \bar{\theta}^m || \hat{\beta}_n - \beta|| \xrightarrow{p} 0.$$ Additionally, by LLN (Lemma \ref{lem:SLLN}) and the boundedness of $\E[W_i^m ]$ (Lemma \ref{lem:NB_model_moments_bound}), $$(1/n) \sum_{i=1}^n W^m_i \xrightarrow{p} \E[W_i^m ].$$ Finally, by the triangle inequality, $(1/n) \sum_{i=1}^n \hat{W}_i^m  \xrightarrow{p} \E[W_i^m ]$.
The corollary follows by the boundedness of $Z_{ij}$:
$$\frac{1}{n} \sum_{i=1}^n Z_{ij}^m \hat{W}_i^m \leq C^m \frac{1}{n} \sum_{i=1}^n \hat{W}^m_i = O_p(1).$$
$\square$

\section{Proof of Proposition \ref{thm:asy_distribution_misspecification}}\label{sec:prop_asy_proof}
We decompose
\begin{align*}
	T_n(X,Y,Z)
	&
	=\frac{\frac{1}{\sqrt{n}}\sum_{i=1}^n X_i\hat r_i}{\sqrt{\frac{1}{n}X^\top \hat W X-\left(\frac{1}{n}X^\top \hat W Z\right)\left(\frac{1}{n}Z^\top \hat W Z\right)^{-1}\left(\frac{1}{n}Z^\top \hat W X\right)}}\\
	&
	=\frac{\frac{1}{\sqrt{n}}\sum_{i=1}^n \frac{X_i (Y_i-\mu_i)}{1+\mu_i/\bar\theta}+\frac{1}{\sqrt{n}}\sum_{i=1}^n \frac{X_i (\mu_i-\hat \mu_i)}{1+\mu_i/\bar\theta}}{\sqrt{\frac{1}{n}X^\top \hat W  X-\left(\frac{1}{n}X^\top \hat W  Z\right)\left(\frac{1}{n}Z^\top \hat W Z\right)^{-1}\left(\frac{1}{n}Z^\top \hat W  X\right)}}.
\end{align*}
The proof consists of the following three steps:
\begin{enumerate}
	\item Show that
	\begin{align}
		&\nonumber
		\frac{1}{\sqrt{n}}\sum_{i=1}^n \frac{X_i (\mu_i-\hat \mu_i)}{1+\mu_i/\bar\theta}\\
		&\label{eq:sampling_convergence_term1}
		=-\E[Z_i^\top X_i W_i]\left(\E[W_iZ_i Z_i^\top]\right)^{-1}\frac{1}{\sqrt{n}}\sum_{i=1}^n Z_i\xi_i+o_p(1) \equiv \star.
	\end{align}
	\item Show that
	\begin{align}\label{eq:sampling_convergence_term2}
		\frac{1}{\sqrt{n}}\sum_{i=1}^n X_i \xi_i + \star \convd N\left(0, \E\left[ R_i^2 S_i/W_i \right] \right).
	\end{align}
	\item Show that 
	\begin{align}\label{eq:sampling_convergence_term3}
		\frac{1}{n}X \hat W X-\left(\frac{1}{n}Z^\top  X \hat W \right)\left(\frac{1}{n}Z^\top \hat W Z\right)^{-1}\left(\frac{1}{n} \hat W X Z \right) \convp \E\left[R_i^2\right].
	\end{align}
	
	We obtain the result by combining (\ref{eq:sampling_convergence_term1}), (\ref{eq:sampling_convergence_term2}), and (\ref{eq:sampling_convergence_term3}) via Slutsky's theorem.
\end{enumerate}

\underline{Proof of (\ref{eq:sampling_convergence_term1})}. 
By a second-order Taylor expansion,
$$\mu_i - \hat{\mu}_ i = \exp(\beta^\top  Z_i) - \exp(\hat{\beta}_n^\top  Z_i) = \exp(\beta^\top  Z_i) (\beta - \hat{\beta}_n)^\top  Z_i + \frac{\exp(c_i)}{2}((\beta - \hat{\beta}_n)^\top  Z_i)^2,$$ where $c_i$ is between $\hat{\beta}^\top Z_i$ and $\beta^\top Z_i$. The $c_i$s satisfy
\begin{equation}\label{eqn:max_ci}
	\max_{i \in [n]} c_i = O_p(1).
\end{equation}
Indeed,
$\max_{i \in [n]} |\beta^\top  Z_i| \leq C \left\| \beta \right\| = O_p(1),$ and
\begin{align*}
	\max_{i \in [n]} | \hat{\beta}^\top  Z_i | \leq \max_{i \in [n]} |(\hat{\beta}_n - \beta)^\top  Z_i | &+ \max_{i \in [n]}  | \beta^\top  Z_i | \\ &\leq C ||\hat{\beta}_n - \beta || + C ||\beta|| = o_p(1) + O_p(1) = O_p(1).
\end{align*}
Claim (\ref{eqn:max_ci}) follows from the fact that $c_i$ is between $\beta^\top Z_i$ and $\hat{\beta}_n^\top  Z_i$. Using the Taylor expansion, we can write
\begin{align*}
	\frac{1}{\sqrt{n}}\sum_{i=1}^n \frac{X_i (\mu_i-\hat \mu_i)}{1+\mu_i/\bar\theta} = &\frac{1}{\sqrt{n}} \sum_{i=1}^n X_i W_i (\beta - \hat{\beta}_n)^\top  Z_i \\ &+ O_p(1) \frac{1}{\sqrt{n}} \sum_{i=1}^n \frac{X_i ((\hat{\beta}_n - \beta)^\top  Z_i)^2}{1 + \mu_i/\bar \theta} \equiv I + II.
\end{align*}
Focusing first on term I, we see
\begin{align*}
	\frac{1}{\sqrt{n}} &\sum_{i=1}^n X_i W_iZ_i^\top (\beta - \hat{\beta}_n) \\ 
	&= \frac{1}{n} \sum_{i=1}^n Z_i^\top X_i W_i  ( \sqrt{n}(\beta - \hat{\beta}_n ))  \\
	&= - \frac{1}{n} \sum_{i=1}^n Z_i^\top X_i W_i \left[\left(\E[W_i Z_i Z_i^\top ]\right)^{-1} \frac{1}{\sqrt{n}} \sum_{j=1}^n Z_j \xi_j + o_p(1) \right] \text{ (Lemma \ref{lem:consistency_beta})} \\
	&= -\left(\E \left[Z_i^\top X_i W_i\right] + o_p(1) \right) \left(\E\left[W_i Z_i Z_i^\top  \right]\right)^{-1} \frac{1}{\sqrt{n}} \sum_{i=1}^n Z_i \xi_i + o_p(1) \text{ (LLN)}  \\
	&= -\E \left[Z_i^\top X_i W_i\right] \left(\E\left[W_i Z_i Z_i^\top  \right]\right)^{-1} \frac{1}{\sqrt{n}} \sum_{j=1}^n Z_j \xi_j + o_p(1). \text{ (CLT)}
\end{align*}
Next, turning our attention to term II, we observe
\begin{align*}
	\left| O_p(1) \frac{1}{\sqrt{n}} \sum_{i=1}^n \frac{X_i ((\hat{\beta}_n - \beta)^\top  Z_i)^2}{1 + \mu_i/\bar \theta} \right| &\leq O_p(1) \frac{1}{\sqrt{n}} \sum_{i=1}^n \frac{|X_i| \cdot ||\hat{\beta}_n - \beta||^2 \cdot ||Z_i||^2}{1 + \mu_i/\bar \theta} \\ &\leq O_p(1) ||\sqrt{n}(\hat{\beta}_n - \beta)||\cdot ||\hat{\beta}_n - \beta|| \cdot \frac{1}{n} \sum_{i=1}^n \frac{|X_i| C^2}{1 + \mu_i/\bar \theta} \\ &= O_p(1) O_p(1) o_p(1) O_p(1) = o_p(1),
\end{align*}
where we have applied Lemma \ref{lem:consistency_beta}, CLT, and the boundedness of $\E[X_i^4]$ (Assumption \ref{assu:moment_treatment}). Combining terms I and II yields claim (\ref{eq:sampling_convergence_term1}).

\underline{Proof of (\ref{eq:sampling_convergence_term2})}.
We can express (\ref{eq:sampling_convergence_term2}) as
\begin{align*}
	\frac{1}{\sqrt{n}}&\sum_{i=1}^n X_i \xi_i- \E[Z_i^\top X_i W_i]\left(\E[W_iZ_i Z_i^\top]\right)^{-1}\frac{1}{\sqrt{n}}\sum_{i=1}^n Z_i \xi_i + o_p(1) \\
	&=\frac{1}{\sqrt{n}}\sum_{i=1}^n\left(X_i-\E[Z_i^\top X_i W_i]\left(\E[W_iZ_i Z_i^\top]\right)^{-1}Z_i\right) \xi_i + o_p(1) \\  &\equiv \frac{1}{\sqrt{n}} \sum_{i=1}^n a_i \xi_i + o_p(1).
\end{align*}
We observe that
$$\E\left[ a_i \xi_i \right] = \E\left[ a_i  \E\left[\xi_i | (X_i , Z_i) \right] \right] = \E\left[a_i \frac{\mu_i - \mu_i}{1 + \mu_i/\bar \theta} \right] = 0.$$  Furthermore, as we showed in claim (\ref{eqn:exp_xi_i_squared}),
$$
\textrm{Var}\left[a_i \xi_i \right] = \E\left[ a_i^2  \E\left[\frac{(Y_i - \mu_i)^2}{(1 + \mu_i/\bar \theta)^2} \Big| (X_i, Z_i) \right] \right] = \E\left[a_i^2 S_i \right].
$$
Thus, by CLT,
$$ \frac{1}{\sqrt{n}} \sum_{i=1}^n a_i \xi_i + o_p(1) \convd N \left(0, \E[a_i^2 S_i] \right).$$
Finally, we can express $\E\left[ a_i^2 S_i \right]$ as
\begin{align*}
	\E&[a_i^2S_i] \\ 
	&= \E\left[\left(\sqrt{W}_i X_i  - \E\left[Z_i^\top X_i W_i \right] \left(\E\left[ W_i Z_i Z_i^\top  \right] \right)^{-1} \sqrt{W}_i Z_i \right)^2 S_i/W_i \right] \\ 
	&= \E\left[R_i^2 S_i/W_i \right].
\end{align*}

\underline{Proof of (\ref{eq:sampling_convergence_term3})}. By Lemmas \ref{lem:inverse_matrix_convergence},  \ref{lem:unconditional_convergence_term2_1}, and  \ref{lem:unconditional_convergence_term2_2}, and the continuous mapping theorem,
\begin{multline}\label{eqn:sampling_convergence_term3_limit}
	\frac{1}{n}X \hat W X-\left(\frac{1}{n} Z^\top  \hat W X \right)\left(\frac{1}{n}Z^\top \hat WZ\right)^{-1}\left(\frac{1}{n} X \hat W Z \right) \\ \convp \E[X_i^2W_i] - \E\left[ Z_i^\top  W_i X_i \right] \left(\E\left[ W_iZ_iZ_i^\top \right] \right)^{-1} \E\left[ X_i W_i Z_i \right].
\end{multline}
Let
$$a = \sqrt{W_i} X_i; \quad B = \E\left[Z_i^\top X_i W_i \right]; \quad C = \left( \E\left[W_i Z_iZ_i^\top  \right] \right)^{-1}; \quad D = \sqrt{W_i} Z_i.$$ We can express (\ref{eqn:sampling_convergence_term3_limit}) as
\begin{multline*}
	\E[X_i^2W_i] - \E\left[Z_i^\top X_i W_i \right] \left(\E\left[ W_iZ_iZ_i^\top \right] \right)^{-1} \E\left[W_i X_i Z_i \right] \\ = \E\left[ a^2\right] - 2 \E \left[aBCD\right] + \E\left[(BCD)^2\right]  = \E\left[\left( a - BCD \right)^2  \right]  = \E\left[R_i^2  \right].
\end{multline*}
Indeed, $\E[(\sqrt{W_i} X_i)^2] = \E[W_iX_i^2] = \E[a_i^2].$ Additionally, because $a$ is a scalar,
\begin{multline*}
	-2 \E\left[Z_i^\top X_i W_i \right]  \left(\E\left[W_i Z_iZ_i^\top \right] \right)^{-1} \E\left[W_i X_i Z_i \right] \\ = -2BC \E\left[aD\right] = -2\E\left[BCaD \right] = -2\E\left[ aBCD \right].
\end{multline*}
Finally, because $BCD$ is a scalar,
\begin{multline*}
	\E\left[Z_i^\top X_i W_i \right] \left(\E\left[W_iZ_iZ_i^\top \right] \right)^{-1} \E\left[W_i X_i Z_i \right] \\
	= \E\left[ \E\left[Z_i^\top X_i W_i \right]  \left( \E\left[W_i Z_iZ_i^\top  \right] \right)^{-1} \left(W_i Z_i Z_i^\top  \right)  \left( \E\left[W_i Z_iZ_i^\top  \right] \right)^{-1} \E\left[W_i X_i Z_i \right]\right] \\ 
	=\E\left[BCDD^\top  C^\top  B^\top  \right] = \E\left[ BCD (BCD)^\top  \right] = \E\left[(BCD)^2\right].
\end{multline*}
$\square$

\section{Proof of Proposition \ref{thm:asy_permutation_distribution}}\label{sec:asy_permutation_distribution_proof}

We write the permuted test statistic $T_n(X_{\pi},Y,Z)$ as 
\begin{align*}
	\frac{\sum_{i=1}^n X_{\pi(i)}\hat r_{i}}{\sqrt{\sum_{i=1}^n \hat r_{i}^2}}\sqrt{\frac{\frac{1}{n}\sum_{i=1}^n \hat r_{i}^2}{\frac{1}{n}X_\pi^\top \hat W X_{\pi}-\left(\frac{1}{n}X_\pi^\top \hat W Z\right)\left(\frac{1}{n}Z^\top \hat W Z\right)^{-1}\left(\frac{1}{n}Z^\top \hat W X_{\pi}\right)}}.
\end{align*}
The proof consists of the following four steps:
\begin{enumerate}
	\item Show that 
	\begin{align}\label{eq:conditional_normality_NB}
		\frac{\sum_{i=1}^n X_{\pi(i)}\hat r_{i}}{\sqrt{\sum_{i=1}^n \hat r_{i}^2}}\mid\mathcal{F}_n\convdp N(0,1).
	\end{align}
	\item Show that 
	\begin{align}\label{eq:second_moment_r_empirical_convergence}
		\frac{1}{n}\sum_{i=1}^n \hat r_{i}^2 \convp \E\left[\frac{(Y_{i}-\mu_{i})^2}{(1+\mu_{i}/\bar\theta)^2}\right]; \quad \E\left[\frac{(Y_{i}-\mu_{i})^2}{(1+\mu_{i}/\bar\theta)^2}\right]=\E[S_{i}].
	\end{align}
	\item Show that
	\begin{align}\label{eq:NB_denominator_permuting_1}
			\frac{1}{n}X_{\pi}^\top \hat W X_{\pi }  \convp \E[W_{i}].
	\end{align}
	
	\item Show that 
	\begin{align}\label{eq:NB_denominator_permuting_2}
		\left(\frac{1}{n}X_\pi^\top \hat W Z\right)\left(\frac{1}{n}Z^\top \hat W Z\right)^{-1}\left(\frac{1}{n}Z^\top \hat W X_\pi\right)\convp0 .
	\end{align}
\end{enumerate}
Combining the convergence results (\ref{eq:conditional_normality_NB}), (\ref{eq:second_moment_r_empirical_convergence}), (\ref{eq:NB_denominator_permuting_1}), and (\ref{eq:NB_denominator_permuting_2}) via continuous mapping theorem and conditional Slutsky's theorem (Lemma \ref{lem:cond_slutsky}) proves the proposition.
\\ \\
\noindent
\textbf{Proof of the conditional convergence (\ref{eq:conditional_normality_NB}).} We apply the H\'ajek CLT (Lemma \ref{lem:hajek_clt}) to establish this result. Define
\begin{align*}
	c_{i}\equiv \frac{\hat r_{i}}{\sqrt{\sum_{i=1}^n \hat r_{i}^2}}.
\end{align*}
We verify that the $c_i$s satisfy the four conditions of the H\'ajek central limit theorem.

\underline{Sum zero}. The score vector $U_{\bar \theta}(\beta)$ is zero when evaluated at the MLE $\hat{\beta}_n$:
$$U_{\bar \theta}(\hat{\beta}_n) = 0 \iff Z^\top  \hat{r} = 0.$$ The covariate matrix $Z$ contains a column of ones (corresponding to the intercept term; Assumption \ref{assu:intercept}). Thus, the $\hat{r}_i$s sum to zero, implying that the $c_i$s sum to zero.

\underline{Sum one}. The squared $c_i$s sum to one:
$$ \sum_{j=1}^n c^2_j = \frac{1}{ \sum_{i=1}^n \hat{r}_i^2} \sum_{j=1}^n \hat{r}_j^2 = 1.$$

\underline{Truncation condition}. We seek to show that
\begin{multline*}
	\lim_{d\to\infty} \limsup_{n\to\infty} \sum_{i=1}^n c_i^2 \indicator(|\sqrt{n} c_i| \geq d|) =  \\
	\lim_{d\rightarrow\infty}\limsup_{n\rightarrow\infty}\frac{1}{(1/n) \sum_{i=1}^n \hat r_{i}^2}\frac{1}{n}\sum_{i=1}^n \hat r_{i}^2\indicator\left(\left|\frac{\hat r_{i}}{\sqrt{\sum_{i=1}^n \hat r_{i}^2/n}}\right|\geq d\right) = 0,\ \text{almost surely}.
\end{multline*}
The Cauchy-Schwarz inequality implies
\begin{multline}\label{eqn:truncation_decomp}
	\frac{1}{(1/n) \sum_{i=1}^n \hat{r}_i^2} \frac{1}{n}\sum_{i=1}^n \hat r_{i}^2\indicator\left(\left|\frac{\hat r_{i}}{\sqrt{\sum_{i=1}^n \hat r_{i}^2/n}}\right|\geq d\right) \\ \leq \frac{1}{(1/n) \sum_{i=1}^n \hat{r}_i^2} \sqrt{\frac{1}{n}\sum_{i=1}^n \hat r_{i}^4}\sqrt{\frac{1}{n}\sum_{i=1}^n \indicator\left(\left|\frac{\hat r_{i}}{\sqrt{\sum_{i=1}^n \hat r_{i}^2/n}}\right|\geq d\right)}.
\end{multline}
By Lemma \ref{lem:second_moment_r_lower_bound}, 
$$ \liminf_{n\rightarrow\infty}\frac{1}{n}\sum_{i=1}^n \hat r_i^2>0, \text{almost surely}; \quad  \limsup_{n\rightarrow\infty}\frac{1}{n} \sum_{i=1}^n \hat{r}_i^4<\infty, \text{almost surely}.$$
Additionally, by Lemma \ref{lem:truncation_condition_Hajek_CLT},
$$ \limsup_{d\rightarrow\infty}\limsup_{n\rightarrow\infty}\frac{1}{n}\sum_{i=1}^n \indicator\left(\left|\frac{\hat r_{i}}{\sqrt{\sum_{i=1}^n \hat r_{i}^2/n}}\right|\geq d\right)=0, \text{almost surely}.
$$
Hence, the result holds.

\underline{Tail assumption}. Let $\ep > 0$ be given. We seek to show that
$$\lim_{n\to\infty} \P\left(\max_{i \in [n]}  \frac{|X_i|}{\sqrt{n}} > \ep \right) = 0.$$ Applying the union bound, Markov's inequality, and Assumption \ref{assu:moment_treatment} (i.e., the boundedness of $\E[X_i^4]$), we see that
\begin{equation}\label{eqn:tail_assumption}
\P\left( \max_{i \in [n]}  \frac{|X_i|}{\sqrt{n}} > \ep \right) \leq \sum_{i=1}^n \P\left( \frac{|X_i|}{\sqrt{n}} > \ep \right) = \sum_{i=1}^n \P\left( \frac{X_i^4}{n^2} > \ep^4 \right)  \leq \frac{\E[X_i^4] n}{ \ep^4n^2} = \frac{\E[X_i^4]}{\ep^4n}.
\end{equation}
Sending $n \to \infty$, we obtain the result.
\\ \\ \noindent
\textbf{Proof of the convergence (\ref{eq:second_moment_r_empirical_convergence}).} Lemma \ref{lem:second_moment_r_lower_bound} implies that
\begin{align*}
	\frac{1}{n}\sum_{i=1}^n \hat r_{i}^2 \convp \E[\xi_{i}^2].
\end{align*}
Also, observe that
\begin{multline}\label{eqn:exp_xi_i_squared}
	\E[ \xi_i^2] =	\E\left[\frac{(Y_{i}-\mu_{i})^2}{(1+\mu_{i}/\bar\theta)^2}\right] =\E\left[\frac{1}{(1+\mu_{i}/\bar\theta)^2}\E[(Y_{i}-\mu_{i})^2|Z_i]\right] \\ =
	\E\left[ \frac{1}{(1 + \mu_i/\bar \theta)^2} \textrm{Var}(Y_i | Z_i) \right] = \E\left[\frac{\mu_i(1+\mu_i/\theta)}{(1+\mu_{i}/\bar\theta)^2}\right] = \E\left[ S_i \right],
\end{multline}
completing the proof.
\\ \\
\textbf{Proof of the convergence (\ref{eq:NB_denominator_permuting_1}).} 
We seek to show
\begin{align}
	&\label{eq:quadratic_mean_convergence}
	\E\left[\frac{1}{n}\sum_{i=1}^n X_{\pi(i)}^2\hat W_{i}|\mathcal{F}_n\right]\convp \E[W_i],\\
	&\label{eq:quadratic_variance_convergence}
	\textrm{Var} \left[\frac{1}{n}\sum_{i=1}^n X_{\pi(i)}^2\hat W_{i}|\mathcal{F}_n\right]\convp 0.
\end{align}
First, we establish (\ref{eq:quadratic_mean_convergence}). By Lemma \ref{lem:hoeffding_identity} (i.e., Hoeffding's identity for randomly permuted inner products), we have
\begin{align*}
	\E\left[\frac{1}{n}\sum_{i=1}^n X_{\pi(i)}^2\hat W_{i}|\mathcal{F}_n\right]=\frac{1}{n}\sum_{i=1}^n \hat W_{i}\frac{1}{n}\sum_{i=1}^n X_{i}^2.
\end{align*}
Next, by Lemma \ref{lem:consistency_W_i}, $(1/n) \sum_{i=1}^n \hat{W}_i  \convp \E[W_i].$ Additionally, by Assumption \ref{assu:moment_treatment} (i.e., $\E[X_i^2] = 1$) and Lemma \ref{lem:SLLN} (i.e., LLN), $(1/n)\sum_{i=1}^n X_i^2 \convp 1$. Thus,
$$ \frac{1}{n}\sum_{i=1}^n \hat W_{i}\frac{1}{n}\sum_{i=1}^n X_{i}^2 \convp \E[W_i],$$ proving claim (\ref{eq:quadratic_mean_convergence}). Next, we establish (\ref{eq:quadratic_variance_convergence}). By Lemma \ref{lem:hoeffding_identity} ,
\begin{align*}
	\textrm{Var}\left[\frac{1}{n}\sum_{i=1}^n X_{\pi(i)}^2\hat W_{i}|\mathcal{F}_n\right] 
	\leq \frac{1}{n-1}\left(\frac{1}{n}\sum_{i=1}^n \hat W_{i}^2\right)\left(\frac{1}{n}\sum_{i=1}^n X_{i}^4\right).
\end{align*}
By Lemma \ref{lem:consistency_W_i}, $(1/n) \sum_{i=1}^n \hat{W}_i^2 \convp \E[W_i^2].$ Additionally, by Assumption 1 (i.e., $E[X_i^4] < \infty$) and Lemma \ref{lem:SLLN} (i.e., LLN), $(1/n) \sum_{i=1}^n X_i^4 < \infty$. Hence,
$$ \frac{1}{n-1}\left(\frac{1}{n}\sum_{i=1}^n \hat W_{i}^2\right)\left(\frac{1}{n}\sum_{i=1}^n X_{i}^4\right) = \frac{1}{n-1}O_p(1) O_p(1) \convp 0,$$ proving claim (\ref{eq:quadratic_variance_convergence}). Finally, combining (\ref{eq:quadratic_mean_convergence}) and (\ref{eq:quadratic_variance_convergence}) via Lemma \ref{lem:conditional_convergence} yields the convergence (\ref{eq:NB_denominator_permuting_1}).
\\ \\ \noindent
\textbf{Proof of the convergence (\ref{eq:NB_denominator_permuting_2}).} To prove the convergence (\ref{eq:NB_denominator_permuting_2}), it is sufficient to prove 

$$\left\|\frac{1}{n} X_\pi^\top  \hat{W} Z\right\| \cdot \left\| \left(\frac{1}{n} Z^\top  \hat{W}Z \right)^{-1} \right\| \cdot \left\|\frac{1}{n} Z^\top  \hat{W} X_\pi \right\| \convp 0.$$ We will demonstrate
\begin{equation}\label{eq:sandwitch_permute_term1}
\left\| \frac{1}{n} X^\top _\pi \hat{W} Z \right\|  = o_p(1)
\end{equation}
and
\begin{equation}\label{eq:sandwitch_permute_term2}
\left\| \left(\frac{1}{n} Z^\top  \hat{W} Z\right)^{-1} \right\| = O_p(1),
\end{equation}
from which the result will follow.

\underline{Proving claim (\ref{eq:sandwitch_permute_term1})}. Let $j \in [p]$ be given. We seek to show 
\begin{equation}\label{eq:sandwitch_permute_term1_j}
\frac{1}{n} \sum_{i=1}^n X_{\pi(i)}\hat{W}_iZ_{ij} \convp 0. 
\end{equation}
To establish this result, we will demonstrate
\begin{align}
	\label{eq:sandwitch_mean_convergence}
	\E\left[\frac{1}{n}\sum_{i=1}^n X_{\pi(i)}\hat W_iZ_{ij}|\mathcal{F}_n\right]&\convp 0,\\
	\label{eq:sandwitch_variance_convergence}
	\textrm{Var}\left[\frac{1}{n}\sum_{i=1}^n X_{\pi(i)}\hat W_iZ_{ij}|\mathcal{F}_n\right]&\convp 0.
\end{align}
First, by Lemma \ref{lem:hoeffding_identity} (i.e., Hoeffding's identity), Lemma \ref{lem:consistency_W_i} (i.e., boundedness in probability of $\hat{W}_i Z_{ij}$), Assumption \ref{assu:moment_treatment} (i.e., $\E[X_i] = 0$), and LLN (i.e., Lemma \ref{lem:SLLN}), 
$$
\E\left[\frac{1}{n}\sum_{i=1}^n X_{\pi(i)}\hat W_iZ_{ij}|\mathcal{F}_n\right]=\left( \frac{1}{n}\sum_{i=1}^n \hat W_iZ_{ij}\right) \left(\frac{1}{n}\sum_{i=1}^n X_{i} \right) = O_p(1) o_p(1) \convp 0,
$$
proving (\ref{eq:sandwitch_mean_convergence}). Next, by Lemma \ref{lem:hoeffding_identity} (i.e., Hoeffding's identity), Lemma \ref{lem:consistency_W_i} (i.e., boundedness in probability of $\hat{W}_i^2 Z_{ij}$), Assumption \ref{assu:moment_treatment} (i.e., $\E[X_i^2] = 1$), and LLN (Lemma \ref{lem:SLLN}),
\begin{align*}
	&
	\textrm{Var}\left[\frac{1}{n}\sum_{i=1}^n X_{\pi(i)}\hat W_iZ_{ij}|\mathcal{F}_n\right]\leq \frac{1}{n-1}\left(\frac{1}{n}\sum_{i=1}^n \hat W_i^2Z_{ij}^2\right)\left(\frac{1}{n}\sum_{i=1}^n X_{i}^2\right) \\ &= \frac{1}{n-1} O_p(1) O_p(1) \convp 0,
\end{align*}
proving (\ref{eq:sandwitch_variance_convergence}). Combining (\ref{eq:sandwitch_mean_convergence}) and (\ref{eq:sandwitch_variance_convergence}) via Lemma  \ref{lem:conditional_convergence} yields claim (\ref{eq:sandwitch_permute_term1_j}). Finally, because $j \in [p]$ was chosen arbitrarily, claim (\ref{eq:sandwitch_permute_term1}) is proven.

\underline{Proving claim (\ref{eq:sandwitch_permute_term2})}. By Lemma \ref{lem:inverse_matrix_convergence}, the continuity of the operator norm, and the continuous mapping theorem,
$$ \Bigg\| \left( \frac{1}{n} Z^\top  \hat{W} Z \right)^{-1} \Bigg\| \convp \left\| \left( \E[W_i Z_i Z_i^\top ] \right)^{-1} \right\|,$$ implying claim (\ref{eq:sandwitch_permute_term2}). Finally, combining claims (\ref{eq:sandwitch_permute_term1}) and (\ref{eq:sandwitch_permute_term2}) yields the convergence result (\ref{eq:NB_denominator_permuting_2}). $\square$

\section{Proof of Theorem \ref{thm:camp}}\label{sec:camp_proof}

Define the test functions $\phi_n^p$, $\phi_n^{p'}$, and $\phi_n^{s'}$ by
\begin{align*}
	\phi_n^p(X,Y,Z) &\equiv \indicator\left(T_n(X,Y,Z) \geq \mathbb{Q}_{1-\alpha} \left[ T_n(X_\pi, Y, Z) | \mathcal{F}_n \right] \right), \\
	\phi_n^{p'}(X,Y,Z) &\equiv \indicator\left(T_n(X,Y,Z)/\sigma_p \geq \mathbb{Q}_{1-\alpha} \left[ T_n(X_\pi, Y, Z)/\sigma_p | \mathcal{F}_n \right] \right). \\
	\phi^{s'}_n &\equiv \indicator\left( T_n(X,Y,Z) / \sigma_p \geq z_{1-\alpha} \right).
\end{align*}
Suppose that the data are generated by the NB GLM (\ref{eqn:nb_glm}) and that the distribution $\mathcal{L}$ lies within $\mathcal{R} \cap \mathcal{N}$. By Proposition \ref{thm:asy_permutation_distribution} and Lemma \ref{lem:conditional-convergence-to-quantile}, the $1-\alpha$ quantile of the scaled permutation distribution converges in probability to the $1-\alpha$ quantile of the standard Gaussian i.e., 
$$\mathbb{Q}_{1-\alpha}\left[ T_n(X_\pi, Y, Z) / \sigma_p | \mathcal{F}_n \right] \convp z_{1-\alpha}.$$ Additionally, by Proposition $\ref{thm:asy_distribution_misspecification}$ and the continuous mapping theorem, $T_n(X,Y,Z)/\sigma_p$ converges weakly to $N(0, \sigma_s^2/\sigma_p^2)$, which does not accumulate near the critical value $z_{1-\alpha}$. Thus, we can apply Lemma \ref{lem:equivalence-lemma}, setting $S_n(X,Y,Z)$ to $T_n(X,Y,Z)/\sigma_p$, $C_n(X,Y,Z)$ to $\mathbb{Q}_{1-\alpha}[T_n(X_\pi, Y, Z)/\sigma_p | \mathcal{F}_n]$, $\phi_n^1$ to $\phi_n^{p'}$, and $\phi_n^2$ to $\phi_n^{s'}$. We obtain
\begin{equation}\label{eqn:test_equiv_1}
\lim_{n\to\infty}\P \left[\phi_n^{p'}(X,Y,Z) = \phi_n^{s'}(X,Y,Z) \right] = 1.
\end{equation}
Because $T_n(X,Y,Z)/\sigma_s$ converges weakly to the standard Gaussian (Theorem \ref{thm:asy_distribution_misspecification}), we can express the asymptotic type-I error of $\phi^{s'}_n$ as
\begin{equation}\label{eqn:test_equiv_2}
\lim_{n\to\infty} \E\left[ \phi^{s'}_n(X,Y,Z) \right] = \lim_{n\to\infty}\P\left[ \frac{T_n(X,Y,Z)}{\sigma_p} \cdot \frac{\sigma_p}{\sigma_s} \geq \frac{z_{1-\alpha} \sigma_p}{\sigma_s} \right] = 1 - \Phi\left(\frac{z_{1-\alpha} \sigma_p}{\sigma_s}\right).
\end{equation}
On the other hand,
\begin{equation}\label{eqn:test_equiv_3}
	\phi_n^{p'}(X,Y,Z) = \phi_n^p(X,Y,Z), \text{ almost surely}.
\end{equation}
Combining (\ref{eqn:test_equiv_1}), (\ref{eqn:test_equiv_2}), and (\ref{eqn:test_equiv_3}) yields the result:
$$ \lim_{n\to\infty}\E\left[\phi_n^p(X,Y,Z)\right] =  1 - \Phi\left(\frac{z_{1-\alpha} \sigma_p}{\sigma_s}\right) \equiv \textnormal{Err}_\mathcal{R}.$$ 
We can bound the excess asymptotic type-I error  $\textnormal{Err}_\mathcal{R} - \alpha$ of $\phi_n^p$ via the mean value theorem as follows:
\begin{align*}
\textnormal{Err}_{\mathcal{R}} - \alpha &= 1 - \alpha - \Phi(z_{1-\alpha}\sigma_p/\sigma_s) \\ 
&= \Phi(\Phi^{-1}(1 - \alpha)) - \Phi(\Phi^{-1}(1-\alpha) \sigma_p/\sigma_s) \\
&= f(c)\left(\Phi^{-1}(1-\alpha) - \Phi^{-1}(1-\alpha)\sigma_p/\sigma_s \right) \\ 
&\leq \frac{1}{\sqrt{2\pi}}z_{1-\alpha}\left( 1 - \sigma_p/\sigma_s \right).  
\end{align*}
Above, $f$ is the density of the standard Gaussian (i.e., the derivative of the CDF $\Phi$); $c$ is a number between $\Phi^{-1}(1-\alpha)$ and $\Phi^{-1}(1 - \alpha)\sigma_p/\sigma_s$; and $\Phi^{-1}(1-\alpha) = z_{1-\alpha}$ is the $1-\alpha$ quantile of the standard Gaussian. We also have used the fact that $\sup_{x \in \R} f(x) \leq 1/\sqrt{2\pi}.$

Next, suppose that the distribution $\mathcal{L}$ lies within $\mathcal{K} \cap \mathcal{N}$. Then $$X_i \indep Z_i; \quad X_i \indep Y_i | Z_i.$$ By the contraction property of conditional independence, $X_i \indep (Y_i, Z_i)$. The permutation test $\phi_n^p(X,Y,Z)$ is a finite-sample valid test of independence between $X_i$ and $(Y_i, Z_i)$. Thus, 
$$ \E\left[ \phi^p_n(X,Y,Z) \right] \leq \alpha.$$ 

\section{Results for the linear model}\label{sec:lin_model}
Suppose we observe i.i.d.\ data $(X_1, Y_1, Z_1), \dots, (X_n, Y_n, Z_n)$  from the linear model
	\begin{equation}\label{eqn:random_design_linear_model}
		\begin{cases}
		(X_i, Z_i) \sim \mathcal{P} \\
		Y_i = \gamma X_i + \beta^T Z_i + \ep_i, \\
	\end{cases}
	\end{equation}
where $\ep_i$ is a mean-zero error term and $\gamma\in\R$ and $\beta \in\R^p$ are unknown regression coefficients. Let $\hat{\beta}_n = (Z^TZ)^{-1}Z^TY$ be the OLS estimator for $\beta$ under the null hypothesis of $\gamma = 0$, and let $\hat{r} = Y - Z\hat{\beta}_n$ denote the vector of residuals from this regression model. The score test statistic $S_n(X,Y,Z)$ for testing $\gamma = 0$ is given by 
\begin{equation}\label{eq_lm_score_test_statistic}
	S_n(X, Y, Z) = \frac{X^T \hat{r}}{\sqrt{X^TX - X^TZ(Z^TZ)^{-1}Z^T X}}.    
\end{equation}
We derive the asymptotic permutation distribution of the score test statistic in the linear model. We make the following assumption about the error term $\ep_i$.

\begin{assumption}[Bounded moments of error term]\label{assu:lm_error}
	We assume $\E[\ep_i] = 0$ and $\E[\ep_i^2] = 1$, without loss of generality. Additionally, we assume $\E[\ep_i^4] < \infty$.
\end{assumption}

Note that we do not necessarily assume that $\ep_i$ is Gaussian. Our main theorem for the linear model is as follows.

\begin{proposition}[Asymptotic permutation distribution of score test statistic in linear model]\label{thm:lm_asy_perm_distribution}
Suppose that Assumptions \ref{assu:moment_treatment} -- \ref{assu:lm_error} hold. Let $\mathcal{F}_n$ be the sequence of $\sigma$-algebras generated by the data. Then
$$S_n(X, Y, Z) | \mathcal{F}_n \convdp N(0,1).$$
\end{proposition}

\subsection{Lemmas for linear model}

Both Lemmas assume that Assumptions \ref{assu:moment_treatment} -- \ref{assu:lm_error} hold. We also leverage the fact that $\hat{\beta}_n$ is a strongly consistent estimator for $\beta$ in the random-design linear model (\ref{eqn:random_design_linear_model}).

\begin{lemma}[Consistency of residuals in the linear model]\label{lem:lm_resid_consistency}
For $m \in \{1, 2, 4\}$,
$$\lim_{n\to\infty} \frac{1}{n} \sum_{i=1}^n \hat{r}_i^m = \E[\ep_i^m] , \text{ almost surely.}$$
\end{lemma}

\begin{lemma}[Truncation condition in the linear model]\label{lem:lm_truncation_condition}
We have that
$$\lim_{d \to \infty} \limsup_{n \to \infty} \frac{1}{n} \sum_{i=1}^n \indicator\left( \left| \frac{\hat{r}_i}{ \sqrt{\sum_{i=1}^n \hat{r}^2_i/n}} \right| \geq d \right) = 0, \text{ almost surely}.$$
\end{lemma}

\subsection{Proofs of lemmas for linear model}

\textbf{Proof of Lemma \ref{lem:lm_resid_consistency}}. We can express the $i$th residual $\hat{r}_i$ as $\hat{r}_i = \ep_i - Z_i^T(\hat{\beta}_n - \beta )$. Let $m \in \{1, 2, 4\}$ be given. By the binomial theorem,
$$ \hat{r}_i^m = \ep_i^m + \sum_{k=1}^m (-1)^k\binom{m}{k} \ep_i^{m-k}(Z_i^T(\hat{\beta}_n - \beta))^k.$$ Applying the triangle inequality,
$$
\left| \hat{r}_i^m - \ep_i^m \right|\leq \sum_{k=1}^m \binom{m}{k}|\ep_i|^{m-k} C^k ||\hat{\beta}_n - \beta||^k.
$$
By Assumption \ref{assu:lm_error}, $\E\left[\ep_i^4\right] < \infty$, implying that $\E[|\ep_i|^{m-k}] < \infty$ for all $m \in \{1, 2, 4\}$ and $k \in [m]$.  Additionally, $||\hat{\beta}_n - \beta||^k$ converges almost surely to zero. Thus, by SLLN (Lemma \ref{lem:SLLN}),
\begin{align*}
\lim_{n\to\infty} \left| \frac{1}{n} \sum_{i=1}^n \hat{r}_i^m - \frac{1}{n} \sum_{i=1}^n \ep_i^n \right| \\ &\leq \lim_{n\to\infty}  \frac{1}{n} \sum_{i=1}^n \left|\hat{r}_i^m - \ep_i^m  \right| \\ 
&\leq \lim_{n\to\infty}  \sum_{k=1}^m \binom{m}{k} C^k ||\hat{\beta}_n - \beta||^k  \frac{1}{n} \sum_{i=1}^n |\ep_i|^{m-k} = 0,
\end{align*}
almost surely. Moreover, $\lim_{n\to\infty} (1/n) \sum_{i=1}^n \ep_i^m = \E[\ep_i^m],$ almost surely. We conclude by the triangle inequality that
$\lim_{n\to\infty} (1/n) \sum_{i=1}^n \hat{r}_i^m = \E[\ep_i^m],$ almost surely.

\textbf{Proof of Lemma \ref{lem:lm_truncation_condition}}. We proceed in a manner similar to the proof of Lemma \ref{lem:truncation_condition_Hajek_CLT}. Let the sequence of events $\mathcal{E}_1, \mathcal{E}_2, \dots$ be defined by
$$\mathcal{E}_n \equiv \left\{ \frac{1}{n} \sum_{i=1}^n \hat{r}_i^2 \geq 0.9 \right\}.$$ We see that
\begin{align*}
\frac{1}{n} \sum_{i=1}^n \indicator\left( \left| \frac{\hat{r}_i}{ \sqrt{\sum_{i=1}^n \hat{r}^2_i/n}} \right| \geq d \right) \leq \frac{1}{n} \sum_{i=1}^n \indicator \left[ \frac{C||\hat{\beta}_n - \beta|| + |\ep_i|}{\sqrt{0.9}} \geq d \right] + \indicator[\mathcal{E}_n^c].
\end{align*}
Next, define the sequence of events $\mathcal{D}_1, \mathcal{D}_2, \dots$ by $\mathcal{D}_n \equiv \{||\hat{\beta}_n - \beta|| \leq 1/C \}$. We observe
$$\frac{1}{n} \sum_{i=1}^n \indicator \left[ \frac{C||\hat{\beta}_n - \beta|| + |\ep_i|}{\sqrt{0.9}} \geq d \right] \leq \frac{1}{n} \sum_{i=1}^n \indicator\left[ \frac{|\ep_i|}{\sqrt{0.9}} \geq d \right] + \indicator[\mathcal{D}^c_n].$$
By SLLN (Lemma \ref{lem:SLLN}) and Chebyshev's inequality,
$$\limsup_{n\to\infty} \frac{1}{n} \sum_{i=1}^n \indicator\left[\frac{|\ep_i|}{\sqrt{0.9}} \geq d \right] = \P\left( \frac{|\ep_i|}{\sqrt{0.9}} \geq d \right) \leq \frac{1}{0.9 \cdot d^2}, \textrm{ almost surely}.$$
Additionally, by Lemma \ref{lem:lm_resid_consistency} and the strong consistency of $\hat{\beta}_n$, $$\limsup_{n\to\infty} \indicator[\mathcal{E}_n^c] = \limsup_{n\to\infty} \indicator[\mathcal{D}^c_n] = 0, \text{ almost surely.}$$ Combining these results, we conclude that
$$\limsup_{n \to \infty} \frac{1}{n} \sum_{i=1}^n \indicator\left( \left| \frac{\hat{r}_i}{ \sqrt{\sum_{i=1}^n \hat{r}^2_i/n}} \right| \geq d \right) \leq \frac{1}{0.9 \cdot d^2}, \text { almost surely}.$$ Sending $d$ to $\infty$ completes the proof.

\subsection{Proof of Proposition \ref{thm:lm_asy_perm_distribution}}

We write the permuted test statistic $S_n(X_\pi, Y, Z)$ as
$$
S_n(X_\pi, Y, Z) = \frac{\sum_{i=1}^n X_{\pi(i)} \hat{r}_i}{ \sqrt{\sum_{i=1}^n \hat{r}_i^2 }} \sqrt{\frac{\frac{1}{n} \sum_{i=1}^n \hat{r}_i^2 }{\frac{1}{n} X_\pi^T X_\pi - \left(\frac{1}{n} X_\pi^TZ\right) \left(\frac{1}{n} Z^TZ\right)^{-1}\left(\frac{1}{n} Z^TX_\pi \right)}}. 
$$
The proof consists of the following three steps:
\begin{enumerate}
	\item Show that
	\begin{equation}\label{eqn:lm_num_conv}
 \frac{\sum_{i=1}^n X_{\pi(i)} \hat{r}_i}{ \sqrt{\sum_{i=1}^n \hat{r}_i^2 }} \Big| \mathcal{F}_n \convdp N(0,1).
	\end{equation}
	\item Show that
	\begin{equation}\label{lm:denom_1_conv}
		\frac{1}{n} X^\top_\pi X_\pi \convp 1.
	\end{equation}
	\item Show that
	\begin{equation}\label{lm:denom_2_conv}
		\left(\frac{1}{n} X_\pi^TZ\right) \left(\frac{1}{n} Z^TZ\right)^{-1}\left(\frac{1}{n} Z^TX_\pi \right) \convp 0.
	\end{equation}
\end{enumerate}
By Lemma \ref{lem:lm_resid_consistency},
\begin{equation}\label{lm:num_conv}
	\frac{1}{n} \sum_{i=1}^n \hat{r}_i^2 \convp 1.
\end{equation}
The result follows by combining (\ref{eqn:lm_num_conv}), (\ref{lm:denom_1_conv}), (\ref{lm:denom_2_conv}), and (\ref{lm:num_conv}) via the continuous mapping theorem and conditional Slutsky's theorem (Lemma \ref{lem:cond_slutsky}).
\\ \\ \noindent
\textbf{Proof of the conditional convergence (\ref{eqn:lm_num_conv}).} We apply the H\'ajek CLT (Lemma \ref{lem:hajek_clt}) to establish this result. Define $c_i \equiv \hat{r}_i/(\sum_{i=1}^n \hat{r}_i^2).$
We show that the $c_i$s satisfy the four conditions of the H\'ajek CLT.

\underline{Sum zero}. The linear model contains an intercept term (Assumption \ref{assu:intercept}), implying that the residuals and thus the $c_i$s sum to zero.

\underline{Sum one}. The squared $c_i$s sum to one by definition.

\underline{Truncation condition}. The Cauchy-Schwarz inequality implies the decomposition (\ref{eqn:truncation_decomp}). By Lemma \ref{lem:lm_resid_consistency} and Assumption \ref{assu:lm_error},
$$\liminf_{n\to\infty} \frac{1}{n} \sum_{i=1}^n \hat{r}_i^2 = 1 > 0 \text{ almost surely}; \quad \limsup_{n\to\infty} \frac{1}{n} \sum_{i=1}^n \hat{r}_i^4 < \infty, \text{ almost surely}.$$ Combining these results with Lemma \ref{lem:lm_truncation_condition} yields the truncation condition.

\underline{Tail assumption}. The proof is the same as in the NB regression case; see (\ref{eqn:tail_assumption}).
\\ \\ \noindent
\textbf{Proof of (\ref{lm:denom_1_conv}).} By Lemma \ref{lem:hoeffding_identity} (i.e., Hoeffding's identity),
$$\E \left[\frac{1}{n} \sum_{i=1}^n X_{\pi(i)}^2 \Big| \mathcal{F}_n \right] = \frac{1}{n}\sum_{i=1}^n X_i^2 \convp 1.$$
$$\textrm{Var}\left[\frac{1}{n} \sum_{i=1}^n X_{\pi(i)}^2 \Big| \mathcal{F}_n \right] \leq \frac{1}{n-1}\left(\frac{1}{n} \sum_{i=1}^n X_i^4 \right) = \frac{1}{n-1} O_p(1) \convp 0.$$
Applying Lemma \ref{lem:conditional_convergence} to the above two statements yields claim (\ref{lm:denom_1_conv}).
\\ \\ \noindent
\textbf{Proof of (\ref{lm:denom_2_conv}).} Let $j \in [p]$ be given. Again by Lemma \ref{lem:hoeffding_identity},
$$\left|\E\left[\frac{1}{n} \sum_{i=1}^n X_{\pi(i)} Z_{ij} \Big| \mathcal{F}_n \right] \right| = \left| \frac{1}{n} \sum_{i=1}^n Z_{ij} \right| \cdot \left|\frac{1}{n} \sum_{i=1}^n X_i \right| \leq Co_p(1) \convp 0.$$
$$\textrm{Var}\left[ \frac{1}{n} \sum_{i=1}^n X_{\pi(i)} Z_{ij} \Big| \mathcal{F}_n \right] \leq \frac{1}{n-1} \left(\frac{1}{n} \sum_{i=1}^n X^2_{i} \right) \left(\frac{1}{n} \sum_{i=1}^n Z^2_{ij} \right) \leq \frac{1}{n-1} C^2 O_p(1) \convp 0.$$
Combining the above two statements via Lemma \ref{lem:conditional_convergence} yields $(1/n) X_\pi^\top Z \convp 0.$  Additionally, by the invertibility of $\E[Z_iZ_i^T]$ (Assumption \ref{assu:lower_bound_eigenvalue}), the continuity of matrix inverse, and the continuous mapping theorem, $((1/n) Z^\top Z)^{-1} \convp (\E[Z_iZ_i^\top])^{-1}.$ Claim (\ref{lm:denom_2_conv}) follows by another application of the continuous mapping theorem.

\section{Additional simulation studies}\label{sec:additional_simulation_studies}

This section reports the results of supplementary simulation studies S1-S4.

\paragraph{Simulation S1 (Low signal-to-noise ratio; Figure \ref{fig:sim_studies_s1_s3}S1).} We sought to explore the impact of a low signal-to-noise ratio on method performance. We generated data from the NB GLM (\ref{eqn:nb_glm}) with an unconfounded treatment and a large sample size ($n = 1,000$). We varied the effect size of the nuisance covariate vector, ranging from a high SNR ($\left\|\beta\right\|_\infty = 0.1$) to a low SNR ($\left\|\beta\right\|_\infty = 1.2$). The methods performed similarly when the SNR was high; on the other hand, when the SNR was low, \texttt{MASS} and robust \texttt{MASS} exhibited substantially greater power than the MW test and the residual permutation test. We concluded on the basis of this simulation that there are at least two settings in which robust \texttt{MASS} is considerably more powerful than the residual permutation test: low SNR settings and confounded settings (Figure \ref{fig:sim_2}a-c).

\begin{figure}
	\centering
	\includegraphics[width=1.0\linewidth]{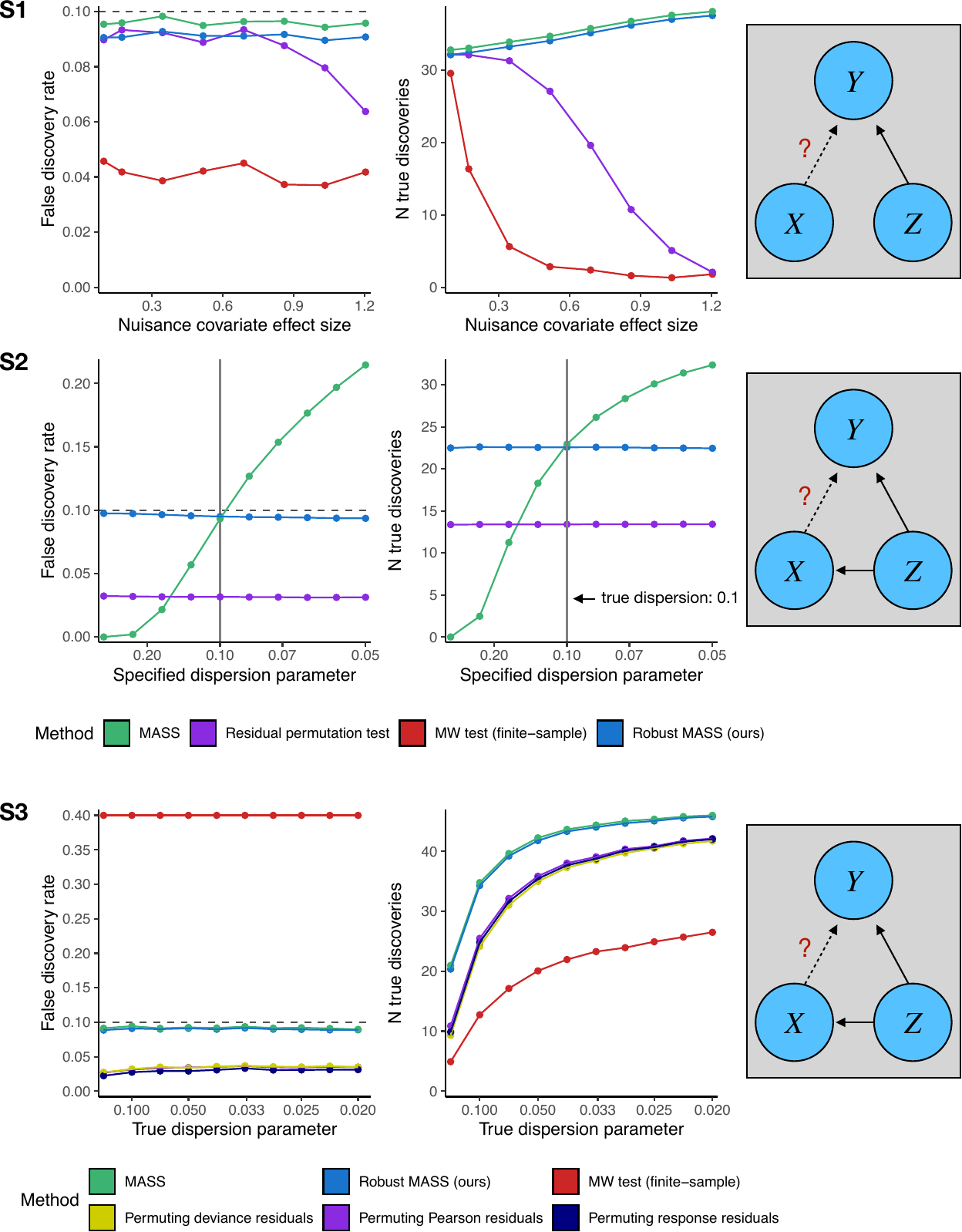}
	\caption{Additional simulation studies S1--S3.}\label{fig:sim_studies_s1_s3}
\end{figure}

\paragraph{Simulation S2 (Misspecified NB dispersion parameter; Figure \ref{fig:sim_studies_s1_s3}S2).} We sought to more thoroughly explore the impact of NB dispersion parameter misspecification on method performance. We generated data from the NB GLM (\ref{eqn:nb_glm}) with a confounded treatment and a large sample size ($n = 1,000$). We set the true dispersion $\phi$ to $0.1$. Instead of estimating $\phi$ (as we did in the main text simulation studies), we set $\phi$ to a specified value $\bar \phi$, which we varied over the grid $\{1/2, 1/2, \dots, 1/20\}.$ When the dispersion was correctly specified (i.e., $\bar \phi = \phi = 0.1$), \texttt{MASS} controlled type-I error and made a large number of true discoveries. However, when the specified dispersion was greater than (resp., less than) the true dispersion, \texttt{MASS} exhibited an inflation of false negative (resp., false positive) discoveries. Robust \texttt{MASS} and the residual permutation test controlled type-I error over all values of $\bar\phi$, and robust \texttt{MASS} achieved much higher power than the residual permutation test. These results highlight the utility of the permuted score test in protecting against dispersion parameter misspecification.

\paragraph{Simulation S3 (Pearson and deviance residual-based permutation test; Figure \ref{fig:sim_studies_s1_s3}S3).} The residual permutation test (Algorithm \ref{algo:permuting_residuals}) tests for association between the residuals $e$ from an NB regression of $Y$ onto $Z$ and the treatment vector $X$ via a permutation test. Our main text simulation studies set $e$ to the response residuals (i.e., $e_i = Y_i - \hat{\mu}_i$). However, we instead could set $e$ to the Pearson or deviance residuals. To explore whether this change might improve performance, we conducted a simulation study in which we evaluated six methods: (i) \texttt{MASS}, (ii) robust \texttt{MASS}, (iii) the MW test, and (iv)-(vi) the residual permutation test based on the response residuals, Pearson residuals, and deviance residuals. We generated the data from the NB model (\ref{eqn:nb_glm}) with a confounded treatment and a large sample size ($n = 1,000$). Reasoning that the dispersion parameter might impact relative method performance, we varied the true dispersion parameter $\phi$ over the grid $\{1/5, 1/10, \dots, 1/50\}.$ All methods controlled type-I error, save the MW test, which suffered confounding. \texttt{MASS} and robust \texttt{MASS} exhibited the highest power, the MW test exhibited the lowest power, and all three variants of the residual permutation test were intermediate. We concluded that swapping the response residuals for Pearson or deviance residuals did not improve performance of the residual permutation test.

\paragraph{Simulation S4 (DESeq2).} We evaluated \texttt{DESeq2}, robust \texttt{DESeq2}, \texttt{MASS}, and robust \texttt{MASS} on simulated data generated from the \texttt{DESeq2} model. We briefly describe the \texttt{DESeq2} model here. Let $Y \in \mathbb{N}^{n \times m}$ be the matrix of gene expression counts on $m$ genes across $n$ samples. Let $Z \in \mathbb{R}^{n \times p}$ be the matrix of nuisance covariates, and let $X \in \{0,1\}^n$ be the treatment vector. The expression level of gene $j \in \{1, \dots, p\}$ in sample $i \in \{1, \dots, n\}$ is modeled as
\begin{equation}\label{eqn_deseq_2_model}
Y_{ij} \sim \textrm{NB}_{\phi_j}(\mu_{ij}); \quad   \log(\mu_{ij}) = \beta^j_{0} + \gamma^j X_i + \beta^jZ_i + \log(s_i).
\end{equation}
Here, $\phi_j$ is the gene-specific NB size parameter; $\beta^j_0$, $\gamma^j$, and $\beta^j$ are gene-specific regression coefficients; and $s_i$ is the ``size factor'' of the $i$th sample. $s_i$ is an unobserved variable meant to capture the sequencing depth of sample $i$. \texttt{DESeq2} proceeds in four main steps. First, \texttt{DEseq2} estimates the sample-specific size factors via the ``median of ratios'' method. Let $$m_{j} = \left(\prod_{i=1}^n Y_{ij} \right)^{1/n}$$ be the geometric mean of the $j$th gene. The $i$th size factor $s_i$ is estimated as
$$\hat{s}_i = \textrm{median}_{j \in \{1 ,\dots, p\}} \left( \frac{Y_{ij}}{m_j} \right).$$
Second, \texttt{DESeq2} fits the NB GLM (\ref{eqn_deseq_2_model}) to each gene using an algorithm similar to that of \texttt{MASS}. Specifically, \texttt{DESeq2} plugs in the estimated size factors $\hat{s}_1, \dots, \hat{s}_n$ for the unobserved size factors $s_1, \dots, s_n$, treating these quantities as fixed and known offset terms. This step yields initial estimates $\hat{\phi}^\textrm{init}_1, \dots, \hat{\phi}^\textrm{init}_p$ of the gene-wise NB size parameters. Third, \texttt{DESeq2} regularizes the estimates $\hat{\phi}^\textrm{init}_1, \dots, \hat{\phi}^\textrm{init}_p$ via an empirical Bayes procedure that involves sharing information across genes, yielding improved estimates $\hat{\phi}^\textrm{final}_1, \dots, \hat{\phi}^\textrm{final}_p$. Finally, \texttt{DESeq2} refits the NB GLM (\ref{eqn_deseq_2_model}) for each gene, plugging in $\hat{\phi}_j^\textrm{final}$ for $\phi_j$ and $\hat{s}_1, \dots, \hat{s}_n$ for $s_1, \dots, s_n$, again treating these quantities as fixed and known. A $p$-value for the null hypothesis $\gamma^j = 0$ is computed via a Wald test.

\begin{figure}
	\centering
	\includegraphics[width=0.7\linewidth]{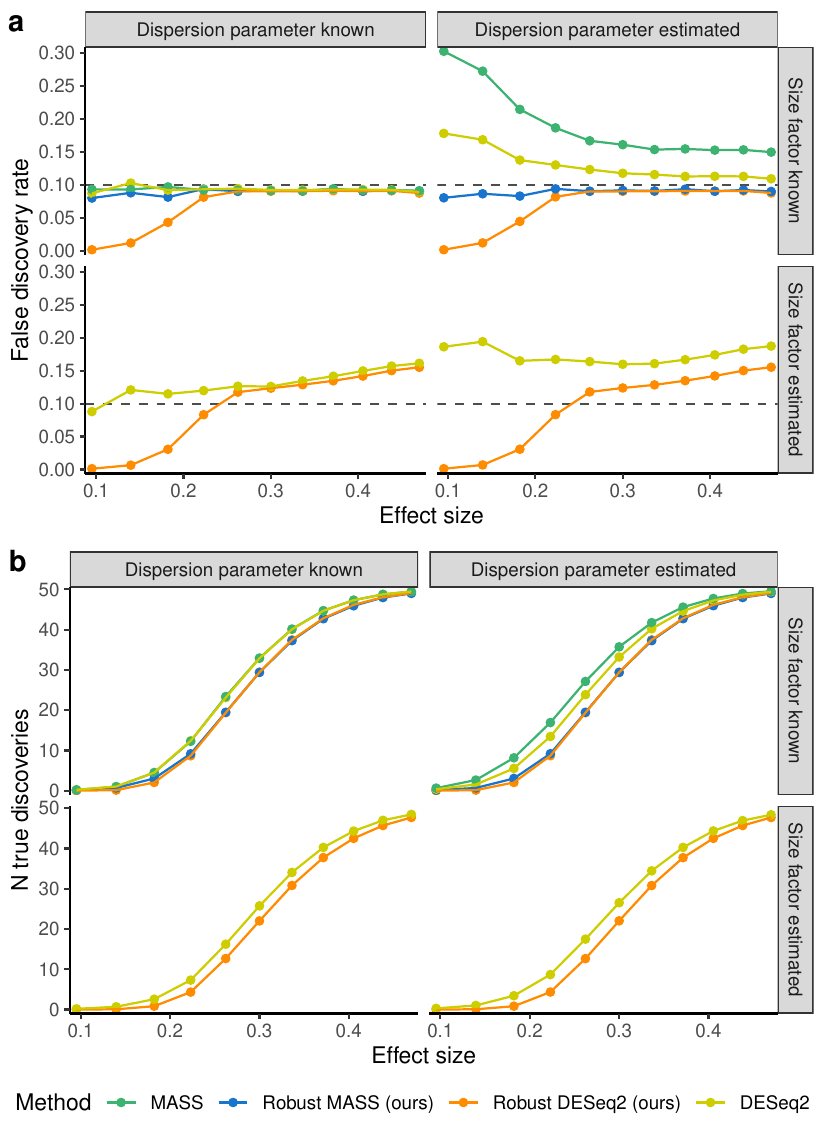}
	\caption{Additional simulation study S4.}\label{fig:sim_studies_s4}
\end{figure}

\begin{figure}
	\centering
	\includegraphics[width=0.2\linewidth]{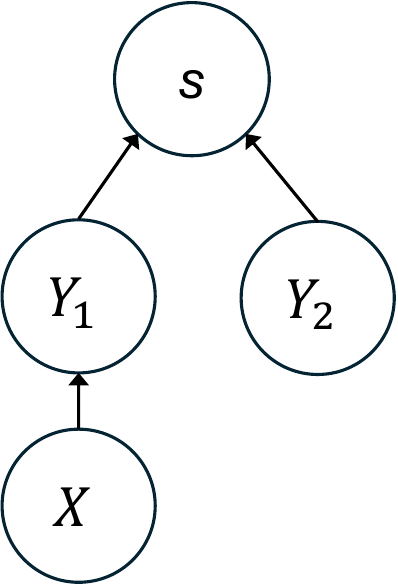}
	\caption{An example of collider bias in differential expression analysis. Gene $Y_1$ is marginally dependent with the treatment $X$, while gene $Y_2$ is marginally independent of $X$. Conditioning on the size factor $s$ --- a collider of $X$ and $Y_2$ --- induces dependence between $X$ and $Y_2$.}
	\label{fig:collider_bias}
\end{figure}

We simulated data from the \texttt{DESeq2} model. We set the sample size $n$ to a moderate value (i.e., $n = 50$) and the ambient gene expression level $\beta_0^j$ to a large number (i.e., $\beta_0^j = \log(80)$) to model a standard bulk RNA-seq experiment. We sampled the size factors $s_1 \dots, s_n$ from a $\textrm{Unif}(0.5, 1.5)$ distribution. We generated $X_i$ and $Z_i$ such that $X_i$ and $Z_i$ were independent. We varied the effect size under the alternative hypothesis $\gamma^j$ over the grid $\{0.10, 0.15, \dots, 0.55, 0.60\}.$ We applied \texttt{DESeq2}, robust \texttt{Deseq2}, \texttt{MASS}, and robust \texttt{MASS} to analyze the data. We either revealed the true size factors $s_1, \dots, s_n$ to the methods or required the methods to estimate these factors; similarly, we either revealed the true dispersion parameters $\phi_1, \dots, \phi_p$ the the methods or required the methods to estimate these parameters. Given that \texttt{MASS} does not contain a module for estimating size factors, we assessed \texttt{MASS} and robust \texttt{MASS} only in the settings where we supplied the size factors.

The results of this simulation study are displayed in Figure \ref{fig:sim_studies_s4}. Plot \textit{a} (resp., \textit{b}) shows the results for false discovery rate control (resp., power). Both plots contain four panels. The left (resp., right) column displays the results for known (resp., estimated) dispersion parameter, while the top (resp., bottom) row displays the results for known (resp., estimated) size factors. In the setting where the NB dispersion parameter and size factors were both known (i.e., top-left panel), all methods controlled the false discovery rate at the nominal level of 10\%. \texttt{MASS} and \texttt{DESeq2} exhibited slightly higher power than robust \texttt{MASS} and robust \texttt{DESeq2}. Next, in the setting where the size factors were known the NB size parameter was estimated (i.e., top-right panel), \texttt{MASS} and \texttt{DESeq2} failed to control type-I error, likely due to inaccurate estimation of the NB size parameter.

Next, in the setting where the dispersion was known and the size factors were estimated (lower-left panel), \texttt{DESeq2} and robust \texttt{DESeq2} exhibited mild type-I error inflation as the effect size under the alternative hypothesis increased. This result was surprising: we would expect robust \texttt{DESeq2} to control type-I error, as it is a permutation test, and the data were generated such that $X_i$ was unconfounded by $Z_i$. We reasoned that collider bias was responsible for the observed mild type-I error inflation. As a brief illustration, consider a toy RNA-seq dataset consisting of two genes, $Y_1$ and $Y_2$, and a treatment $X$ (Figure \ref{fig:collider_bias}). Suppose that $Y_1$ and $Y_2$ are independent. Additionally, suppose that $X$ impacts $Y_1$ but not $Y_2$. The size factor $s$ is a function of both $Y_1$ and $Y_2$. \texttt{DESeq2} and robust \texttt{DESeq2} condition on the size factor $s$ by including it as an offset in the model. This introduces a subtle problem: $s$ is a collider of $X$ and $Y_2$; thus, conditional on $s$, $X$ and $Y_2$ are dependent, resulting in a loss of type-I error control. To the best of our knowledge, collider bias is not widely recognized as a challenge in differential expression analysis. Moreover, collider bias is largely orthogonal to the issues we aim to tackle in this work (i.e., failure of asymptotic and parametric assumptions in the NB regression model). Thus, we delay a more careful study of collider bias to a followup work. Finally, in the setting where the dispersion parameter and size factors were both estimated (lower right panel), \texttt{DESeq2} exhibited more severe type-I error inflation than robust \texttt{DESeq2}, likely due to inaccurate dispersion parameter estimation.

\section{Efficiently computing GLM score tests}\label{sec:computing_glm_score_tests}

This section aims to compare Algorithm $R$ and Algorithm $Q$ for computing GLM score tests (Section \ref{sec:computational_accelerations}). We work within the context of an arbitrary (i.e., not necessarily NB) GLM. Suppose that we observe data $(X_1, Z_1, Y_1), \dots, (X_n, Z_n, Y_n) \in (\R, \R^p, \R)$ from a GLM
$$
\begin{cases}
Y_i | Z_i  \sim \textrm{expfam}(\mu_i) \\ 
g(\mu_i) = \gamma X_i + \beta^T Z_i,
\end{cases}
$$
where $\textrm{expfam}(\mu_i)$ is an exponential family-distributed random variable with mean $\mu_i$, $(\gamma, \beta) \in (\R, \R^p)$ are unknown regression coefficients, and $g : \R \to \R$ is a differentiable link function. Let $\eta_i = \beta^T Z_i$ denote the $i$th linear component of the model under the null hypothesis. Let $d \eta_i / d\mu_i$ denote the derivative of the linear component of the model with respect to the mean, and let $M = \textrm{diag}\{ d \eta_i / d \mu_i\}$ denote the diagonal matrix of these derivatives. Next, let $V(\mu_i)$ denote the variance of $Y_i$ given mean $\mu_i$. Define $$ W_i = \frac{1}{V(\mu_i) (d\eta_i/d\mu_i)^2},$$ and let $W = \textrm{diag}\{W_i\}_{i=1}^n$ denote the $W_i$s assembled into a diagonal matrix. Let $\hat{\beta}_n$ denote the estimator for $\beta$ under the null hypothesis of $\gamma = 0$.  Let $\hat{W}$ and $\hat{M}$ denote $W$ and $M$ evaluated at $\beta = \hat{\beta}_n$, respectively. Finally, let $e = \hat{M}(Y - \hat{\mu})$ denote the vector of ``working residuals.'' The GLM score test statistic for testing $\gamma = 0$ is given by $$z_\textrm{score} = \frac{X^T \hat{W} e}{X^T \hat{W}X - X^T \hat{W} Z (Z^T \hat{W} Z)^{-1} Z^T \hat{W}X}.$$ Defining $\hat{r} = \hat{W}e$, we can express the numerator as $X^T \hat{r}$. 

The GLM fitting procedure (i.e., iteratively reweighted least squares) computes the estimate $\hat{\beta}_n$, the fitted values $\hat{\mu}$, the weight matrix $\hat{W}$, the working residual vector $e$, and a QR decomposition of the matrix $(\hat{W})^{1/2}Z$. Let \texttt{fit} be a fitted GLM object in R. We can access the fitted values, weight matrix, working residual vector, and QR decomposition via \texttt{fit\$fitted.values}, \texttt{fit\$weights}, \texttt{fit\$residuals}, and \texttt{fit\$qr}, respectively. Both Algorithm $R$ and Algorithm $Q$ use these pieces to compute the GLM score test but in different ways. We present each of these algorithms and compare their theoretical efficiency.

\subsection{Algorithm \textit{R}}

Algorithm $R$ involves directly decomposing the information matrix $Z^\top \hat{W}Z$ via the QR decomposition of $\sqrt{W} Z$. Observe that
$$ Z^\top \hat{W} Z  = \left(\sqrt{\hat{W}}Z\right)^\top \sqrt{\hat{W}}Z = (QR)^\top QR = R^\top Q^\top QR = R^\top R.$$ Thus,
$$ X^\top \hat{W} Z (Z^\top \hat{W}Z)^{-1} Z^\top \hat{W}X = X^\top \hat{W}Z R^{-1} (R^{-1})^\top Z^\top \hat{W}X = X^\top D^\top D X = ||DX||^2,$$ where $D = (R^{-1})^\top Z^\top \hat{W}.$ We quickly and stably can compute $R^{-1}$ via backsubstitution, as $R$ is upper triangular. The terms $X^\top \hat{W}X$ and $X^\top \hat{W} e$ likewise are simple to compute. These observations motivate Algorithm R (Algorithm \ref{algo:algo_r}). If $X$ is binary, we can leverage its binary structure to accelerate step 5 of the algorithm via sparse matrix multiplication.

\begin{algorithm}
	\caption{Algorithm $R$ for computing GLM score tests}
	\KwIn{vectors $X_1, \dots, X_B \in \R^n$ to test for inclusion in the fitted model; weights $\hat{w} = [\hat{W}_1, \dots, \hat{W}_n]$; working residual vector $e$; and QR decomposition of the matrix $\sqrt{\hat{W}}Z.$}\label{algo:algo_r}

	Compute the inverse $(R^\top)^{-1}$ of the matrix $R^\top$ via backsubstitution.

	Compute $D = (R^\top)^{-1} Z^\top \hat{W}$
		
	Compute $\hat{r} = \hat{W} e.$
	
	\For{$X \in \{X_1, \dots, X_B\}$}{
		Compute $$\begin{cases}
			\texttt{top} = \hat{r}^\top X\\
			\texttt{bottom\_left} = w^\top (X \odot X) \\
			\texttt{bottom\_right} = D X.
		\end{cases}
		$$
        
		Compute $z_\textrm{curr} = \texttt{top}/\sqrt{\texttt{bottom\_left} - ||\texttt{bottom\_right}||^2}$
	}
\end{algorithm}

\paragraph{Analysis of Algorithm $R$.} We calculate the number of additions (or subtractions) and multiplications (or divisions) that Algorithm $R$ requires. First, consider the precomputation, i.e.\ steps 1--3. Computing the inverse of $R^\top$ via backsubstitution requires $p^3 + p^2/2$ multiplications and $p^3 - p^2/2$ additions. Next, computing $Z^\top \hat{W}$ requires $np$ multiplications. Leveraging the upper triangular structure of $R$, computing $(R^\top)^{-1} Z^\top \hat{W}$ requires $np(p+1)/2$ multiplications and $np(p-1)/2$ additions. Finally, computing $\hat{r} = \hat{W} e$ requires $n$ multiplications. Thus, the total number of additions and multiplications required in the precomputation portion (i.e., steps 1 -- 3) is  $p^3 - p^2/2 + np^2/2 - np/2$ and $p^3 + p^2/2 + np^2/2 + (3/2)np + n$, respectively.

Next, we analyze the iterative part of the algorithm (i.e., steps 4--5). Consider the standard (i.e., sparsity-unaware) version of the algorithm. Computing $\texttt{top} = \hat{r}^\top X$ requires $n$ multiplications and $n-1$ additions. Next, computing $\texttt{bottom\_left} = w^\top (X \odot X)$ requires $2n$ multiplications and $n-1$ additions. Computing $\texttt{bottom\_right} = DX$ requires $np$ multiplications and $(n-1)p$ additions. Next, computing $|| \texttt{bottom\_right}||^2$ requires $p$ multiplications and $p-1$ additions. Finally, computing $z_\textrm{curr}$ requires a single addition and a single multiplication (ignoring the square root for simplicity). Thus, the iterative part of the algorithm requires $2n + np - 2$ additions and $3n + np + p + 1$ multiplications.

Next, consider the sparsity-exploiting version of the algorithm. Suppose that $X$ is binary and contains $s$ nonzero entries, i.e.\ $s = \sum_{i=1}^n X_i$. Computing $\texttt{top} = \sum_{i:X_i = 1}\hat{r}_i$ requries $s-1$ additions. Likewise, computing $\texttt{bottom\_left} = \sum_{i:X_i = 1} \hat{W}_i$ requires $s-1$ additions. Computing \texttt{bottom\_right} requires $p(s-1)$ additions. Finally, following the logic from above, computing the test statistic using these ingredients requires an additional $p - 1$ additions and $p+1$ multiplications. Hence, the iterative part of the algorithm requires $ps + 2s - 2$ additions and $p+1$ multiplications.

\subsection{Algorithm \textit{Q}}
Algorithm $Q$ is based on an algebraic manipulation of the score test statistic. Let $E = X - Z(Z^\top \hat{W} Z)^{-1} Z^\top \hat{W}X$ be the vector of residuals from an OLS regression of $X$ onto $Z$ using weight matrix $\hat{W}$. We can express the denominator of the score test statistic as
\begin{equation}\label{eqn:score_test_denom_rewrite}
X^\top \hat{W} X - X^\top\hat{W}Z (Z^\top \hat{W} Z)^{-1} Z^\top \hat{W}X = E^\top \hat{W} E.
\end{equation}
Additionally, we can express the numerator of the score test statistic as
$$
X^\top \hat{W} e = E^\top \hat{W} e + X^\top \hat{W} Z(Z^\top \hat{W} Z)^{-1} Z^\top \hat{W}e.
$$
Note that $Z^\top W e$ is the score vector of the GLM. Evaluating the score vector at the MLE causes it to vanish, i.e.\ $Z^\top \hat{W} e = 0.$ Hence, 
\begin{equation}\label{eqn:score_test_num_rewrite}
	X^\top \hat{W}e = E^\top\hat{W} e.
\end{equation}
Combining our expressions for denominator (\ref{eqn:score_test_denom_rewrite}) and numerator (\ref{eqn:score_test_num_rewrite}), we see that
$$z_\textrm{score} = \frac{E^\top \hat{W} e}{\sqrt{E^\top \hat{W}E}} = \frac{\left(\sqrt{\hat W}E\right)^\top \left( \sqrt{\hat{W}}e\right)}{\sqrt{\left(\sqrt{\hat W} E \right)^\top \left( \sqrt{\hat W} E \right)}}.$$ The vector $\hat{W}^{1/2}e$ is immediately computable, and we can compute $\hat{W}^{1/2}E$ via the QR decomposition:
$$\sqrt{\hat{W}} E = \sqrt{\hat{W}}X - QR((QR)^\top QR)^{-1} (QR)^\top X = \sqrt{\hat{W}} X - Q Q^\top \sqrt{\hat{W}} X,$$
where we have used the fact that $R(R^\top R)^{-1}R^\top = I$ for a full-rank, square matrix $R$. These observations motivate Algorithm Q (Algorithm \ref{algo:algo_q}). If $X$ is binary, we can leverage sparse matrix multiplication in steps 3 and 4 to accelerate the algorithm slightly. We are not aware of any use of the sparsity-exploiting version of Algorithm $Q$ in the literature, but we analyze it here for completeness.

\begin{algorithm}
	\caption{Algorithm $Q$ for computing GLM score tests.}\label{algo:algo_q}
	\KwIn{vectors $X_1, \dots, X_B \in \R^n$ to test for inclusion in the fitted model; weights $\hat{w} = [\hat{W}_1, \dots, \hat{W}_n]$; working residual vector $e$; and QR decomposition of the matrix $\sqrt{\hat{W}}Z.$}
	Compute $e_{\texttt{sqrt\_w}} = \sqrt{\hat{W}} e$
	
	\For{$X \in \{X_1, \dots, X_B\}$}{
		Compute $X_\texttt{sqrt\_w} = \sqrt{\hat{W}} X$
		
		Compute $y = Q^\top X_\texttt{sqrt\_w}$
		
		Compute $z = Q y$
		
		Compute $E_{\texttt{sqrt\_w}} = X_\texttt{sqrt\_w} - z$
		
		Compute $\texttt{top} = (E_{\texttt{sqrt\_w}})^\top (e_{\texttt{sqrt\_w}})$
		
		Compute $\texttt{bottom} = \sqrt{(E_{\texttt{sqrt\_w}})^\top (E_{\texttt{sqrt\_w}})} $
		
		Compute $z = \texttt{top}/\texttt{bottom}$ 
		
	}
\end{algorithm}

\paragraph{Analysis of Algorithm $Q$.} The precomputation $\hat{W}^{1/2}e$ requires $n$ multiplications. We handle the standard (i.e., sparsity-unaware) and sparsity-exploiting algorithms separately, starting with the former. Computing $X_\texttt{sqrt\_w} = \hat{W}^{1/2} X$ requires $n$ multiplications. Next, computing $y= Q^\top X_\texttt{sqrt\_w}$ requires $pn$ multiplications and $p(n-1)$ additions. Likewise, computing $z = Qy$ requires $pn$ multiplications and $p(n-1)$ additions. Additionally, computing $E_\texttt{sqrt\_w} = X_\texttt{sqrt\_w} - z$ requires $n$ additions. Computing $\texttt{top} = (E_{\texttt{sqrt\_w}})^\top (e_{\texttt{sqrt\_w}})$ requires $n$ multiplications and $n-1$ additions, as does computing $\texttt{bottom} =  (E_{\texttt{sqrt\_w}})^\top (E_{\texttt{sqrt\_w}})$. Finally, computing $z = \texttt{top}/\texttt{bottom}$ involves a single division. Hence, the total number of additions and multiplications required by the iterative portion of the algorithm is $2np + 3n - 2p - 2$ and $2np + 3n + 1$, respectively.

Next, consider the sparsity-exploiting version of Algorithm $Q$. Computing $X_\texttt{sqrt\_w} = \hat{W}X$ requires zero multiplications or additions. Computing $y = Q^\top X_\texttt{sqrt\_w}$ requires $sp$ multiplications and $p(s-1)$ additions. The vector $y \in \R^p$ is (in general) dense. Thus, steps 5-9 are the same as in the sparsity-unaware version of the algorithm. In total, the iterative portion of the algorithm requires $np + ps + 3n - 2p - 2$ additions and $np + sp + 2n + 1$ multiplications.

\begin{table}[]
    \centering
    \begin{tabular}{p{4.0cm}|p{6.0cm}|p{4cm}}
         & \textbf{Algorithm $R$} & \textbf{Algorithm $Q$}  \\
         \hline
        \textbf{Sparsity-unaware} & $2p^3 + np^2 + (1 + 2B)np + (1 + 5B)n + Bp - B$ & $n + 4Bnp + 6Bn - 2Bp - B$ \\
         \hline
        \textbf{Sparsity-exploiting} & $2p^3 + np^2 + np + n + Bp\pi n + Bp + 2B\pi n - B$ & $n + 2Bnp + 2Bp\pi n + 5Bn - 2Bp + B$
    \end{tabular}
    \caption{The total number of floating point operations required of the sparsity-unaware and sparsity-exploiting versions of Algorithm $R$ and Algorithm $Q$ as a function of $n, p, \pi$, and $B$.}\label{tab:algo_qr_tab}
\end{table}

\subsection{Comparing Algorithm $R$ and Algorithm $Q$}

\begin{figure}
	\centering
	\includegraphics[width=0.8\linewidth]{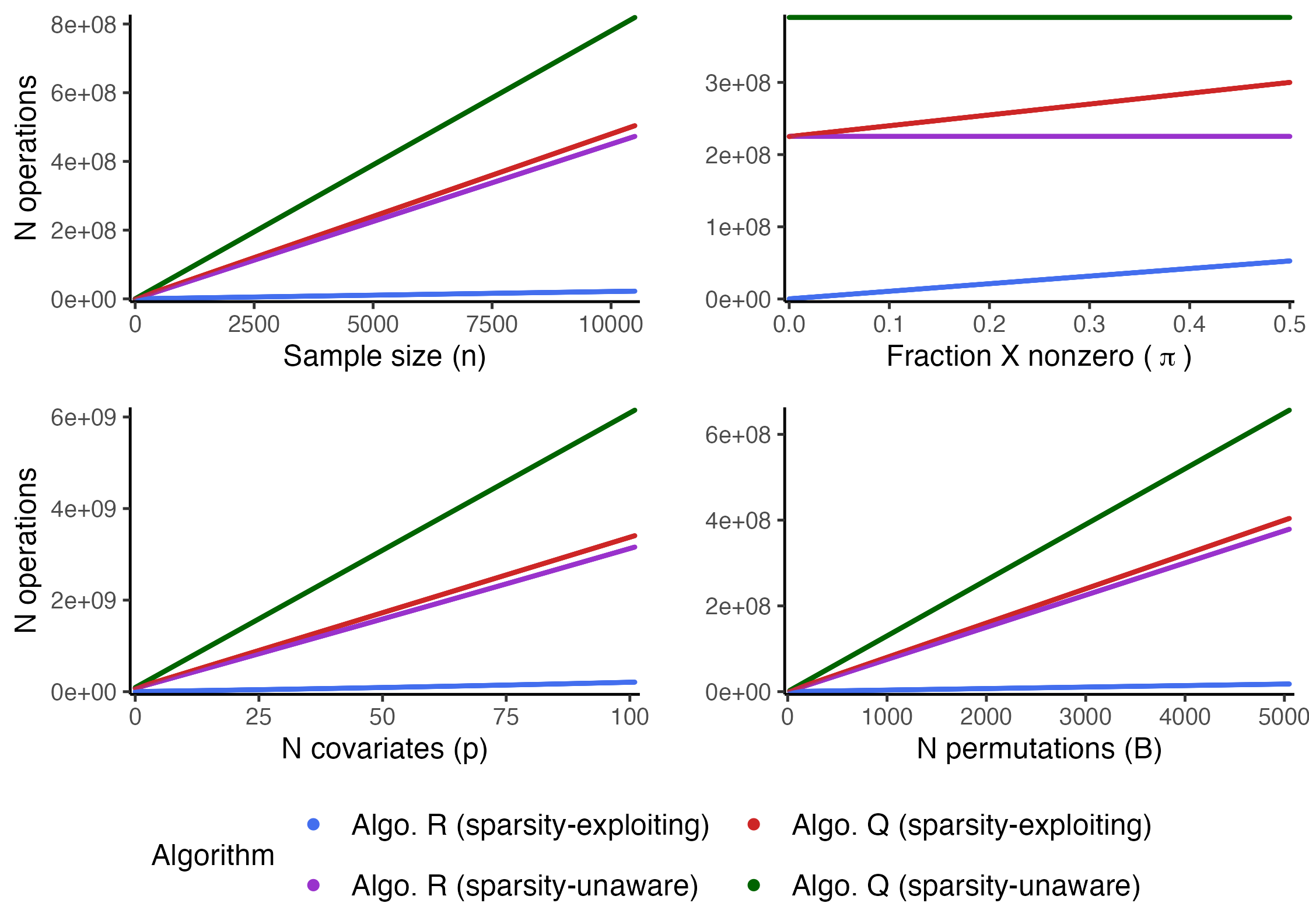}
	\caption{The floating point operation count of the sparsity-unaware and sparsity-exploiting versions of Algorithms $R$ and $Q$ as a function of $n$, $p$, $\pi$, and $B$.}
	\label{fig:lin_alg_sup}
\end{figure}

We can tabulate the total number of floating operations required of the sparsity-unaware and sparsity-exploiting versions of Algorithm $R$ and Algorithm $Q$ using the calculations above. Let $n$ denote the sample size, $p$ the number of covariates, $\pi \equiv s/n$ the fraction of nonzero entries in the treatment vector, and $B$ the number of treatment vectors to test for inclusion in the fitted model. The total floating point operation count of each algorithm as a function $n,p,\pi,$ and $B$ is shown in Table \ref{tab:algo_qr_tab}. To visualize these expressions, we plot (in Figure \ref{fig:lin_alg_sup}) the floating point operation count of each algorithm as a function of $n$ (upper left), $\pi$ (upper right), $p$ (lower left), and $B$ (lower right), while holding fixed the remaining parameters (at $n = 5,000$, $p = 5$, $\pi = 0.1$, $B = 3,000$). We see that (i) the sparsity-exploiting algorithms are more efficient than their sparsity-unaware counterparts, and (ii) both variants of Algorithm $R$ are more efficient than both variants of Algorithm $Q$. Note that Figure \ref{fig:lin_alg_main} (in the main text) compares the sparsity-exploiting version of Algorithm $R$ to the sparsity-unaware version of Algorithm $Q$.

\section{A primer on adaptive permutation testing}\label{sec:adaptive_permutation_testing}

We provide a brief, self-contained overview of the adaptive permutation testing scheme of \citet{Fischer2024b}, which is a crucial to the computational efficiency of the permuted score test (Algorithm \ref{algo:permuting_score_stats}). The intuition behind this procedure is as follows. Let $z_\textrm{orig}$ be the test statistic computed on the original (i.e., unpermuted) data, and let $z_1, z_2, z_3, \dots$ be the ``null'' test statistics computed (sequentially) on the permuted data. The null statistics implicitly give rise to a sequence of $p$-values $p_1, p_2, p_3, \dots$. If the $p$-value $p_t$ is unpromising (e.g., $p_t = 0.6$) after some large number of permutations, we can stop the permutation test early, as it is unlikely that we will reject the null hypothesis, even after computing additional null statistics. On the other hand, we also can stop the permutation test early if the $p$-value $p_t$ drops below the rejection threshold at some time $t \in \N$ . The former (resp., latter) form of stopping is called ``stopping due to futility'' (resp., ``stopping due to rejection'').

Central to the framework of \citet{Fischer2024a} is the notion of an ``anytime-valid $p$-value'' \citep{Ramdas2023}. Intuitively, an anytime-valid $p$-value is a $p$-value that is uniformly valid over time. Formally, an anytime-valid $p$-value $\{ p_t \}_{t=1}^\infty$ is a sequence of $[0,1]$-valued random variables such that, under the null hypothesis and for all $u \in [0,1]$,
\begin{equation}\label{eqn_anytime_valid_p_value}
\P\left(\inf_{t \in \N} p_t \leq u \right) \leq u.
\end{equation}
In other words, the sequence $\{p_t\}_{t=1}^\infty$ never drops below $u$ with probability at least $1 - u$. The key property of an anytime-valid $p$-value $\{ p_t\}_{t=1}^\infty$ is that $p_\tau$ is a valid $p$-value for \textit{any} (possibly data-dependent) stopping time $\tau \in \N$.

We can leverage the notion of an anytime-valid $p$-value to construct a method for adaptive permutation testing. Consider a given hypothesis. Let $$K_t = \sum_{i=1}^t \mathbb{I}(z_i \geq z_\textrm{orig})$$ denote the number of ``losses'' incurred at round $t$, where a ``loss'' is defined as a null statistic equaling or exceeding the original statistic. The number of losses is monotonically increasing in the number of rounds. Let $h \in \N$ be a user-specified tuning parameter (typically, $h$ is set to an integer in the range $[10, 30]$). The parameter $h$ trades off computation and power, with larger values conferring greater power at the cost of increased computation. Let $t^*$ denote the round at which the number of losses $K_t$ hits $h$, i.e.
$$t^* = \inf \{t \in \mathbb{N} :  K_t = h\}.$$ \citet{Fischer2024b} proposed the following (right-tailed) anytime-valid (AV) $p$-value:
$$p^\textrm{AV}_t = \begin{cases}  \frac{h}{t + h - K_t}, \textrm{ if } t < t^* \\ \frac{h}{t^*}, \textrm{ if } t \geq t^*. \end{cases}$$ Notice that $p^\textrm{AV}_t$ is fixed at $h/t^*$ after $h$ losses have accrued. The sequence $\{p_t^\textrm{AV} \}_{t=1}^\infty$ is an anytime-valid $p$-value, i.e.\ it satisfies the condition (\ref{eqn_anytime_valid_p_value}). We can prove this fact in two steps. First, we will show that $p_{t^*} = h/t^*$ is a valid $p$-value, and second, we will show that $p_t^\textrm{AV}$ is lower-bounded by $p_{t^*}$ for all $t \in \N$.

\underline{Step 1: $p_{t^*}$ is a valid $p$-value}. Under the null hypothesis, the rank of $z_\textrm{orig}$ is uniformly distributed among $\{z_\textrm{orig}, z_1, z_2, \dots, z_{t^*}\}$. Thus, for $l \in \{h, h+1, h+2, \dots\}$, the stopping time $t^*$ is greater than or equal to $l$ if and only if $h$ null statistics (among the entire set of $l$ null statistics) exceeds the original statistic. This latter event occurs with probability $l/h$, implying that $\P(t^* \geq l) = h/l$. Hence,
$$\P\left( p_{t^*} \leq h/l \right) \leq \P\left( h/t^* \leq h/l \right) = \P(t^* \geq l) = h/l,$$
i.e.\ $p_{t^*}$ is a valid $p$-value. \citet{Besag1991} introduced the $p$-value $p_{t^*}$ in their early work on sequential permutation testing. The anytime-valid $p$-value $p_t^\textrm{AV}$ of \citet{Fischer2024b} can be seen as an anytime-valid generalization of $p_{t^*}$.

\underline{Step 2: $p_t^\textrm{AV}$ is lower-bounded by $h/t^*$ for all $t \in \mathbb{N}$}. Consider round $t$. The number of losses incurred at round $t$ is $K_t$. Thus, we must compute \textit{at least} $K_t - h$ additional null statistics for the number of losses $K_t$ to hit $h$. This implies that $t^* \geq t + K_t - h.$ Taking the reciprocal of this inequality and multiplying by $h$, we see that
$$\frac{h}{t^*} \leq  \frac{h}{t + K_t -h} = p_t^\textrm{AV} .$$ But $t$ was selected arbitrarily, implying that the above inequality holds for all $t \in \N$, i.e.
$$\frac{h}{t^*} \leq \inf_{t \in \N} \frac{h}{t + K_t -h} = \inf_{t \in \N} p_t^\textrm{AV}.$$ Finally, combining the above result with the validity of the $p$-value $h/t^*$, we conclude
$$ \P\left( \inf_{t \in \N} p_t^\textrm{AV} \leq u \right) \leq \P\left( \frac{h}{t^*} \leq u \right) \leq u,$$ i.e., $\{ p_t^\textrm{AV}\}_{t=1}^\infty$ is an anytime-valid $p$-value.  $\square$

Equipped with a suitable anytime-valid $p$-value $\{p_t^\textrm{AV}\},$ we can design a simple procedure for adaptive permutation testing of multiple hypotheses. Suppose that we seek to test the hypotheses $H_0^1, \dots, H_0^m,$ which correspond to ``datasets'' $X^1, \dots, X^m$. (In the context of differential expression testing, each ``dataset'' is a tuple $\{(X_i, Y_i, Z_i)_{i=1}^n\}$ of treatments, gene expressions, and nuisance covariates.) Our objective is to produce a discovery set that controls the FDR. Define the ``active set'' $\mathcal{A}$ as the set of hypotheses currently under consideration (i.e., the set of hypotheses for which we have not yet stopped due to rejection or futility). We initialize the active set $\mathcal{A}$ to the set of all hypotheses, i.e.\ $\mathcal{A} = \{1, \dots, m\}$. Next, define the ``futility set'' $\mathcal{F}$ and the ``rejection set'' $\mathcal{R}$ as the sets of hypotheses for which we have stopped due to futility and rejection, respectively. We initialize $\mathcal{F}$ and $\mathcal{R}$ to the empty set, i.e. $\mathcal{F} = \mathcal{R} = \emptyset.$ Let $p^i_t$ (resp., $K^i_t$) denote the anytime-valid $p$-value (resp., the number of losses) of the $i$th hypothesis at time $t$. Finally, let $X_i$ denote the dataset corresponding to the $i$th hypothesis, and let $\tilde{X}^i$ signify a randomly permuted version of $X^i$.

We proceed as follows (Algorithm \ref{algo:adaptive_permutation_test}). First, we compute the original statistic $z^i_\textrm{orig} = T(X_i)$ for all hypotheses $i \in \{1, \dots, m\}$. While the active set is nonempty, we iterate through the hypotheses in the active set, computing the null statistic $z^t_i = T(\tilde{X}^i)$. If $z_i^t > z_\textrm{orig}^t$, we increment the number of losses $K_t^i$ by 1. (Otherwise, we keep $K_t^i = K_{t-1}^i$.) If the number of losses $K_t^i$ hits $h$, we move hypothesis $i$ from the active set into the futility set. Next, we update the $p$-value via $p^t_i = h/(t + h - K_t^i)$. (If hypothesis $i$ is not in the active set at time $t$, we simply set the $p$-value $p_t^i$ to its previous value, i.e.\ $p_{t}^i = p_{t-1}^i$.) We then run the BH procedure on $\{p_t^1, \dots, p_t^m\},$ yielding a rejection threshold $R_t$. Finally, we iterate through the hypotheses in the active set once more, moving hypothesis $i$ from the active set into the rejection set if $p_t^i$ falls below the rejection threshold $R_t$. We stop this process when the active set is empty.

Let $\tau_i$ denote the stopping time of the $i$th hypothesis, i.e.\ $\tau_i$ is the round at which we move hypothesis $i$ from the active set into either the futility set or rejection set. The key insight is that $p_{\tau_i}$ is a valid $p$-value. Thus, the discovery set --- which we obtain by applying a BH correction to $p_{\tau_1}, \dots, p_{\tau_m}$ --- controls the FDR by virtue of the BH procedure. One can show that, if the number of losses $K^i_t$ for hypothesis $i$ hits $h$ in round $\tau_i$, then hypothesis $i$ cannot be rejected in round $\tau_i$ or in any subsequent round. Thus, hypothesis $i$ remains in the futility set for all rounds $t \geq \tau_i$. Additionally, because the rejection threshold $R_t$ is monotonically increasing and the anytime-valid $p$-value $p^t_i$ is monotonically decreasing in $t$, if hypothesis $i$ is rejected in round $\tau_i$, then hypothesis $i$ remains in the rejection set for all rounds $t \geq \tau_i$.

\begin{algorithm}
	
	\caption{Adaptive permutation test for controlling FDR.}\label{algo:adaptive_permutation_test}
	
	\textbf{Input}: Datasets $X^1, \dots, X^m$, test statistic $T$, tuning parameter $h$ 
	
	Set $\mathcal{A} \leftarrow \{1, \dots, m\}$, $\mathcal{F} \leftarrow \emptyset,$ $\mathcal{R} \leftarrow \emptyset, t \leftarrow 0, \{K^i \leftarrow 0\}_{i=1}^m$
	
	\For{$i \in \{1, \dots, m\}$}{
		Compute $z_\textrm{orig}^i = T(X^i)$ 
	}
	\While{$\mathcal{A}$ is nonempty}{
		$t \leftarrow t + 1$
		
		\For{$i \in \mathcal{A}$}{
			$z^i \leftarrow T(\tilde{X}^i)$
			
			\If{$z^i > z^i_\textrm{orig}$}{
				$K^i \leftarrow K^i + 1$
				
				\If{$K^i == h$}{
					Move $i$ from $\mathcal{A}$ into $\mathcal{F}$
				}
			}
			
			$p^i \leftarrow h/(t + h - K^i)$
			
		}
		
		$R \leftarrow \textrm{Compute the BH rejection threshold on }\{p^1, p^2, \dots, p^m\}$
		
		\For{$i \in \mathcal{A}$}{
			\If{$p^i \leq R$}{
				Move $i$ from $\mathcal{A}$ into $\mathcal{R}$
			}
		}
	}
	\textbf{Return}: $\mathcal{R}$
\end{algorithm}

\end{document}